\documentclass[a4paper,oneside,fleqn,11pt]{article}
\usepackage{amsmath,amssymb,amsfonts,amsthm,type1cm,bm,color,mathrsfs}
\usepackage{mathtools}
\usepackage{graphicx,psfrag,epsf}
\usepackage{rotating}
\usepackage{authblk}
\usepackage{yfonts}
\usepackage{pdflscape}
\usepackage{setspace}
\usepackage{natbib}
\usepackage{float} 
\usepackage{caption}
\usepackage{subcaption}
\RequirePackage[colorlinks,citecolor=blue,urlcolor=pink]{hyperref}
\usepackage{bigstrut}
\usepackage{comment}
\usepackage{chngcntr}
\usepackage{ragged2e}

\mathtoolsset{showonlyrefs=true}

\DeclareTextFontCommand{\texttt}{\ttfamily\upshape}

\DeclareMathOperator\diag{diag}

\DeclareMathOperator\E{\mathbb{E}}
\DeclareMathOperator\Var{Var}

\DeclareMathOperator\Pro{\mathbb{P}}

\DeclareMathOperator\F{F}

\DeclareMathOperator*{\plim}{plim}

\DeclareMathAlphabet\mathbfcal{OMS}{cmsy}{b}{n}
\def\b0{\mathbf{0}}
\def\b1{\mathbf{1}}

\def\bA{\mathbf{A}}
\def\bB{\mathbf{B}}
\def\bC{\mathbf{C}}
\def\bD{\mathbf{D}}
\def\bE{\mathbf{E}}
\def\bF{\mathbf{F}}
\def\bG{\mathbf{G}}
\def\bH{\mathbf{H}}
\def\bI{\mathbf{I}}

\def\bM{\mathbf{M}}
\def\bN{\mathbf{N}}
\def\bO{\mathbf{O}}
\def\bP{\mathbf{P}}
\def\bQ{\mathbf{Q}}
\def\bR{\mathbf{R}}
\def\bS{\mathbf{S}}
\def\bT{\mathbf{T}}
\def\bU{\mathbf{U}}
\def\bV{\mathbf{V}}
\def\bW{\mathbf{W}}
\def\bX{\mathbf{X}}

\def\bZ{\mathbf{Z}}
\def\ba{\mathbf{a}}
\def\bb{\mathbf{b}}

\def\be{\mathbf{e}}
\def\bff{\mathbf{f}}
\def\bg{\mathbf{g}}
\def\bh{\mathbf{h}}

\def\bw{\mathbf{w}}
\def\bx{\mathbf{x}}
\def\by{\mathbf{y}}
\def\bz{\mathbf{z}}

\def\bbR{\mathbb{R}}

\def\bGamma{\boldsymbol{\Gamma}}
\def\bgamma{\boldsymbol{\gamma}}
\def\bkappa{\boldsymbol{\kappa}}

\def\bbeta{\boldsymbol{\beta}}

\def\bDelta{\boldsymbol{\Delta}}
\def\bdelta{\boldsymbol{\delta}}
\def\bLambda{\boldsymbol{\Lambda}}

\def\bPhi{\boldsymbol{\Phi}}

\def\bPsi{\boldsymbol{\Psi}}

\def\bxi{\boldsymbol{\xi}}
\def\bSigma{\boldsymbol{\Sigma}}
\def\bGamma{\boldsymbol{\Gamma}}

\def\bepsilon{\boldsymbol{\epsilon}}

\def\bmu{\boldsymbol{\mu}}
\def\bzero{\mathbf{0}}
\def\bone{\mathbf{1}}

\def\cR{\boldsymbol{\mathcal{R}}}

\newcommand{\vsp}{\vspace{15pt}}

\newtheorem{thm}{Theorem}%
\newtheorem{lem}{Lemma}%
\newtheorem{ass}{Assumption}%
\newtheorem{assb}{Assumption}%
\renewcommand{\theass}{\arabic{ass}}

\newtheorem{rem}{Remark}%
\newtheorem{prop}{Proposition}%

\newcommand{\CD}{\stackrel{d}{\longrightarrow}}
\newcommand{\CP}{\stackrel{p}{\longrightarrow}}

\newcommand{\CDb}{\stackrel{d^{\dag}}{\longrightarrow}}
\newcommand{\CPb}{\stackrel{p^{\dag}}{\longrightarrow}}

\makeatletter
\def\section{\@startsection {section}{1}{\z@}{-3.5ex plus -1ex minus-.2ex}{2.3ex plus .2ex}{\large\bf}}
\makeatother

\makeatletter
\def\subsection{\@startsection {subsection}{1}{\z@}{-3.5ex plus -1ex minus-.2ex}{2.3ex plus .2ex}{\normalsize\bf}}
\makeatother

\title{\textbf{Bootstrap inference in factor-augmented regressions with weak factors}}

\author[$\,\!$]{\textsc{Peiyun Jiang}$^*$}

\author[$\,\!$]{\textsc{Takashi Yamagata}$^\dagger$}
\affil[$*$]{\textit{Center for Data Science, Waseda University
}}
\affil[$\dagger$]{\textit{Department of Economics and Related Studies, University of York}}
\affil[$\dagger$]{\textit{Graduate School of Economics and Management, Tohoku University}}

\date{October 15, 2026}

\begin{document}

\maketitle

\begin{abstract}
\noindent 
In this paper, we develop bootstrap methods for factor-augmented regressions with potentially weak common components, allowing signal eigenvalues to diverge at heterogeneous rates $N^{\alpha_k}$, where $0<\alpha_k\leq1$, $k=1,\ldots,r$. 
The regression estimator targets a rotation of the latent parameter vector that may be data dependent, while factor estimation can induce asymptotic bias. 
The proposed procedure directly constructs the rotation-specific target in the bootstrap world, avoiding the additional randomness induced by existing bootstrap methods. 
Bootstrap distributional validity is established for two data-dependent targets and a signal-dependent population target, where studentization provides asymptotically pivotal statistics for both analytical and bootstrap bias-corrected estimators. 
The analysis also distinguishes distributional approximation relative to a data-dependent target from inference on fixed population restrictions and introduces a bootstrap diagnostic for rotation uncertainty (RU) when signal eigenvalues are close. 
Monte Carlo experiments show improved bias correction, gains from bootstrap distributional approximation, and demonstrate that inference distortions from small signal eigengaps dissipate quickly with $T$ for fixed targets, in contrast to the persistent distortions under random rotations. 
An application to excess bond return predictability illustrates the practical importance of bootstrap distributional approximation and rotation uncertainty for inference on individual factor coefficients. 
\end{abstract}

\textbf{Keywords.} Factor model, Weak factors, Bootstrap, Bootstrap bias correction, Factor-augmented regression
{\let\thefootnote\relax\footnote{$^*$Corresponding author. Email: peiyun.jiang@aoni.waseda.jp; Address: Bldg 1 Waseda Campus, 1-6-1 Nishiwaseda, Shinjuku-ku, Tokyo 169-8050, Japan}}

\section{Introduction}

Factor-augmented regressions are widely used for forecasting macroeconomic and financial variables using a small number of common factors extracted from a large predictor set.
The $h$-ahead forecast regression is $y_{t+h}={\bgamma^*}'\bff_t^*+\bbeta'\bw_t+\epsilon_{t+h}$, $t=1,\ldots,T$, where $\bff_t^*$ is an $r\times1$ vector of latent predictive factors and $\bw_t$ is a $p\times1$ vector of observable predictors.
Since $\bff_t^*$ is unobserved, it is typically replaced by the principal component (PC) estimator $\hat{\bff}_t$, normalized so that $T^{-1}\sum_{t=1}^T\hat{\bff}_t\hat{\bff}_t'=\bI_r$, obtained from predictors satisfying $x_{t,i}=\bb_i^{*\prime}\bff_t^*+e_{t,i}$, $t=1,\ldots,T$, $i=1,\ldots,N$.

\subsection{Weak factors, rotations, and inferential targets}

Most existing work assumes strong factors, for which the signal eigenvalues diverge proportionally with $N$, whereas we allow the $k$-th signal eigenvalue to diverge at rate $N^{\alpha_k}$, with $1\geq\alpha_1\geq\cdots\geq\alpha_r>0$.
This weak-factor specification accommodates empirically relevant differences in factor strength; see, for example, \citet{BaileyEtAl2016,BaileyEtAl2021,BaiNg2023,jiang2024Mw}.

Let $(\hat{\bgamma}',\hat{\bbeta}')'$ denote the least squares estimator from regressing $y_{t+h}$ on $(\hat{\bff}_t',\bw_t')'$.
Because PC factors are identified only up to rotation, inference on $\hat{\bgamma}$ requires specifying the parameter vector it targets.
Under the conventional approximation $\hat{\bff}_t=\hat{\bH}'\bff_t^*+o_p(1)$, where $\hat{\bH}$ is data dependent and infeasible, $\hat{\bgamma}$ estimates the rotation-specific coefficient $\bgamma_{\hat{\bH}}=\hat{\bH}^{-1}\bgamma^*$.
For strong factors, \citet{BaiNg2006} establish centered inference when $\sqrt{T}/N\to0$, while \citet{Ludvigson2011} derive and correct the bias when $\sqrt{T}/N\to c\in(0,\infty)$; \citet{GoncalvesPerron2014} refine the correction and \citet{gonccalves2020bootstrapping} allow cross-sectionally dependent idiosyncratic errors.
For weak factors, \citet{BaiNg2023} develop theory under a common divergence rate, while \citet{jiang2024Mw} allow heterogeneous rates and derive the corresponding asymptotic distributions and analytical bias corrections.

We also consider two alternative rotations.
First, the data-dependent rotation $\hat{\bH}_q=(T^{-1}\sum_{t=1}^T\hat{\bff}_t\bff_t^{*\prime})^{-1}$ yields the target $\bgamma_{\hat{\bH}_q}=\hat{\bH}_q^{-1}\bgamma^*$, whose estimator has a smaller asymptotic bias under weaker rate conditions than for $\bgamma_{\hat{\bH}}$ \citep{BaiNg2023,jiang2023revisiting,jiang2024Mw}.
Second, \citet{jiang2023revisiting} show that there always exists a unique rotation $\bH$, up to column signs, depending only on the population signals, such that $\bff_t^0:=\bH'\bff_t^*$ satisfies $T^{-1}\sum_{t=1}^T\bff_t^0\bff_t^{0\prime}=\bI_r$.
Writing $\bgamma^0=\bH^{-1}\bgamma^*$ gives the equivalent regression $y_{t+h}={\bgamma^0}'\bff_t^0+\bbeta'\bw_t+\epsilon_{t+h}$, and regression on the estimated PC factors consistently estimates the fixed parameter vector $(\bgamma^{0\prime},\bbeta')'$ up to signs.
Unlike $\bgamma_{\hat{\bH}}$ and $\bgamma_{\hat{\bH}_q}$, $\bgamma^0$ does not depend on the sample realization of the PC estimator and therefore gives individual coefficients and their linear combinations a fixed population interpretation.
We develop bootstrap distributional theory for $\hat{\bgamma}$ relative to all three targets and study inference on fixed population restrictions defined in terms of $\bgamma^0$.

\subsection{Bootstrap distributional approximation and bias correction}

Existing bootstrap procedures exploit the fact that, although $\hat{\bH}$ is infeasible in the original sample, its bootstrap analogue is constructible because bootstrap samples are generated from the PC estimates \citep{GoncalvesPerron2014,gonccalves2020bootstrapping}.
Writing $\hat{\bH}^{\dag}$ for the bootstrap rotation, their procedure uses the pre-rotated statistic $\sqrt{T}(\hat{\bH}^{\dag}\hat{\bgamma}^{\dag}-\hat{\bgamma})=\sqrt{T}\hat{\bH}^{\dag}(\hat{\bgamma}^{\dag}-\bgamma_{\hat{\bH}^{\dag}})$ to approximate the distribution of $\sqrt{T}(\hat{\bgamma}-\bgamma_{\hat{\bH}})$, exploiting $\hat{\bH}^{\dag}\CPb\bI_r$.
Related extensions use the same pre-rotation device \citep{DjogbenouGoncalvesPerron2015,GoncalvesPerronDjogbenou2017}.

Our first contribution is to exploit directly the constructibility of the rotation-specific target in the bootstrap world.
Rather than pre-rotating, we center at the bootstrap target itself and use $\sqrt{T}(\hat{\bgamma}^{\dag}-\bgamma_{\hat{\bH}^{\dag}})$, thereby avoiding the additional finite-sample randomness introduced by the estimated pre-rotation matrix.
We establish first-order bootstrap validity under possibly heterogeneous divergence rates for the conventional target $\bgamma_{\hat{\bH}}$, the alternative target $\bgamma_{\hat{\bH}_q}$, and the fixed target $\bgamma^0$.

Our second contribution concerns bias correction and distributional approximation for bias-corrected estimators.
For the uncorrected estimator, the bootstrap reproduces the possibly noncentral first-order distribution, including its asymptotic bias, and hence remains valid even when a centered normal approximation is not.
We consider both the analytical bias corrections of \citet{jiang2024Mw} and a bootstrap correction that averages rotation-specific bootstrap estimation errors and subtracts the resulting bias estimate from the original estimator.
The latter avoids direct estimation of the analytical bias formula, its eigengap quantities, and the divergence rates $\alpha_k$.

For both bias-corrected estimators, we bootstrap the corresponding studentized statistic, a distributional step not previously considered in the factor-augmented regression bootstrap literature.
After first-order bias correction, these statistics are asymptotically pivotal, so conventional normal or chi-squared critical values are first-order valid for inference relative to the corresponding target, including fixed population restrictions when the target is $\bdelta^0$.
The bootstrap consistently approximates their distributions and can therefore provide refinement beyond first-order asymptotic inference.

\subsection{Fixed-target inference and rotation uncertainty}

Distinct population signal eigenvalues are a standard identification condition for PC factors, but in finite samples close eigenvalues can leave the individual factor directions imprecisely determined even when the factor space itself is well estimated.
This matters for inference on the fixed target. Consider a decomposition $\sqrt{T}(\hat\bgamma -\bgamma^0) = \sqrt{T}(\hat\bgamma -\bgamma_{\hat{\bH}_q}) + \sqrt{T}(\hat{\bH}_q^{-1}-\bH^{-1})\bgamma^*$. 
The second term is essentially a rotation gap, which depends on the inverse signal eigengaps and enters the analytical bias correction for the fixed target. Consequently, small eigengaps can induce substantial finite-sample rotation uncertainty while rendering the analytical bias correction unstable. 

To assess the importance of this uncertainty for the $j$-th fixed population coefficient $\gamma_j^0$, we introduce the bootstrap rotation-uncertainty (RU) diagnostic $\widehat{\mathrm{RU}}_j$, which compares bootstrap variation in the rotation gap with the sampling variation of the corresponding coefficient estimator.
Under the maintained eigengap and moment conditions, $\widehat{\mathrm{RU}}_j=o_p(1)$, but finite-sample values can reveal substantial sensitivity to estimated PC directions. The diagnostic is informative rather than a pretest; formal analysis beyond the maintained eigengap conditions is pursued separately. 

\subsection{Finite-sample and empirical evidence}

Monte Carlo experiments under heterogeneous divergence rates show that the proposed direct bootstrap correction generally has smaller absolute bias than the existing pre-rotation correction for the conventional target, with essentially unchanged RMSE, and that bootstrap distributional approximation improves target-relative inference, particularly in smaller samples. 
Experiments with progressively smaller eigengaps show that accurate inference relative to a data-dependent target does not guarantee reliable small-sample inference on an individual fixed population coefficient. While tight eigengaps can temporarily elevate fixed-null size distortions alongside $\widehat{\mathrm{RU}}_j$, this distortion dissipates rapidly as $T$ increases under the fixed population target.%

In predictive regressions for excess bond returns, dependence-robust bootstrap inference and $\widehat{\mathrm{RU}}_j$ distinguish two sources of finite-sample sensitivity.
In the pre-2004 subsample, a small eigengap and large $\widehat{\mathrm{RU}}_4$ indicate substantial rotation uncertainty for the fourth factor coefficient, whereas in the post-2004 subsample the second factor is significant under the normal approximation but only marginally significant under bootstrap inference despite little rotation uncertainty.
Thus, the bootstrap procedures and the diagnostic distinguish findings sensitive to rotation uncertainty from those sensitive to finite-sample distributional approximation.

The rest of the paper is organized as follows. 
Sections~\ref{sec:2}--\ref{sec:fixed-inference} develop the model, bootstrap procedures, distributional theory, bias correction, and fixed-parameter inference, Sections~\ref{sec: MC} and \ref{sec:empirics} present the Monte Carlo and empirical evidence, and Section~\ref{sec:con} concludes; proofs and additional results are provided in the Appendix.

\vsp

\noindent\textbf{Notation.} For a vector $\ba$, $\|\ba\|_2$ denotes the Euclidean norm, while for a matrix $\bA$, $\|\bA\|_{\mathrm{F}}$ and $\|\bA\|_2$ denote the Frobenius and spectral norms and $[\bA]_{jk}$ its $(j,k)$-th element.
We denote the identity matrix of order $s$ by $\bI_s$, $s\times1$ vectors of ones and zeros by $\bone_s$ and $\bzero_s$, and the Hadamard product by $\odot$.
For positive sequences $a_{NT}$ and $b_{NT}$, $a_{NT}\asymp b_{NT}$ means that their ratio is bounded away from zero and infinity; all asymptotic results are understood as $N,T\to\infty$ unless stated otherwise, and $M$ denotes a positive constant independent of $N$ and $T$.

\section{Factor-augmented regression}\label{sec:2}

The factor-augmented regression can be written in matrix form as:
\begin{align}\label{augmodel_mat}
    \by={\bF}^*\bgamma^* + \bW \bbeta + \bepsilon 
    =\bZ^* \bdelta^* + \bepsilon,
\end{align}
where $\by=(y_{1+h},\dots,y_{T+h})'$, $\bepsilon=(\epsilon_{1+h},\dots,\epsilon_{T+h})'$, 
$\bF^{*} = (\bff_1^{*},\dots,\bff_T^{*})'$, $\bW=(\bw_1,\cdots, \bw_T)'$, 
$\bZ^* = (\bF^*,\bW)$ and $\bdelta^* = (\bgamma^{*\prime},\bbeta')'$. 
The latent factor model for the $T\times N$ matrix of predictors is given by
\begin{align}\label{factormodel_mat}
    \bX={\bF}^*\bB^{* \prime} + \bE,
\end{align}
where $\bX = (x_{t,i})$, $\bB^* = (\bb_1^* ,\dots,\bb_N^*)'$ and $\bE = (e_{t,i})$. 
Let $(\lambda_1>\cdots>\lambda_r)$ denote the $r$ largest eigenvalues of the signal component of the model \eqref{factormodel_mat}, namely $T^{-1}\bF^*\bB^{\ast\prime}\bB^{\ast}\bF^{\ast\prime}$, 
and define $\bLambda:=\diag(\lambda_1,\dots,\lambda_r)$. 
We allow the $r$ signal eigenvalues to diverge at possibly different rates, specifically 
$\lambda_k\asymp N^{\alpha_k}$ with $0<\alpha_k\leq1$ for $k=1,\dots,r$. 
We refer to model \eqref{factormodel_mat} with $\alpha_r=1$ as a strong factor (SF) model, and the more general model without this restriction as a weak factor (WF) model. 

\citet{jiang2023revisiting} show that there always exists a unique (up to sign) rotation matrix $\bH := \bP \bV^{-1/2},$ where $\bP$ is the eigenvector matrix of ${\bB^*}'\bB^*(T^{-1}{\bF^*}'\bF^*)$ corresponding to $(\lambda_1,\dots,\lambda_r)$ and $\bV=\bP'(T^{-1}{\bF^*}'\bF^*)\bP$, such that $\bF^0 := \bF^* \bH$, $\bB^0 := \bB^* \bH^{-1\prime},$ which by construction satisfy the $r^2$ restrictions $T^{-1}\bF^{0\prime}\bF^0 = \bI_r$ and $\bB^{0\prime}\bB^0 = \bLambda$. Therefore, $\bH$ is a pure function of signals $(\bF^* , \bB^*)$. It is straightforward to show that
\begin{align}
\label{HalaHhat}
{\bH} = {\bB^*}'\bB^*(T^{-1}{\bF^*}'{\bF^0}) {\bLambda}^{-1}
=(T^{-1}{\bF^0}'{\bF^*})^{-1}.
\end{align}

With this rotation, we can equivalently express models \eqref{augmodel_mat} and \eqref{factormodel_mat} in terms of $\bF^0$, $\bgamma^0 = \bH^{-1}\bgamma^*$, and $\bB^0$, which define the pseudo-true models:
\begin{equation}\label{pseudomodel}
\by={\bF}^0\bgamma^0+\bW\bbeta+\bepsilon=\bZ^0\bdelta^0+\bepsilon,\qquad \bX=\bF^0\bB^{0\prime}+\bE,
\end{equation}
where $\bZ^0=(\bF^0,\bW)$ and
$\bdelta^0=(\bgamma^{0\prime},\bbeta')'$.
Let $\hat{\bF}$ be $\sqrt{T}$ times the leading $r$ orthonormal eigenvectors of $T^{-1}\bX\bX'$, and set $\hat{\bB}:=T^{-1}\bX'\hat{\bF}$. Then $T^{-1}\hat{\bF}'\hat{\bF}=\bI_r$ and $\hat{\bB}'\hat{\bB}=\hat{\bLambda}:=\diag(\hat\lambda_1,\dots,\hat\lambda_r)$, with eigenvalues in descending order. Then, regressing $\by$ on $\hat{\bZ}=(\hat{\bF},\bW)$ yields the least squares estimator $\hat{\bdelta}=(\hat\bZ'\hat\bZ)^{-1}\hat\bZ'\by.$ Under this population representation, $\hat{\bF}$ is consistent for $\bF^0$ up to column signs, and $\hat{\bdelta}$ is consistent for the fixed population parameter $\bdelta^0$.
For distributional analysis, however, it is also useful to center $\hat{\bdelta}$ at data-dependent rotation-specific targets that converge to $\bdelta^0$.  

\citet{BaiNg2002,BaiNg2006} and \citet{StockWatson2002JASA} consider the approximation
\begin{align}\label{FHhat}
\hat{\bF}=\bF^* \hat{\bH} + o_p(1), 
\end{align}
where 
$\hat{\bH} = {\bB^*}'\bB^*(T^{-1}{\bF^*}'\hat{\bF}) \hat{\bLambda}^{-1}$. 
Define $\tilde{\bH}:=\bH^{-1}\hat{\bH}={\bB^0}'\bB^0(T^{-1}{\bF^0}'\hat{\bF})\hat{\bLambda}^{-1}$.
Since $\hat{\bH}$ estimates $\bH$, $\tilde{\bH}$ can be viewed as an estimator of the identity matrix $\bI_r$.
Then the following key identities hold:
\begin{align}\label{keyidentity}
\bF^*\hat{\bH}=\bF^0\tilde{\bH}, \qquad \hat{\bH}^{-1}\bgamma^*=\tilde{\bH}^{-1}\bgamma^0.
\end{align}
Thus, $\hat{\bF}=\bF^0\tilde{\bH}+o_p(1)$, with associated target $\bdelta_{\hat{\bH}}=(\bgamma_{\hat{\bH}}',\bbeta')'$ and $\bgamma_{\hat{\bH}}=\tilde{\bH}^{-1}\bgamma^0=\hat{\bH}^{-1}\bgamma^*$. For the alternative rotation $\hat{\bH}_q:=(T^{-1}\hat{\bF}'\bF^*)^{-1}$, define $\tilde{\bH}_q:=\bH^{-1}\hat{\bH}_q=(T^{-1}\hat{\bF}'\bF^0)^{-1}$. The identities in \eqref{keyidentity} hold with $(\hat{\bH},\tilde{\bH})$ replaced by $(\hat{\bH}_q,\tilde{\bH}_q)$. 

To unify notation, write $\bPhi_{\bR}:=\diag(\bR,\bI_p)$ for any nonsingular $r\times r$ matrix $\bR$.
Let $\cR:=\{\hat{\bH},\hat{\bH}_q,\bH\}$ and, for $\bR\in\cR$, define $\bdelta_{\bR}:=\bPhi_{\bR}^{-1}\bdelta^*$. In particular, $\bdelta_{\bH}=\bdelta^0$, while for $\bR\in\{\hat{\bH},\hat{\bH}_q\}$, $\bdelta_{\bR}=\bPhi_{\bH^{-1}\bR}^{-1}\bdelta^0$. These are the three targets whose asymptotic distributions and biases are derived by \citet{jiang2024Mw} and summarized in Section~\ref{sec:4}.

For studentized inference on these targets, define an estimator of the asymptotic covariance matrix of $\sqrt{T}\hat{\bdelta}$, denoted by $\hat{\bSigma}_{\bdelta}$, as
\begin{equation}\label{eq:Sigma-delta-hat}
\hat{\bSigma}_{\bdelta}:=\hat{\bSigma}_{\bZ^0\bZ^0}^{-1}\hat{\bSigma}_{\bZ^0\bepsilon}\hat{\bSigma}_{\bZ^0\bZ^0}^{-1},
\end{equation}
where $\hat{\bSigma}_{\bZ^0\bZ^0}:=T^{-1}\hat{\bZ}'\hat{\bZ}$ and $\hat{\bSigma}_{\bZ^0\bepsilon}$ is a function of $\hat{\bZ}$ and $\hat{\bepsilon}$ that estimates $\bSigma_{\bZ^0\bepsilon}$, the asymptotic covariance matrix of $T^{-1/2}\bZ^{0\prime}\bepsilon$; its consistency is assumed in Assumption~\ref{ass:Aug_errors}(iv) in Appendix~\ref{sec:bootstrap-proofs}. 
To make this dependence explicit when needed, we write $\hat{\bSigma}_{\bdelta}(\hat{\bZ},\hat{\bepsilon})$.

The bootstrap uses $(\hat{\bF},\hat{\bB})$ as the signal population $(\bF^{0\dag},\bB^{0\dag})$, preserving the pseudo-true normalization. Section~\ref{sec:3} constructs the corresponding bootstrap targets directly, without pre-rotation.

\section{New bootstrap procedure}\label{sec:3}

We first describe the bootstrap procedure for approximating the rotation-specific distribution associated with the conventional target $\bdelta_{\hat{\bH}}$. 
Variables generated under the bootstrap law are denoted by the superscript ``$\dag$'', and the superscript $(b)$ refers to the $b$-th bootstrap sample, $b=1,\dots,B$. 

\begin{enumerate}
\item Generate the bootstrap data $\bX^{\dag(b)}$ using a resampled error matrix $\bE^{\dag(b)}$: $\bX^{\dag(b)}=\bF^{0\dag}\bB^{0\dag\prime} + \bE^{\dag(b)}$ where $\bF^{0\dag}:=\hat{\bF}$ and $\bB^{0\dag}:=\hat{\bB}$. 
Using $\bX^{\dag(b)}$, obtain the bootstrap PC estimators $\hat{\bF}
^{\dag(b)}$ and $\hat{\bB}^{\dag(b)}$, correcting their signs if necessary so that all sample correlations $cor(\hat{f}_{k,t}^{\dag(b)},f_{k,t}^{0\dag})$, $k=1,2,\dots,r$, are positive. 
Construct the bootstrap rotation matrix 
$
 \tilde{\bH}^{\dag(b)}
=\bB^{0\dag\prime}\bB^{0\dag}(T^{-1}\bF^{0\dag\prime}{\hat{\bF}}^{\dag(b)})(\hat{\bB}^{\dag(b)\prime}\hat{\bB}^{\dag(b)})^{-1}.   
$

\item Generate the bootstrap data $\by^{\dag(b)}$ using a resampled error vector $\bepsilon^{\dag(b)}$: $\by^{\dag(b)}=\bF^{0\dag}{\bgamma^{0\dag}} + \bW\bbeta^{\dag} +\bepsilon^{\dag(b)} =\bZ^{0\dag}\bdelta^{0\dag} + \bepsilon^{\dag(b)}$ where ${\bgamma^{0\dag}}:=\hat{\bgamma}$, $\bbeta^{\dag}:=\hat{\bbeta}$, $\bZ^{0\dag} := (\bF^{0\dag},\bW)$ and $\bdelta^{0\dag}:=(\bgamma^{0\dag\prime},\bbeta^{\dag\prime})'$.
Using $\hat{\bZ}^{\dag(b)} := (\hat{\bF}^{\dag(b)},\bW)$, compute the bootstrap estimator $\hat{\bdelta}^{\dag(b)}:=(\hat{\bZ}^{\dag(b)\prime}\hat{\bZ}^{\dag(b)})^{-1}\hat{\bZ}^{\dag(b)\prime}\by^{\dag(b)}$ and the corresponding bootstrap target $\bdelta_{\hat{\bH}^{\dag(b)}}=(\bgamma^{0\dag\prime}\tilde{\bH}^{\dag(b)\prime -1},\bbeta^{\dag\prime})'=\bPhi_{\tilde{\bH}^{\dag(b)}}^{-1} \bdelta^{0\dag}$ to obtain
\begin{align}\label{btsdist}
    (\hat{\bSigma}_{\bdelta}^{\dag(b)})^{-1/2}\sqrt{T}(\hat{\bdelta}^{\dag(b)} - \bdelta_{\hat{\bH}^{\dag(b)}}),
\end{align}
where $\hat{\bSigma}_{\bdelta}^{\dag(b)}:=\hat{\bSigma}_{\bdelta}(\hat{\bZ}^{\dag(b)},\hat{\bepsilon}^{\dag(b)})$ and $\hat{\bepsilon}^{\dag(b)}:=\by^{\dag(b)}-\hat{\bZ}^{\dag(b)}\hat{\bdelta}^{\dag(b)}$.
\item Repeat Steps 1-2 for $b=1,2,\dots,B$ to construct the bootstrap distribution of \eqref{btsdist}.

\end{enumerate}
The justification for using the bootstrap statistic \eqref{btsdist} relies on two features made explicit by the pseudo-true representation: (i) the PC estimators recover the uniquely normalized signal parameters $(\bF^0,\bB^0)$ up to sign, and (ii) the key identity $\tilde{\bH}^{-1}\bgamma^0=\hat{\bH}^{-1}\bgamma^*$ links the rotation-specific target in the original sample to its directly constructible bootstrap analogue.

\begin{rem}\label{rem:1}
Given the bootstrap errors $(\bE^{\dag},\bepsilon^{\dag})$, our procedure differs from those of \citet{GoncalvesPerron2014,gonccalves2020bootstrapping} in the construction of the centered bootstrap statistic.
Rather than using $\sqrt{T}(\hat{\bdelta}^{\dag}-\bdelta_{\hat{\bH}^{\dag}})$ directly, \citet[Corollary 3.1]{GoncalvesPerron2014} use the pre-rotated statistic $\sqrt{T}(\bPhi_{\tilde{\bH}^{\dag}}\hat{\bdelta}^{\dag}-\bdelta^{0\dag})=\sqrt{T}\bPhi_{\tilde{\bH}^{\dag}}(\hat{\bdelta}^{\dag}-\bdelta_{\hat{\bH}^{\dag}})$.
Since $\bPhi_{\tilde{\bH}^{\dag}}\CPb\bI_{r+p}$ in probability, this pre-rotation does not affect first-order validity, but it introduces additional finite-sample randomness through the estimated matrix $\bPhi_{\tilde{\bH}^{\dag}}$.
Our procedure avoids this additional source of variation by centering directly at the constructible rotation-specific bootstrap target.
\end{rem}

The same resampling scheme applies to the other targets. For $\hat{\bH}_q$, replace the relative rotation by $\tilde{\bH}_q^{\dag(b)}=(T^{-1}\hat{\bF}^{\dag(b)\prime}\bF^{0\dag})^{-1}$ and center at $\bPhi_{\tilde{\bH}_q^{\dag(b)}}^{-1}\bdelta^{0\dag}$. For the fixed target $\bdelta^0$, center at $\bdelta^{0\dag}=\hat{\bdelta}$ without constructing a rotation. Resampling of $(\bE^{\dag},\bepsilon^{\dag})$ should reflect the error dependence: Section~\ref{sec: MC} uses a wild bootstrap, whereas Section~\ref{sec:empirics} uses dependent resampling, further examined in Section~\ref{app:depboot}. Section~\ref{sec:4} establishes validity for all three targets.

\section{Bootstrap validity and analytical bias correction}\label{sec:4}
In this section, we establish the asymptotic validity of the proposed bootstrap procedure in approximating the distribution of the estimator $\hat{\bdelta}$ relative to the rotated parameter vectors under different rotation matrices.

\subsection{Assumptions}
We retain the original-sample identification and signal-strength conditions below. The remaining original-sample conditions, Assumptions~\ref{ass:errors}--\ref{ass:Aug_errors}, and their discussion are collected in the Appendix. These conditions are adapted from \citet{jiang2024Mw}.

\begin{ass}\normalfont\label{ass:eigen}
The smallest eigenvalues of ${\bB^*}'\bB^*$ and $T^{-1}{\bF^*}'\bF^*$ are bounded away from zero.
\end{ass}

\begin{ass}[Signal strength]\normalfont\label{ass:signal}
There exist random or non-random variables $d_1,\dots,d_r>0$ and constants $0<\alpha_r\leq \dots \leq \alpha_1 \leq 1$ such that {$\lambda_k=d_kN^{\alpha_k}$} for $k=1,\dots, r$ with ordered $0<\lambda_r< \dots < \lambda_1$ for large $N$. If $d_k$'s are random, we have $\E[d_k^2] \le M$ for all $k$ and $d_k\ge 1/M $ almost surely hold for some constant $M<\infty$. 
In addition, $d_1,\dots,d_r$ are well separated, in the sense that 
$\max_{k\ne\ell:\,\alpha_k=\alpha_\ell}\left|\frac{d_k}{d_k-d_\ell}\right|=O_p(1)$, 
so that the signal eigenvalues with the same $\alpha_k$ remain asymptotically distinct. 
\end{ass}

The signal strength $\lambda_k=d_kN^{\alpha_k}$ comprises the divergence rate $\alpha_k$ and scale coefficient $d_k$; equal rates still require separated scales under Assumption~\ref{ass:signal}. Write $\bN:=\diag(N^{\alpha_1},\dots,N^{\alpha_r})$ and $\bD:=\diag(d_1,\dots,d_r)$, so that $\bLambda=\bD\bN$. No additional normalization of $(\bF^*,\bB^*)$ is imposed.

We now state the assumptions required for the bootstrap analysis, denoted by superscripts ``$\dag$'' in the assumption numbers.

\begin{rem}\label{rem:bootstrap-notation}
    The operators $\Pro^{\dag}$, $\E^{\dag}$ and $\Var^{\dag}$ are conditional on the original sample. Conditional stochastic orders and convergence are defined in Remark~\ref{rem:bootstrap-orders} in Appendix~\ref{sec:bootstrap-proofs}.
\end{rem}

\setcounter{assb}{2}
\begin{assb}[Idiosyncratic errors]
\normalfont\label{ass:errors_b}${ }^{ }$ 
{(i) $\E^{\dag}[e_{t,i}^{\dag}]=0$ and $\E^{\dag}[e_{t,i}^{\dag4}]=O_p(1)$ for all $i$ and $t$;} 
(ii) For all $i$, $|\E^{\dag}[e_{s,i}^{\dag} e_{t,i}^{\dag}]| \leq|\gamma_{s, t}^{\dag}|$ for some $\gamma_{s, t}^{\dag}$ such that $\sum_{t=1}^T|\gamma_{ s, t}^{\dag}| =O_p(1)$; 
(iii)  For all $t$, $|\E^{\dag}[e_{t,i}^{\dag} e_{t,j}^{\dag}]| \leq|\tau_{i, j}^{\dag}|$ for some $\tau_{i, j}^{\dag}$ such that $\sum_{j=1}^N|\tau_{ i,j}^{\dag}| =O_p(1)$; 
(iv) $||\bE^{\dag} ||_2^2=O_{p^{\dag}}(\max{\{N,T}\})$, in probability.
\end{assb}

\begin{assb}[Factors and Loadings]\normalfont\label{ass:factor and loadings_b}
${ }^{ }$ 
(i) $\E^{\dag}\| \bN^{-\frac{1}{2}}\sum_{i=1}^{N} \hat{\bb}_ie_{t,i}^{\dag}   \|_2^2 =O_p(1)$ for each $t$;  
(ii) $\E^{\dag}\| T^{-\frac{1}{2}}\sum_{t=1}^{T}\hat{\bz}_t e_{t,i}^{\dag}\|_2^2 =O_p(1)$ for each $i$;  
(iii) $\E^{\dag}\| T^{-\frac{1}{2}}\bN^{-\frac{1}{2}} \sum_{t=1}^{T}\sum_{i=1}^{N} \hat{\bb}_i e_{t,i}^{\dag}\hat{\bz}_t^{\prime} \|_{2}^2 =O_p(1)$;  
(iv) %
$T^{-1} \sum_{t=1}^{T}  ( \bN^{-\frac{1}{2}}\sum_{i=1}^{N} \hat{\bb}_i e_{t,i}^{\dag})  ( \bN^{-\frac{1}{2}}\sum_{i=1}^{N} \hat{\bb}_i e_{t,i}^{\dag})'- \bGamma^{\dag} =o_{p^{\dag}}(1) $, in probability, where $\bGamma^{\dag} :=   T^{-1} \sum_{t=1}^T  \Var^{\dag}(\bN^{-\frac{1}{2}}\sum_{i=1}^{N} \hat{\bb}_i e_{t,i}^{\dag} )$ is positive definite almost surely.
\end{assb}

\begin{assb}[Weak dependence between idiosyncratic errors and regression errors]\label{ass:2errors_b} 
$\E^{\dag} ||T^{-\frac{1}{2}}\bN^{-\frac{1}{2}} \sum_{t=1}^{T} \sum_{i=1}^N  \hat{\bb}_i e_{t,i}^{\dag} \epsilon_{t+h}^{\dag}||_2^2 =O_p(1)$.
\end{assb}

\begin{assb}[Moments and CLT for the Score Vector]\emergencystretch=3em\normalfont\label{ass:Aug_errors_b}
${ }^{ }$ 
(i) $\E^{\dag}[\epsilon_{t+h}^\dag]=0$ and $T^{-1}\sum_{t=1}^T\E^{\dag}|\epsilon_{t+h}^{\dag}|^2=O_p(1)$; 
(ii) $(\bSigma_{\bZ^0\bepsilon}^{\dag})^{-1/2}T^{-1/2}\hat{\bZ}'\bepsilon^{\dag}\allowbreak\CDb N(\mathbf{0},\bI_{r+p})$, in probability, where $\E^{\dag}\|T^{-1/2}\hat{\bZ}'\bepsilon^{\dag}\|_2^2=O_p(1)$ and 
$\bSigma_{\bZ^0\bepsilon}^{\dag}:=\Var^{\dag}(T^{-1/2}\hat{\bZ}'\bepsilon^{\dag})$ 
is positive definite almost surely.
\end{assb}

\begin{assb}\normalfont\label{ass:bootstrap_b}
The following bootstrap consistency conditions hold: 
 $\bSigma_{\bZ^0\bepsilon}^{\dag}\CP\bSigma_{\bZ^0\bepsilon}$ and $\bGamma^\dag\CP\bGamma$. 

\end{assb}

Assumptions~\ref{ass:errors_b}--\ref{ass:Aug_errors_b} are the bootstrap analogues of Assumptions~\ref{ass:errors}--\ref{ass:Aug_errors}. 
Assumption~\ref{ass:bootstrap_b} requires consistency of covariance quantities and parallels Conditions E* and F* of \citet{GoncalvesPerron2014}.

Given these assumptions, we now present our main theoretical results.

\subsection{Asymptotic distributions and bootstrap validity }

We use $\bdelta_{\bR}^{\dag}$ to denote the corresponding bootstrap target, so that $\bdelta_{\hat{\bH}}^{\dag}:=\bdelta_{\hat{\bH}^{\dag}}$, $\bdelta_{\hat{\bH}_q}^{\dag}:=\bdelta_{\hat{\bH}_q^{\dag}}$, and $\bdelta_{\bH}^{\dag}:=\bdelta^{0\dag}$.
For $\bR\in\cR$, let $\bmu_{\bR}$ denote the first-order asymptotic bias of $\sqrt{T}(\hat{\bdelta}-\bdelta_{\bR})$, with
\begin{align}
\bmu_{\hat{\bH}}
&:=
c_1\bkappa_{\bdelta^*},
&
\bmu_{\hat{\bH}_q}
&:=
c_2\bar{\bkappa}_{\bdelta^*},
&
\bmu_{\bH}
&:=
c_2\bar{\bkappa}_{\bdelta^*}
+
c_2\bh_{\bgamma^*},
\label{eq:muR}
\end{align}
where $c_1 := \lim \sqrt{T}N^{\frac12\alpha_1-\frac32\alpha_r}$ and $c_2 := \lim\sqrt{T}N^{-\alpha_r}$, and the explicit expressions for $\bkappa_{\bdelta^*}$, $\bar{\bkappa}_{\bdelta^*}$, and $\bh_{\bgamma^*}$ are given in Appendix~\ref{app:rotation-specific-bootstrap}.

The following theorem collects the non-bootstrap asymptotic distributional results of \citet{jiang2024Mw} for the three targets in this unified notation.

\begin{thm}[Asymptotic distributions for the three targets]
\label{thm:asymptotic-all-rotations}
Suppose Assumptions \ref{ass:eigen}--\ref{ass:Aug_errors} hold, with $\alpha_r>1/2$ and $N^{1-\alpha_r}/\sqrt{T}\to0$.
For $\bR=\hat{\bH}$, suppose additionally that $\sqrt{T}N^{\frac12\alpha_1-\frac32\alpha_r}\to c_1\in[0,\infty)$, while for $\bR\in\{\hat{\bH}_q,\bH\}$ suppose that $\sqrt{T}N^{-\alpha_r}\to c_2\in[0,\infty)$.
Then, for each $\bR\in\cR$, $\sqrt{T}(\hat{\bdelta}-\bdelta_{\bR})\CD N(\bmu_{\bR},\bSigma_{\bdelta})$.
\end{thm}

The next theorem gives the corresponding bootstrap result and is the main first-level bootstrap validity result of this paper.

\begin{thm}[Bootstrap validity for the three targets]
\label{thm:bootstrap-all-rotations}
Suppose the conditions of Theorem~\ref{thm:asymptotic-all-rotations} hold together with Assumptions \ref{ass:errors_b}--\ref{ass:bootstrap_b}.
Then, for each $\bR\in\cR$, $\sqrt{T}\left(\hat{\bdelta}^{\dag}-\bdelta_{\bR}^{\dag}\right)\CDb N\left(\bmu_{\bR},\bSigma_{\bdelta}\right)$ in probability, with the same $\bmu_{\bR}$ and $\bSigma_{\bdelta}$ as in Theorem~\ref{thm:asymptotic-all-rotations}, and consequently
\begin{align}
\sup_{\bx\in\bbR^{r+p}}
\left|
\Pro^{\dag}\!\left[
\sqrt{T}\left(\hat{\bdelta}^{\dag}-\bdelta_{\bR}^{\dag}\right)
\leq\bx
\right]
-
\Pro\!\left[
\sqrt{T}\left(\hat{\bdelta}-\bdelta_{\bR}\right)
\leq\bx
\right]
\right|
=
o_p(1).
\label{eq:bootstrap-valid-all-R}
\end{align}
\end{thm}

The rotation-specific statements underlying Theorem~\ref{thm:bootstrap-all-rotations}, including the explicit bias expressions and the rotation-gap expansion for the fixed target, are given in Appendix~\ref{app:rotation-specific-bootstrap}.
An important implication of Theorem~\ref{thm:bootstrap-all-rotations} is that the bootstrap reproduces the possibly noncentral first-order distribution without requiring knowledge of the divergence rates $(\alpha_1,\ldots,\alpha_r)$.

\begin{rem}[Finite-sample and probability-limit rotations]\label{rem:finite-rotation}
There are three related differences between our formulation for $\hat{\bH}$ and that of \citet{GoncalvesPerron2014} for strong factors.
First, they formulate the limiting parameter using the probability-limit rotation $\bH_0=\plim\hat{\bH}$, whereas we retain the finite-sample rotation-specific target $\bgamma_{\hat{\bH}}=\hat{\bH}^{-1}\bgamma^*$.
Second, the bootstrap OLS estimator naturally targets the coefficient associated with the realized bootstrap rotation $\hat{\bH}^{\dag}$, and replacing this target by its probability-limit counterpart requires $\sqrt{T}(\bdelta_{\hat{\bH}^{\dag}}-\bdelta^{0\dag})=o_p^\dag(1)$, which need not hold when factor-estimation bias is first-order relevant.
Third, asymptotically equivalent rotations can induce different rotation-specific first-order biases, as reflected in \eqref{eq:muR}, which a common probability-limit formulation cannot distinguish.
\end{rem}

\begin{rem}[The fixed target and the rotation-gap component]
\label{rem:rotation-gap}
The additional bias $c_2\bh_{\bgamma^*}$ for the fixed target $\bdelta^0$ is the first-order limit of $\sqrt{T}(\bdelta_{\hat{\bH}_q}-\bdelta^0)$ derived in \citet{jiang2024Mw}. 
It depends on the eigengap matrix $\bM:=(m_{k\ell})$, $m_{k\ell}:=\plim\{\lambda_k/(\hat{\lambda}_\ell-\lambda_k)\}$ for $k\neq\ell$ and $m_{kk}:=0$; for $\alpha_k=\alpha_\ell$, $m_{k\ell}=\lambda_k/(\lambda_\ell-\lambda_k)$ is large when the two signal eigenvalues are close.
Its dependence on inverse signal eigengaps makes the rotation gap important when eigenvalues are close in finite samples. 
This feature motivates the distinction between target-relative distributional approximation and inference on fixed population coordinates developed in Section~\ref{sec:fixed-inference}.
The corresponding bootstrap rotation-gap result is stated formally as Proposition~\ref{prop:Q-I} in Appendix~\ref{app:rotation-specific-bootstrap}.
\end{rem}

\subsection{Analytically bias-corrected estimators and their distributional approximation}
\label{subsec:linear-inference}

When $\bmu_{\bR}\neq\bzero$, a centered normal approximation is not first-order valid for the uncorrected estimator. 
We first consider the analytical bias correction of \citet{jiang2024Mw}. 
They show that the corresponding analytical correction removes the rotation-specific first-order bias and yields a centered asymptotic distribution. We restate this result below in the unified notation of 
Theorem~\ref{thm:asymptotic-all-rotations}, before establishing its bootstrap counterpart.

For $\bR\in\cR$, the analytically bias-corrected estimator and its bootstrap analogue are
$\hat{\bdelta}_{\bR}^{abc}=\hat{\bdelta}-\hat{\bb}_{\bR}^{a}$ and
$\hat{\bdelta}_{\bR}^{abc\dag}=\hat{\bdelta}^{\dag}-\hat{\bb}_{\bR}^{a\dag}$,
where $\hat{\bb}_{\bR}^{a}$ is the analytical plug-in bias correction of \citet{jiang2024Mw}.
Its explicit expression and the primitive conditions ensuring
$\sqrt{T}\hat{\bb}_{\bR}^{a}\CP\bmu_{\bR}$ are given there.
The bootstrap correction $\hat{\bb}_{\bR}^{a\dag}$ is computed by applying the same formula within each bootstrap sample.
We additionally assume that
$\sqrt{T}\hat{\bb}_{\bR}^{a\dag}\CPb\bmu_{\bR}$ in probability.
Both conditions can be verified given consistent estimators of the idiosyncratic covariance matrix: the former is \citet[Theorem 4]{jiang2024Mw}, and the latter follows from the same argument applied conditionally on the original sample.

\begin{thm}[Asymptotic normality of the analytically bias-corrected estimator]
\label{cor:abc}
Suppose the conditions of Theorem~\ref{thm:asymptotic-all-rotations} hold for $\bR\in\cR$ and $\sqrt{T}\hat{\bb}_{\bR}^{a}\CP\bmu_{\bR}$.
Then, $\sqrt{T}\left( \hat{\bdelta}_{\bR}^{abc} - \bdelta_{\bR} \right) \allowbreak\CD N\left(\bzero,\bSigma_{\bdelta}\right).$ \end{thm}

Theorem~\ref{cor:abc} restates the corresponding non-bootstrap result of \citet{jiang2024Mw}. We now establish that the distribution of the analytically bias-corrected estimator can itself be consistently
approximated by the bootstrap. 

\begin{thm}[Bootstrap validity for the analytically bias-corrected statistic]
\label{thm:abc-bootstrap}
Suppose the conditions of Theorem~\ref{thm:bootstrap-all-rotations} hold for $\bR\in\cR$, together with $\sqrt{T}\hat{\bb}_{\bR}^{a}\CP\bmu_{\bR}$ and $\sqrt{T}\hat{\bb}_{\bR}^{a\dag}\CPb\bmu_{\bR}$ in probability.
Then, $\sqrt{T}\left( \hat{\bdelta}_{\bR}^{abc\dag} - \bdelta_{\bR}^{\dag} \right) \CDb N\left(\bzero,\bSigma_{\bdelta}\right)$ in probability, and consequently
\begin{align}
\sup_{\bx\in\bbR^{r+p}}
\left|
\Pro^{\dag}\!\left[
\sqrt{T}\left(
\hat{\bdelta}_{\bR}^{abc\dag}
-
\bdelta_{\bR}^{\dag}
\right)
\leq\bx
\right]
-
\Pro\!\left[
\sqrt{T}\left(
\hat{\bdelta}_{\bR}^{abc}
-
\bdelta_{\bR}
\right)
\leq\bx
\right]
\right|
\CP
0.
\label{eq:abc-bootstrap-valid}
\end{align}
\end{thm}

The proofs of Theorems~\ref{cor:abc} and \ref{thm:abc-bootstrap} are direct applications of Theorems~\ref{thm:asymptotic-all-rotations} and \ref{thm:bootstrap-all-rotations}, respectively, together with Slutsky's theorem, and are given in Appendix~\ref{app:bc-proofs}.

\section{Bootstrap bias correction and distributional approximation}
\label{sec:bc-db}

As an alternative to analytical bias correction, we estimate the rotation-specific bias by averaging bootstrap estimation errors, avoiding direct estimation of the analytical bias formula and its nuisance quantities. This correction is justified under the conditional-moment conditions in Section~\ref{subsec:bc-db-theory}. To bootstrap the distribution of the resulting bias-corrected estimator, the correction must be recomputed within each first-level sample. The required second resampling level estimates only this conditional mean, not a second-level distribution or quantile.

\subsection[Bootstrap bias correction and implementation]{Bootstrap bias correction and implementation}
\label{subsec:bc-estimators}
\label{subsec:double-bootstrap}

For $\bR\in\cR$, the first-level bootstrap of Section~\ref{sec:3} yields the
draws $\hat{\bdelta}^{\dag(b)}-\bdelta_\bR^{\dag(b)}$, $b=1,\dots,B$. Their
average estimates the bias of $\hat{\bdelta}$ relative to $\bdelta_\bR$: $\hat{\bb}_{\bR}^{\dag} = B^{-1} \sum_{b=1}^{B} \left( \hat{\bdelta}^{\dag(b)} - \bdelta_{\bR}^{\dag(b)} \right),\allowbreak \qquad \hat{\bdelta}_{\bR}^{bc} = \hat{\bdelta} - \hat{\bb}_{\bR}^{\dag} .$ For $\bR=\bH$ this reads
$\hat{\bb}_{\bH}^{\dag}=B^{-1}\sum_{b}(\hat{\bdelta}^{\dag(b)}-\hat{\bdelta})$
and $\hat{\bdelta}_{\bH}^{bc}=2\hat{\bdelta}-B^{-1}\sum_b
\hat{\bdelta}^{\dag(b)}$, the familiar bootstrap bias correction, here justified under weak factors and all three rotation targets by Theorem~\ref{thm:bootstrap-all-rotations}. 
The distributional object of interest is the studentized bias-corrected estimation error
\begin{align}\label{eq:bc-object}
\hat{\bSigma}_{\bdelta}^{-1/2}\sqrt{T}\left(\hat{\bdelta}_{\bR}^{bc}-\bdelta_\bR\right).
\end{align}
Section~\ref{subsec:bc-db-theory} establishes normal and bootstrap approximations to \eqref{eq:bc-object}.

For the conventional target $\bdelta_{\hat{\bH}}$, construct the bootstrap analogue of \eqref{eq:bc-object} as follows.

\noindent\textbf{Step 1.} Apply Steps 1--2 of Section~\ref{sec:3}, retaining the estimates and fitted residuals.
\textbf{Step 2.} Set $(\bF^{0\ddag(b)},\bB^{0\ddag(b)},\bdelta^{0\ddag(b)}):=(\hat{\bF}^{\dag(b)},\hat{\bB}^{\dag(b)},\hat{\bdelta}^{\dag(b)})$. From this population and the fitted residuals, generate $B_2$ draws recursively using the resampling and sign-normalization rules of Section~\ref{sec:3}.
\textbf{Step 3.} Estimate the bias of $\hat{\bdelta}^{\dag(b)}-\bdelta_{\hat{\bH}^{\dag(b)}}$ by 
$\hat{\bb}_{\hat{\bH}}^{\ddag(b)}
=
B_2^{-1}
\sum_{\ell=1}^{B_2}
\left(
\hat{\bdelta}^{\ddag(b,\ell)}
-
\bdelta_{\hat{\bH}^{\ddag(b,\ell)}}
\right)$, and form the bootstrap bias-corrected estimator $\hat{\bdelta}_{\hat{\bH}}^{\dag(b),bc}=\hat{\bdelta}^{\dag(b)}-\hat{\bb}_{\hat{\bH}}^{\ddag(b)}$.
Form the corresponding bootstrap statistic
\begin{align}\label{eq:db-stat}
(\hat{\bSigma}_{\bdelta}^{\dag(b)})^{-1/2}\sqrt{T}(\hat{\bdelta}_{\hat{\bH}}^{\dag(b),bc}-\bdelta_{\hat{\bH}^{\dag(b)}}).
\end{align}

Repeat Steps 1--3 for $b=1,\dots,B$ to obtain the bootstrap distribution of \eqref{eq:db-stat}. The construction applies to all $\bR\in\cR$; details and target-specific formulas are given in the Appendix.

\subsection{Bootstrap bias-corrected estimators and their distributional approximation}
\label{subsec:bc-db-theory}

Similar to the first-level bootstrap target  $\bdelta_{\bR}^{\dag}$ for a rotation $\bR$ defined above, denote the corresponding second-level bootstrap target by $\bdelta_{\bR}^{\ddag}$, so that $\bdelta_{\hat{\bH}}^{\ddag}:=\bdelta_{\hat{\bH}^{\ddag}}$, $\bdelta_{\hat{\bH}_q}^{\ddag}:=\bdelta_{\hat{\bH}_q^{\ddag}}$, and $\bdelta_{\bH}^{\ddag}:=\bdelta^{0\ddag}$. 
Let $\Pro^{\ddag}$, $\E^{\ddag}$ and $\Var^{\ddag}$ denote probability, expectation and variance conditional on the first-level bootstrap sample. 
For later use, define the bootstrap population counterparts of the bias estimates $\hat{\bb}_{\bR}^{\dag}$ and $\hat{\bb}_{\bR}^{\ddag}$ given above as 
  $\bb_{\bR}^{\dag}   =   \E^{\dag}(\hat{\bdelta}^{\dag}  -\bdelta_\bR^{\dag})$ and 
  $\bb_{\bR}^{\ddag} = \E^{\ddag}(\hat{\bdelta}^{\ddag}  -
  \bdelta_\bR^{\ddag})$, 
  respectively, $\bR\in\cR$.
These conditional expectations are the ideal $B\to\infty$ and $B_2\to\infty$ counterparts of the finite Monte Carlo averages $\hat{\bb}_{\bR}^{\dag}$ and $\hat{\bb}_{\bR}^{\ddag}$ used in practice. The asymptotic theory below is stated for these ideal quantities and therefore abstracts from Monte Carlo simulation error.

Two further sets of conditions are required: regularity of the
second-level bootstrap world, and uniform-integrability-type conditions
converting the convergence in distribution of
Theorem~\ref{thm:bootstrap-all-rotations} into convergence of the
conditional means.

\begin{assb}[Second-level bootstrap regularity]
\label{ass:ddag}
Conditional on the first-level bootstrap sample, the second-level
bootstrap errors $(\bE^{\ddag},\bepsilon^{\ddag})$ satisfy the analogues
of Assumptions \ref{ass:errors_b}--\ref{ass:bootstrap_b} with
$(\bF^{0\dag},\bB^{0\dag},\hat{\bF}^{\dag},\hat{\bB}^{\dag},
\hat{\bZ}^{\dag},\hat{\bepsilon}^{\dag})$ replaced by
$(\bF^{0\ddag},\bB^{0\ddag},\hat{\bF}^{\ddag},\hat{\bB}^{\ddag},
\hat{\bZ}^{\ddag},\hat{\bepsilon}^{\ddag})$, where
$\bF^{0\ddag}:=\hat{\bF}^{\dag}$, $\bB^{0\ddag}:=\hat{\bB}^{\dag}$ and
$\bdelta^{0\ddag}:=\hat{\bdelta}^{\dag}$. In particular, the second-level
covariance, signal-strength and eigengap objects consistently estimate
their first-level analogues, and the corresponding conditional CLTs hold
in $\Pro^{\dag}$-probability.
\end{assb}

\begin{assb}[Convergence of conditional biases]
\label{ass:bc-highlevel}
For each $\bR\in\cR$:
(i) $\sqrt{T}\,\bb_{\bR}^{\dag}\CP\bmu_{\bR}$; 
(ii) $\sqrt{T}\,\bb_{\bR}^{\ddag}
=\bmu_{\bR}+o_{p^{\dag}}(1)$, in probability. 
\end{assb}

Under the conditions of Theorem~\ref{thm:bootstrap-all-rotations}, including the corresponding rotation-specific rate conditions, part (i) strengthens the first-level bootstrap convergence in distribution to convergence of the corresponding conditional first moment, while part (ii) imposes the analogous first-moment convergence in the second-level bootstrap world. 

\begin{thm}[Asymptotic normality of the bootstrap bias-corrected estimator]
\label{thm:bc-all-rotations}
Suppose the conditions of Theorem~\ref{thm:bootstrap-all-rotations} hold for $\bR\in\cR$, together with Assumption~\ref{ass:bc-highlevel}(i).
Then, for each $\bR\in\cR$, $\hat{\bSigma}_{\bdelta}^{-1/2} \sqrt{T} \left( \hat{\bdelta}_{\bR}^{bc} - \bdelta_\bR \right) \CD N(\boldsymbol{0},\bI_{r+p}).$ \end{thm}

\begin{thm}[Bootstrap validity for the bias-corrected statistic]
\label{thm:db-all-rotations}
Suppose the conditions of Theorem~\ref{thm:bc-all-rotations} hold, together with Assumptions \ref{ass:ddag} and \ref{ass:bc-highlevel}(ii).
Then, for each $\bR\in\cR$, $\left(\hat{\bSigma}_{\bdelta}^{\dag}\right)^{-1/2} \sqrt{T} \left( \hat{\bdelta}_{\bR}^{\dag,bc} - \bdelta_\bR^{\dag} \right) \CDb N(\boldsymbol{0},\bI_{r+p}),$ in probability, and consequently $\sup_{\bx\in\bbR^{r+p}} \bigl| \allowbreak\Pro^{\dag} \bigl[ \bigl(\hat{\bSigma}_{\bdelta}^{\dag}\bigr)^{-1/2} \allowbreak\sqrt{T} \bigl( \hat{\bdelta}_{\bR}^{\dag,bc} - \bdelta_\bR^{\dag} \bigr) \allowbreak\leq \bx \bigr] -\allowbreak \allowbreak\Pro \bigl[ \hat{\bSigma}_{\bdelta}^{-1/2} \allowbreak\sqrt{T} \bigl( \hat{\bdelta}_{\bR}^{bc} - \bdelta_\bR \bigr) \allowbreak\leq \bx \bigr] \bigr| = o_p(1).$ \end{thm}

\section{Inference on fixed parameter restrictions and rotation uncertainty}
\label{sec:fixed-inference}

The preceding sections establish normal and bootstrap approximations to the distributions of $\hat{\bdelta}$ and its bias-corrected versions relative to the rotation-specific targets $\bdelta_{\bR}$. 
These results concern estimation error relative to $\bdelta_{\bR}$ and do not in general provide inference on an arbitrary restriction of the fixed population parameter $\bdelta^0$.  
This distinction is important because empirical questions concerning the contribution of individual factors are naturally formulated with respect to the population PC basis $\bF^0$. 

\subsection{Fixed parameter restrictions}
\label{subsec:fixed-restrictions}

Consider the fixed population restriction
\begin{align}\label{eq:fixed-restriction}
H_0:\ \bA\bdelta^0=\ba,
\end{align}
where $\bA$ is an $s\times(r+p)$ matrix of full row rank. For any $\bR\in\cR$, 
\begin{align}\label{eq:fixed-R-decomp}
\sqrt{T}\bA(\hat{\bdelta}-\bdelta^0)=\sqrt{T}\bA(\hat{\bdelta}-\bdelta_{\bR})+\sqrt{T}\bA(\bdelta_{\bR}-\bdelta^0).
\end{align}

The decomposition separates rotation-specific estimation error from the target gap and also applies to the bias-corrected estimators. The gap vanishes for $\bR=\bH$, but need not be negligible on the $\sqrt{T}$ scale for $\bR=\hat{\bH}$ or $\hat{\bH}_q$; inference relative to these data-dependent targets therefore need not be valid for \eqref{eq:fixed-restriction}.

Restrictions on $\bbeta$ and the joint null $\bgamma^0=\bzero$ are rotation invariant, unlike element-wise or generic linear restrictions on $\bgamma^0$. In particular, no contribution from factor $j$ in the population PC basis corresponds to
\begin{align}\label{eq:gammaj-zero}
H_0:\ \gamma_j^0=0.
\end{align}
Recall that $\bgamma_{\hat{\bH}}=\tilde{\bH}^{-1}\bgamma^0$ and $\bgamma_{\hat{\bH}_q}=\tilde{\bH}_q^{-1}\bgamma^0$.
Under \eqref{eq:gammaj-zero}, $\gamma_{\hat{\bH},j}=\sum_{k\neq j}[\tilde{\bH}^{-1}]_{jk}\gamma_k^0$ and $\gamma_{\hat{\bH}_q,j}=\sum_{k\neq j}[\tilde{\bH}_q^{-1}]_{jk}\gamma_k^0$.
Hence, $\gamma_j^0=0$ need not imply $\gamma_{\hat{\bH},j}=0$ or $\gamma_{\hat{\bH}_q,j}=0$, because the corresponding relative rotations generally mix the coordinates of $\bgamma^0$.

The difference between a rotation-specific coefficient and the corresponding fixed population coefficient can be particularly important when adjacent signal eigenvalues are close. 
For the diagnostic developed below, we focus on the alternative rotation $\hat{\bH}_q$ and define 
\begin{align}\label{eq:Hq-fixed-gap}
\bg:=\bgamma_{\hat{\bH}_q}-\bgamma^0=(\tilde{\bH}_q^{-1}-\bI_r)\bgamma^0.
\end{align}
Under Theorem~\ref{thm:asymptotic-all-rotations}, the first-order behavior of $\sqrt{T}\bg$ depends on inverse signal eigengaps. Close eigenvalues can therefore make individual coefficients sensitive to estimated PC directions even when the factor space is accurately estimated.

\subsection{A bootstrap diagnostic for rotation uncertainty in \texorpdfstring{$\bgamma^0$}{gamma0}}
\label{subsec:Shat}

The first-level bootstrap directly provides $\sqrt{T}\bg^{\dag(b)}:=\sqrt{T}(\bgamma_{\hat{\bH}_q^{\dag(b)}}-\bgamma^{0\dag})=\sqrt{T}(\tilde{\bH}_q^{\dag(b)-1}-\bI_r)\bgamma^{0\dag}$ without additional resampling. For $j\in[r]$, define
\begin{align}\label{eq:Shat}
\widehat{\mathrm{RU}}_j^2:=\frac{\Var^{\dag}(\sqrt{T}\,g_j^{\dag})}{[\hat{\bSigma}_{\bdelta}]_{jj}}.
\end{align}

In practice, $\Var^\dag$ is replaced by the sample variance over the $B$ first-level bootstrap draws. 
Thus, $\widehat{\mathrm{RU}}_j^2$ is the ratio of bootstrap rotation-gap variance to sampling variance in the corresponding factor coefficient, while $\widehat{\mathrm{RU}}_j$ compares the corresponding standard deviations. 
Equivalently, since $\gamma_j^{0\dag}$ is fixed conditional on the sample, the numerator measures the bootstrap variation of the rotation-specific target $\gamma_{\hat{\bH}_q^\dag,j}$ itself. 

\begin{rem}\label{rem:Shat}
Under the maintained rate and eigengap conditions and suitable conditional second-moment convergence, the degenerate first-order gap limit and positive limiting coefficient variance imply $\widehat{\mathrm{RU}}_j=o_p(1)$. Larger finite-sample values indicate greater sensitivity to estimated PC directions. This is a diagnostic, not a size-controlled pretest: the theory supplies no fixed critical value for classifying rotation uncertainty. A formal analysis under shrinking eigengaps, including the behavior of $\widehat{\mathrm{RU}}_j$ when rotation uncertainty remains first-order relevant, is beyond the scope of this paper.%
\end{rem}

\section{Monte Carlo experiments}\label{sec: MC}

We examine the finite-sample performance of the proposed bootstrap procedures for bias correction, distributional approximation, and inference on fixed population restrictions.
We consider two complementary classes of designs covered by our theory: distinct divergence rates for bias correction and target-relative normal and bootstrap approximation, and equal divergence rates with progressively smaller eigengaps for rotation uncertainty and inference on individual coordinates of $\bgamma^0$.
Fixed-zero restrictions are also considered to distinguish inference on fixed population parameters from distributional approximation relative to data-dependent rotation targets.

\subsection{Design}\label{subsec:MCdesign}

Data are generated from $\bX=\bF^{0}\bB^{0\prime}+\bE$ with $r=2$, where, as in \citet{jiang2024Mw}, each replication draws a $T\times N$ matrix $\bA$ with i.i.d.\ $N(0,1)$ entries and SVD $\bA=\bU\bS\bV'$, sets $\bF^{0}$ to $\sqrt{T}$ times the first $r$ columns of $\bU$ and $\bB^{0}$ to the first $r$ columns of $\bV$ post-multiplied by $\bD^{1/2}\bN^{1/2}$, with $\bD=\diag(d_1,d_2)$ and $\bN=\diag(N^{\alpha_1},N^{\alpha_2})$.
Given $\bH=\big(\begin{smallmatrix}1 & 1/2\\ 1/2 & 2\end{smallmatrix}\big)$, we set $\bF^{*}=\bF^{0}\bH^{-1}$ and $\bB^{*}=\bB^{0}\bH'$.
The idiosyncratic errors are cross-sectionally heteroskedastic and skewed, with $\be_t=\bSigma_e^{1/2}\bxi_t$, $\bSigma_e=\diag(\sigma_{e1}^2,\dots,\sigma_{eN}^2)$, $\sigma_{ei}^2\sim\text{i.i.d. }U[0.5,1.5]$, and i.i.d.\ standardized $\chi^2(4)$ elements of $\bxi_t$.

The forecast equation is $y_{t+1}=\bff_t^{0\prime}\bgamma^0+\bw_t'\bbeta+\epsilon_{t+1}$, $t=1,\dots,T$.
For the bias and distributional-approximation experiments of Sections~\ref{subsubsec:MCbias} and \ref{subsubsec:MCdist}, we set $\bgamma^0=\bone_r$ and $\bbeta=\bone_p$, with $p=2$; for the fixed-restriction experiments of Section~\ref{subsubsec:MCgap} and Appendix~\ref{app:fixed-difA}, we set $\bgamma^0=(1,0)'$ and $\bbeta=(1,1)'$, so that $H_0:\gamma_2^0=0$ holds in every replication.
As a rotation-invariant control, we set $\bgamma^0=(1,1)'$ and $\bbeta=(0,1)'$, so that $H_0:\beta_1=0$.
In all designs, $w_{t,p}=1$ and $w_{t,\ell}=\sigma_w[\rho_{fw}\bff_t^{0\prime}\bone_r r^{-1/2}+(1-\rho_{fw}^2)^{1/2}\zeta_{t,\ell}]$, where $\zeta_{t,\ell}\sim i.i.d.\,N(0,1)$, $\sigma_w^2=1$, and $\rho_{fw}=0.6$.
The regression errors are non-normal, with $\epsilon_{t+1}=\sigma_\epsilon v_{t+1}$, where $v_{t+1}$ is standardized $t(5)$ and $\sigma_\epsilon^2=0.5$.

The first class has distinct divergence rates with $d_1=d_2=1$ and $(\alpha_1,\alpha_2)\in\{(1.0,0.9),\\(0.9,0.8),(0.8,0.7)\}$, so that the relative eigengap $1-\lambda_2/\lambda_1=1-N^{-0.1}$ widens asymptotically.
The second sets $\alpha_1=\alpha_2=0.8$, $d_1=1$, and $d_2\in\{0.50,0.75,0.90\}$, giving fixed relative eigengaps $0.50$, $0.25$, and $0.10$.
Both classes satisfy the maintained theoretical assumptions, while the equal-rate designs examine the finite-sample consequences of increasingly close signal eigenvalues, for which rotation uncertainty can become substantial as discussed in Remark~\ref{rem:rotation-gap}.
We consider $T=N\in\{50,100,200\}$.

We estimate $\bF^0$ by the PC estimator $\hat{\bF}$ of Section~\ref{sec:2}; in the original sample and at both bootstrap levels, columns are matched to the corresponding reference factors by maximizing absolute sample correlations and signed so that the matched correlations are positive.
For each target $\bdelta_\bR$, $\bR\in\{\hat{\bH},\hat{\bH}_q,\bH\}$, we compare the uncorrected LS estimator, the analytical bias correction based on the plug-in bias expressions in Theorem~\ref{thm:bootstrap-all-rotations}, the GP bootstrap bias correction \citep{GoncalvesPerron2014,gonccalves2020bootstrapping} for $\bR\in\{\hat{\bH},\hat{\bH}_q\}$, and the proposed bootstrap bias correction of Section~\ref{subsec:bc-estimators}.

For distributional approximation, we construct nominal 95\% confidence intervals using standard normal or studentized bootstrap critical values, with each first-level statistic for the bootstrap bias-corrected estimator using the second-level resampling of Section~\ref{subsec:double-bootstrap}.
For the fixed-restriction experiments, the observed statistic tests the fixed population restriction directly, while bootstrap critical values come from the corresponding recentered bootstrap distribution.
Throughout, we use Rademacher wild-bootstrap multipliers taking $-1$ and $1$ with probability $1/2$ at both levels, with $B=299$ first-level and $B_2=19$ second-level draws; \citet{DavidsonFlachaire2008} show that this choice can yield smaller finite-sample rejection-probability errors than the \citet{Mammen1993} two-point bootstrap, including under asymmetric disturbances. 
Bootstrap critical values are the empirical 2.5th and 97.5th quantiles of the corresponding studentized bootstrap distribution, computed using empirical order statistics.
For studentization, we use $\hat{\bSigma}_{\bdelta}(\hat{\bZ},\hat{\bepsilon})$ defined in \eqref{eq:Sigma-delta-hat}, with $\hat{\bSigma}_{\bZ^0\bepsilon}=T^{-1}\sum_{t=1}^T\hat{\epsilon}_{t+1}^{\,2}\hat{\bz}_t\hat{\bz}_t'$, and recompute its analogue $\hat{\bSigma}_{\bdelta}(\hat{\bZ}^{\dag},\hat{\bepsilon}^{\dag})$ in each first-level bootstrap sample.
For the analytical bias corrections, we estimate $\bGamma$ by $\hat{\bB}'\diag(\hat\sigma_1^2,\dots,\hat\sigma_N^2)\hat{\bB}$, where $\hat\sigma_i^2=T^{-1}\sum_t\hat e_{t,i}^2$.
All experiments use 2,000 Monte Carlo replications.

For the distributional experiments, we report bias, RMSE, and empirical coverage of nominal 95\% confidence intervals for the corresponding target, evaluated relative to $\bdelta^0$ for $\bR=\bH$ and to the realized rotation-specific target in each replication for $\bR\in\{\hat{\bH},\hat{\bH}_q\}$.
For the fixed-restriction experiments, we report empirical rejection frequencies of nominal 5\% two-sided tests of $H_0:\gamma_2^0=0$ and, as a rotation-invariant control, $H_0:\beta_1=0$.
For the equal-divergence-rate designs, we additionally report the median rotation-uncertainty diagnostic $\widehat{\mathrm{RU}}_2$ defined in \eqref{eq:Shat}.
The main text reports results most directly relevant to bias correction, distributional approximation, fixed-parameter inference, and rotation uncertainty, with additional coverage, RMSE, and diagnostic results in Appendix~\ref{app:addMC}.

\subsection{Results}\label{subsec:MCresults}

We organize the results in three steps.
First, under heterogeneous divergence rates, we compare the competing bias corrections for the conventional rotation $\hat{\bH}$.
Second, in the same designs, we compare normal and bootstrap approximations to the distributions of the estimators relative to the three targets.
Third, we turn to equal divergence rates and examine how increasingly close signal eigenvalues affect inference on the fixed population coordinate and the rotation-uncertainty diagnostic $\widehat{\mathrm{RU}}_2$.
Additional fixed-restriction experiments under heterogeneous divergence rates and rotation-invariant checks are reported in Appendix~\ref{app:addMC}.

\subsubsection{Bias correction under heterogeneous weak factors}
\label{subsubsec:MCbias}

We begin with the designs having heterogeneous factor strengths and distinct divergence rates.
We focus first on the conventional rotation $\hat{\bH}$, for which the first-order bias is largest and the rate requirement is most demanding.

Figure~\ref{fig:bias.difA} reports the bias of the estimators relative to $\gamma_{\hat{\bH},2}$.
The uncorrected estimator has a substantial negative bias, which becomes larger as the second factor weakens.
All three bias corrections reduce this bias considerably.
The proposed bootstrap correction gives the smallest absolute bias in every design and sample size, followed closely by the GP correction, while the analytical correction leaves somewhat more residual bias.
The bias decreases with the sample size for all methods, as predicted by the asymptotic theory.

\begin{figure}[!htb]
\centering
\begin{subfigure}[b]{0.32\textwidth}\centering
\includegraphics[width=\textwidth]{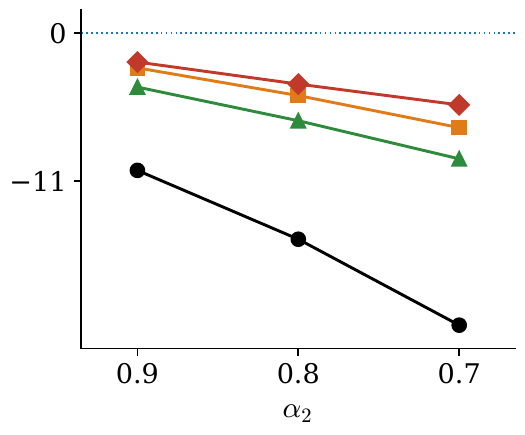}
\caption{$T=50$}
\end{subfigure}\hfill
\begin{subfigure}[b]{0.32\textwidth}\centering
\includegraphics[width=\textwidth]{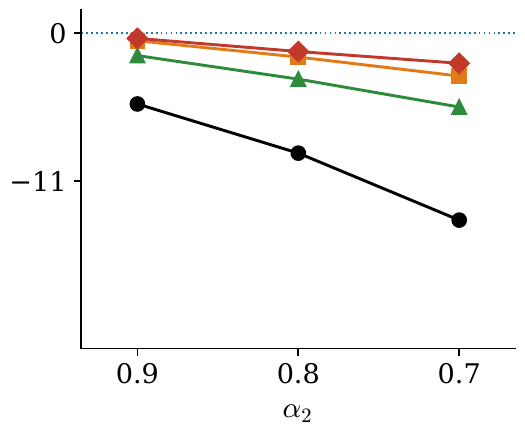}
\caption{$T=100$}
\end{subfigure}\hfill
\begin{subfigure}[b]{0.32\textwidth}\centering
\includegraphics[width=\textwidth]{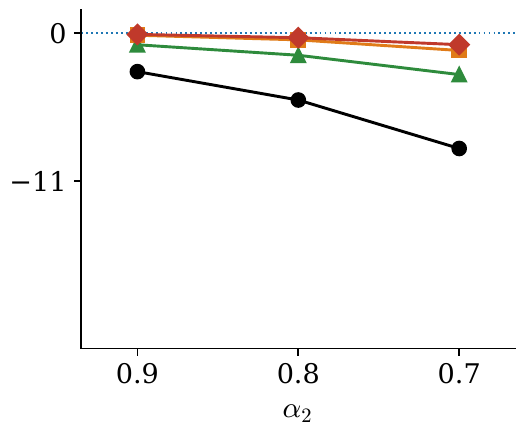}
\caption{$T=200$}
\end{subfigure}
\includegraphics[width=0.95\textwidth]{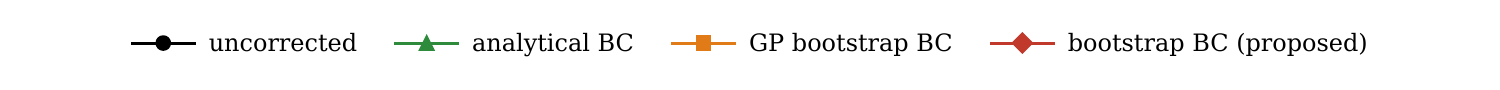}
\caption{Bias  ($\times 100$) of estimators relative to $\gamma_{\hat{\bH},2}$}
\label{fig:bias.difA}
\end{figure}

The comparison with GP is informative because the two bootstrap corrections handle the rotation-specific target differently. 
The proposed correction has smaller absolute bias in every heterogeneous-rate design and sample size, with essentially unchanged RMSE, supporting direct construction of the bootstrap estimation error relative to the target. 
We therefore omit GP from the distributional comparisons below. 

\subsubsection{Normal and bootstrap distributional approximation}
\label{subsubsec:MCdist}

\begin{figure}[!htb]
\centering
\begin{subfigure}[b]{0.32\textwidth}\centering
\includegraphics[width=\textwidth]{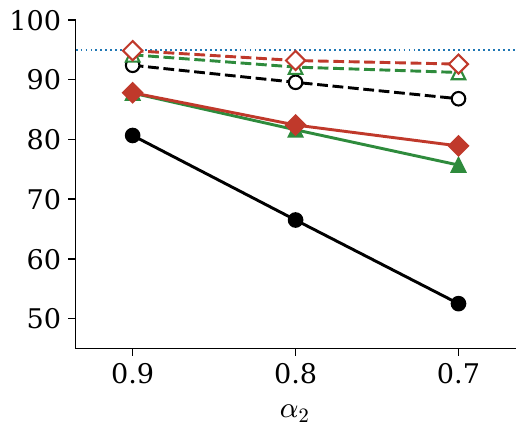}
\caption{$\gamma_{\hat{\bH},2}$, $T=50$}
\end{subfigure}\hfill
\begin{subfigure}[b]{0.32\textwidth}\centering
\includegraphics[width=\textwidth]{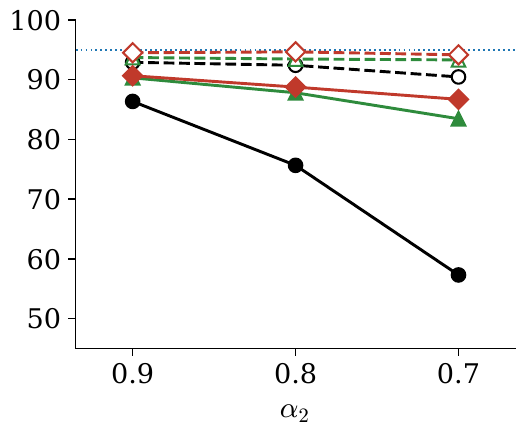}
\caption{$\gamma_{\hat{\bH},2}$, $T=100$}
\end{subfigure}\hfill
\begin{subfigure}[b]{0.32\textwidth}\centering
\includegraphics[width=\textwidth]{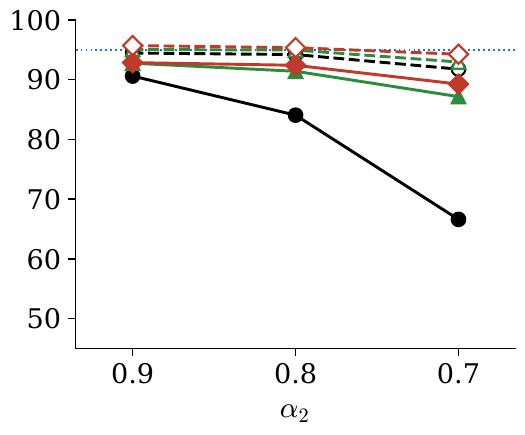}
\caption{$\gamma_{\hat{\bH},2}$, $T=200$}
\end{subfigure}
\begin{subfigure}[b]{0.32\textwidth}\centering
\includegraphics[width=\textwidth]{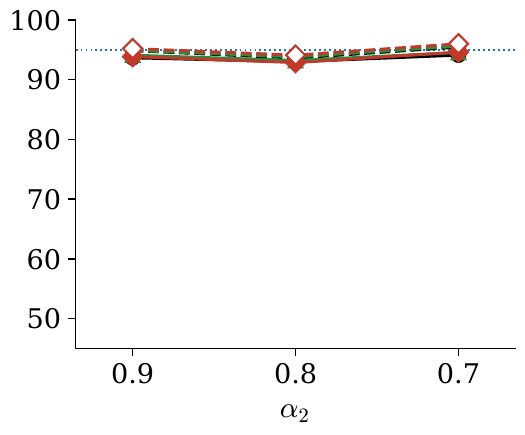}
\caption{$\gamma_{\hat{\bH}_q,2}$, $T=50$}
\end{subfigure}\hfill
\begin{subfigure}[b]{0.32\textwidth}\centering
\includegraphics[width=\textwidth]{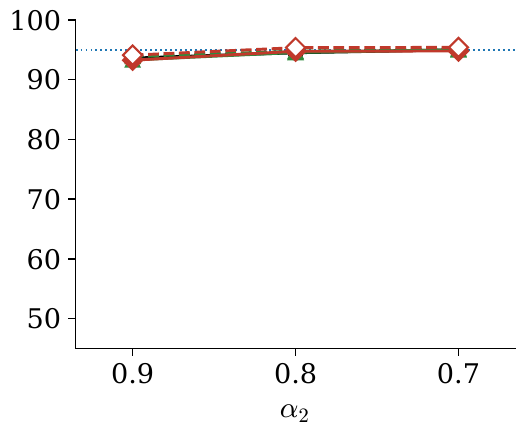}
\caption{$\gamma_{\hat{\bH}_q,2}$, $T=100$}
\end{subfigure}\hfill
\begin{subfigure}[b]{0.32\textwidth}\centering
\includegraphics[width=\textwidth]{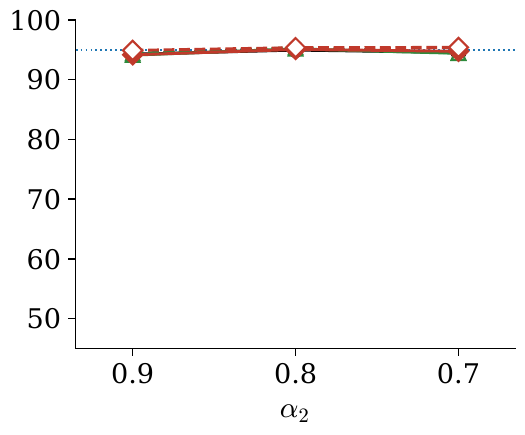}
\caption{$\gamma_{\hat{\bH}_q,2}$, $T=200$}
\end{subfigure}
\begin{subfigure}[b]{0.32\textwidth}\centering
\includegraphics[width=\textwidth]{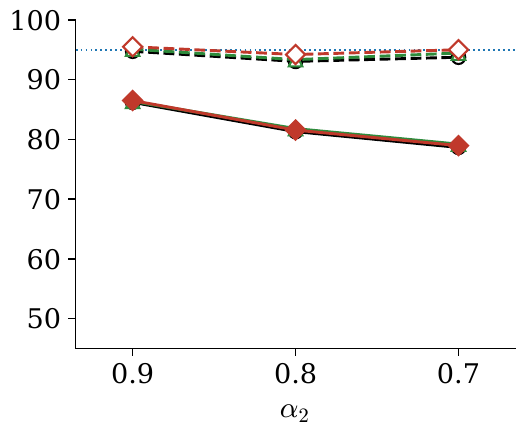}
\caption{$\gamma_2^0$, $T=50$}
\end{subfigure}\hfill
\begin{subfigure}[b]{0.32\textwidth}\centering
\includegraphics[width=\textwidth]{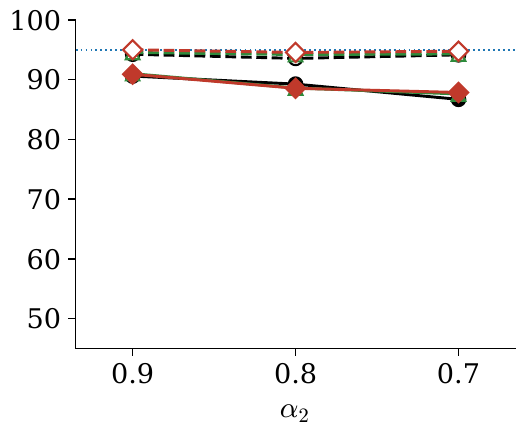}
\caption{$\gamma_2^0$, $T=100$}
\end{subfigure}\hfill
\begin{subfigure}[b]{0.32\textwidth}\centering
\includegraphics[width=\textwidth]{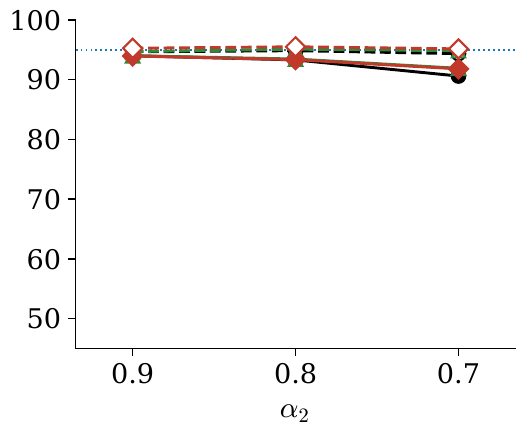}
\caption{$\gamma_2^0$, $T=200$}
\end{subfigure}
\includegraphics[width=0.95\textwidth]{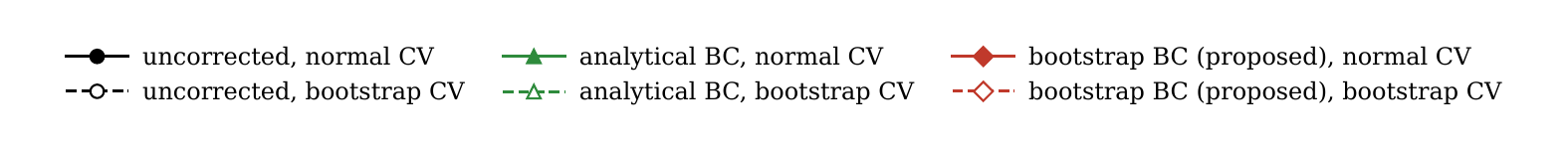}
\caption{Coverage of 95\% confidence intervals for the second factor coefficient}
\label{fig:coverage.difA}
\vspace{2pt}
\begin{minipage}{0.95\textwidth}
\footnotesize
\justifying
\noindent Notes: The top, middle, and bottom rows correspond to $\gamma_{\hat{\bH},2}$, $\gamma_{\hat{\bH}_q,2}$, and $\gamma_2^0$, respectively; coverage for the first two is relative to the realized rotation-specific target in each replication, while the bottom row uses the fixed population target.
The horizontal dotted line denotes nominal 95\% coverage.
\end{minipage}
\end{figure}

Figure~\ref{fig:coverage.difA} compares normal and bootstrap coverage of nominal 95\% confidence intervals for the second factor coefficient under the three targets and heterogeneous divergence rates.
For the conventional target $\gamma_{\hat{\bH},2}$, normal intervals based on the uncorrected estimator substantially undercover, particularly when the second factor is weak, and bias correction alone leaves noticeable undercoverage in smaller samples.
Studentized bootstrap intervals substantially improve coverage, with the proposed bootstrap bias correction combined with bootstrap critical values generally closest to the nominal level.

For the alternative target $\gamma_{\hat{\bH}_q,2}$, coverage is already close to 95\% for all methods, consistent with its smaller first-order bias and weaker rate requirement, leaving little difference between the normal and bootstrap approximations.

For the fixed target $\gamma_2^0$, normal intervals substantially undercover in smaller samples, while analytical or bootstrap bias correction alone has little effect.
Studentized bootstrap intervals instead produce coverage close to 95\% even at $T=50$, with the proposed bootstrap bias-corrected procedure particularly close to nominal coverage throughout.
Thus, the main finite-sample improvement for the fixed target comes from bootstrap approximation of the studentized distribution rather than bias correction alone.

Additional experiments imposing the fixed null $H_0:\gamma_2^0=0$ under heterogeneous divergence rates and rotation-invariant coverage checks for $\beta_1$ are reported in Appendices~\ref{app:fixed-difA} and \ref{app:beta-difA}. 

\subsubsection{Close eigenvalues and rotation uncertainty}
\label{subsubsec:MCgap}

We now turn to equal divergence rates, setting $\alpha_1=\alpha_2=0.8$ and reducing the relative eigengap from $0.50$ to $0.25$ and $0.10$.
We retain the fixed-null design $\bgamma^0=(1,0)'$, so that $H_0:\gamma_2^0=0$ continues to hold in every replication.
These designs isolate the finite-sample effect of increasingly close, though distinct, signal eigenvalues on inference for an individual population factor coefficient.

\begin{figure}[!htb]
\centering
\begin{subfigure}[b]{0.32\textwidth}\centering
\includegraphics[width=\textwidth]{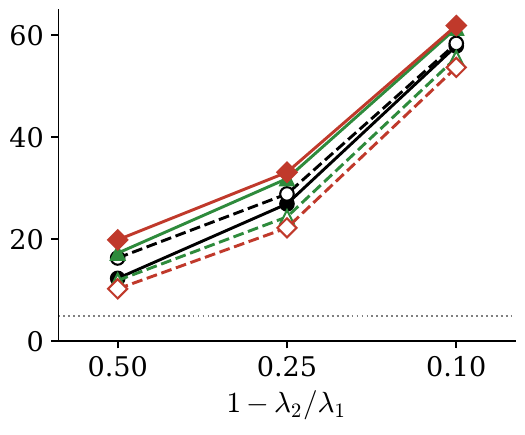}
\caption{$\hat{\bH}$, $T=50$}
\end{subfigure}\hfill
\begin{subfigure}[b]{0.32\textwidth}\centering
\includegraphics[width=\textwidth]{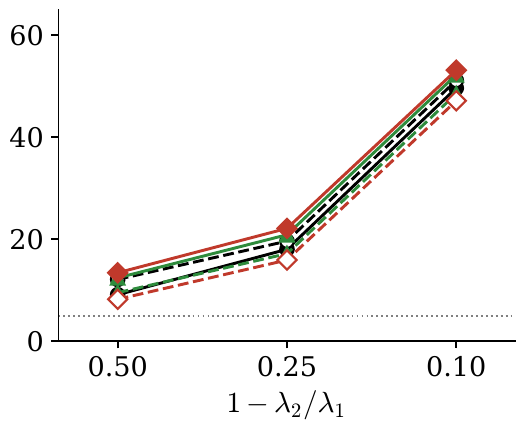}
\caption{$\hat{\bH}$, $T=100$}
\end{subfigure}\hfill
\begin{subfigure}[b]{0.32\textwidth}\centering
\includegraphics[width=\textwidth]{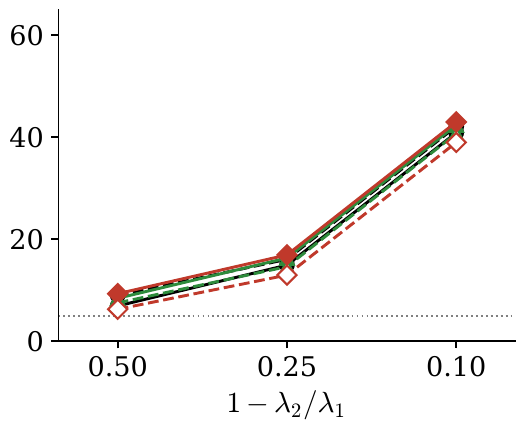}
\caption{$\hat{\bH}$, $T=200$}
\end{subfigure}
\begin{subfigure}[b]{0.32\textwidth}\centering
\includegraphics[width=\textwidth]{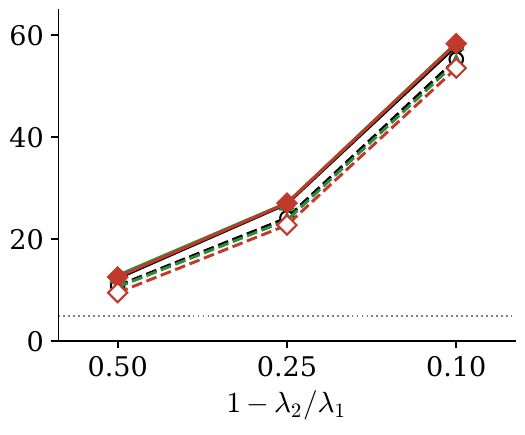}
\caption{$\hat{\bH}_q$, $T=50$}
\end{subfigure}\hfill
\begin{subfigure}[b]{0.32\textwidth}\centering
\includegraphics[width=\textwidth]{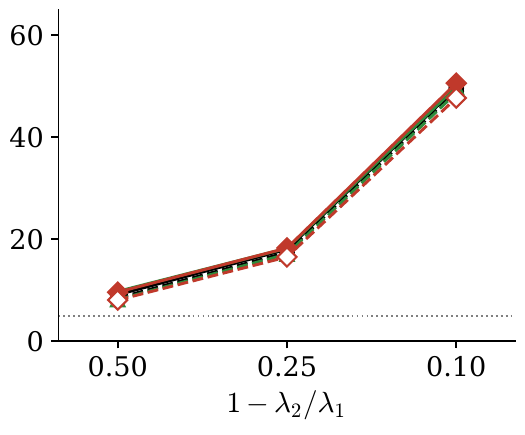}
\caption{$\hat{\bH}_q$, $T=100$}
\end{subfigure}\hfill
\begin{subfigure}[b]{0.32\textwidth}\centering
\includegraphics[width=\textwidth]{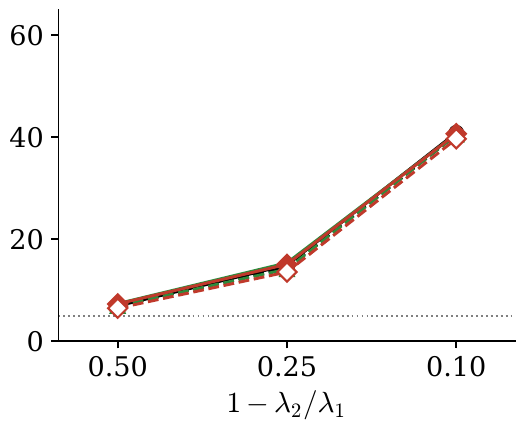}
\caption{$\hat{\bH}_q$, $T=200$}
\end{subfigure}
\begin{subfigure}[b]{0.32\textwidth}\centering
\includegraphics[width=\textwidth]{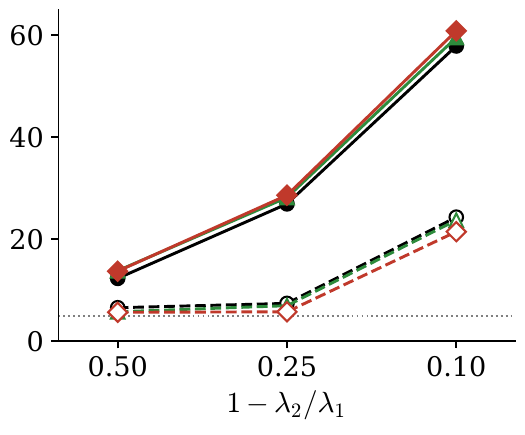}
\caption{$\bH$, $T=50$}
\end{subfigure}\hfill
\begin{subfigure}[b]{0.32\textwidth}\centering
\includegraphics[width=\textwidth]{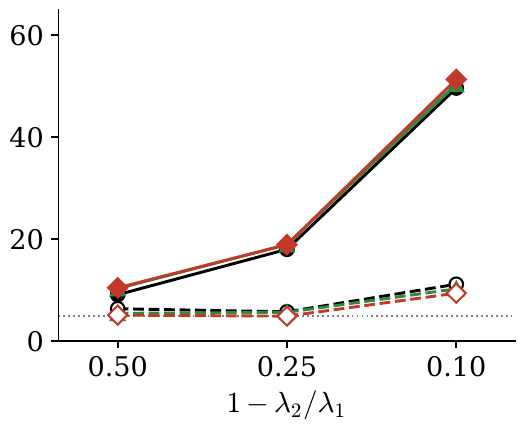}
\caption{$\bH$, $T=100$}
\end{subfigure}\hfill
\begin{subfigure}[b]{0.32\textwidth}\centering
\includegraphics[width=\textwidth]{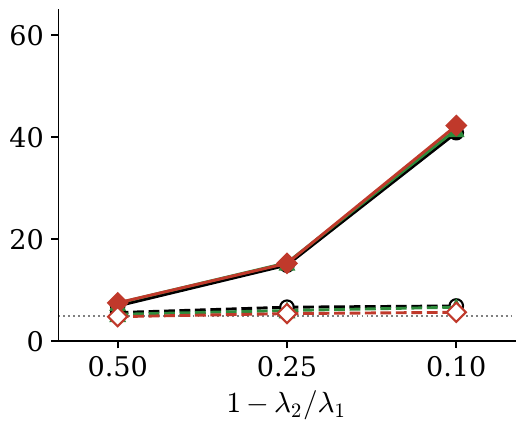}
\caption{$\bH$, $T=200$}
\end{subfigure}
\includegraphics[width=0.95\textwidth]{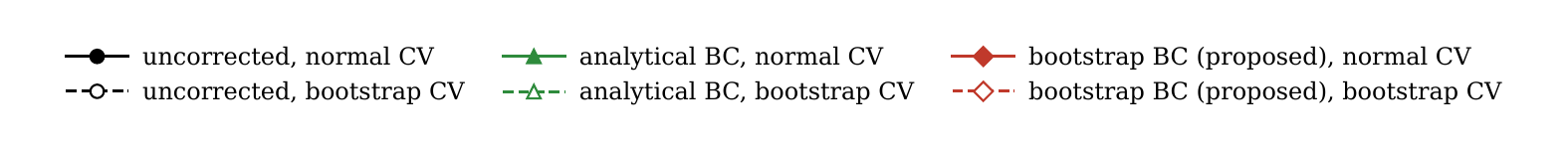}
\caption{Size for $H_0:\gamma_2^0=0$ under close eigengaps}
\label{fig:size.zero.eqA08}
\vspace{2pt}
\begin{minipage}{0.95\textwidth}
\footnotesize
\justifying
\noindent Notes: Rejection frequencies are percentages for nominal 5\% two-sided tests of $H_0:\gamma_2^0=0$; the top, middle, and bottom rows use the estimator and bootstrap construction associated with $\hat{\bH}$, $\hat{\bH}_q$, and $\bH$, respectively.
Normal tests use the standard $[-1.96,1.96]$ critical region, while bootstrap tests use the 2.5th and 97.5th percentiles of the corresponding recentered studentized bootstrap distribution, with the double-bootstrap statistic used for the proposed bias-corrected estimator.
The horizontal dotted line denotes the nominal 5\% level.
\end{minipage}
\end{figure}

Figure~\ref{fig:size.zero.eqA08} shows that fixed-coordinate inference becomes increasingly distorted as the relative eigengap narrows from $0.50$ to $0.10$.
For the data-dependent constructions $\hat{\bH}$ and $\hat{\bH}_q$, at the relative eigengap $0.10$ even the proposed bootstrap bias correction with bootstrap critical values rejects the true null in about 54\%, 47--48\%, and 39--40\% of replications for $T=50$, $100$, and $200$, respectively.
This persistent distortion reflects the distinction in Section~\ref{subsec:fixed-restrictions}: the bootstrap approximates the distribution relative to the data-dependent target, whereas the null concerns the fixed population coefficient $\gamma_2^0$.
For the fixed-target construction $\bH$, rejection frequencies remain close to nominal at relative eigengaps $0.50$ and $0.25$, while at $0.10$ the proposed bootstrap procedure rejects in 21.5\%, 9.4\%, and 5.7\% of replications for $T=50$, $100$, and $200$, respectively.
Thus, close signal eigenvalues can impair fixed-coordinate inference in small samples even under $\bH$, but this distortion declines rapidly with $T$, whereas the more persistent distortions under $\hat{\bH}$ and $\hat{\bH}_q$ reflect the additional rotation gap.

The median rotation-uncertainty diagnostic $\widehat{\mathrm{RU}}_2$ from the same experiment under $H_0:\gamma_2^0=0$ increases markedly as the eigengap narrows: at $T=100$, it is $0.50$, $1.07$, and $2.58$ for relative eigengaps $0.50$, $0.25$, and $0.10$, respectively. 
As $T$ increases, $\widehat{\mathrm{RU}}_2$ declines for the wider eigengaps $0.50$ and $0.25$, while remaining above two at the closest eigengap $0.10$ for all three sample sizes. 
The persistence of a large $\widehat{\mathrm{RU}}_2$ at the closest eigengap parallels the persistent size distortion associated with the data-dependent rotation constructions, showing that the $\hat{\bH}_q$-based diagnostic provides a feasible indication of when variation in the estimated PC direction is large relative to the ordinary sampling uncertainty of the corresponding coefficient. 

\section{Empirical application}\label{sec:empirics}

We apply the proposed procedures to predictive regressions for excess bond returns using macroeconomic factors.
The dependent variable is the $h$-month-ahead log excess return on four-year zero-coupon bonds constructed from the \citet{GurkaynakSackWright2011} yields, for $h\in\{12,24\}$.
The predictors are the \citet{CochranePiazzesi2005} (CP) bond return forecasting factor and $r=5$ principal components extracted from $N=65$ FRED-MD macroeconomic series \citep{MN2016} belonging to Groups 1--4 (output and income; labor market; housing; consumption, orders and inventories), so that the macro factors capture information beyond the yield curve.
We consider the subsamples 1976M6--2004M1 ($T=313$) and 2004M2--2024M10 ($T=249$ for $h=12$), reflecting the structural change in bond risk premia documented by \citet{jiang2024Mw}.
Using the same data and specification as \citet{jiang2024Mw} permits direct comparison of their analytical bias correction with the bootstrap procedures proposed here.

Following \citet{jiang2024Mw}, inference uses prewhitened Newey--West standard errors \citep{AndrewsMonahan1992} with bandwidth at least $h-1$, and PC signs are normalized so that the uncorrected LS factor coefficients are positive, with the same orientation retained for all bias-corrected estimators.
To accommodate serial correlation in overlapping-return regression errors and cross-sectional correlation in the factor-model errors, we implement the bootstrap of Section~\ref{sec:3} using dependent resampling.
Following \citet{gonccalves2020bootstrapping}, we estimate $\bSigma_e$ by POET \citep{FanEtAl2013} at the thresholds used in \citet{jiang2024Mw}, whiten the residuals as $\hat{\bU}=\hat{\bE}\hat{\bSigma}_e^{-1/2}$, block-resample each whitened series, and re-color the draws as $\bE^{\dag}=\hat{\bU}^{\dag}\hat{\bSigma}_e^{1/2}$.
The regression errors are block-resampled with the same block length following \citet{DjogbenouGoncalvesPerron2015}; we set $\ell=h$ to span the mechanically induced dependence horizon, with sensitivity to $\ell\in\{4,7,12,24\}$ for $h=12$ reported in Section~\ref{app:empblock}.
At the second level the same scheme is applied within each first-level sample, with POET recomputed from the first-level residuals at the same threshold, and all bootstrap $t$-statistics use the same prewhitened HAC studentization as the original sample.
We use $B=999$ first-level and $B_2=49$ second-level replications, and Appendix~\ref{app:depboot} reports a Monte Carlo validation of this dependent-resampling implementation.

We focus on inference for the fixed population parameter $\bdelta^0$ and therefore use the fixed-target construction $\bR=\bH$ for both analytical and bootstrap bias correction.
For each coefficient we report LS, analytically bias-corrected, and bootstrap bias-corrected point estimates with two-sided normal and bootstrap $p$-values for $H_0:\delta_j^0=0$, recomputing the analytical correction within each first-level sample and using the second-level resampling of Section~\ref{subsec:double-bootstrap} for the bootstrap bias-corrected estimator.
We also report $\widehat{\mathrm{RU}}_j$ from \eqref{eq:Shat} and the relative eigengaps $1-\hat\lambda_{k+1}/\hat\lambda_k$ of the macro panel.
The $h=12$ results are reported in Table~\ref{tab:emp12}, while those for $h=24$ are in Appendix~\ref{app:emp}.

\begin{table}[!htb]
\centering
\caption{Estimation and inference for $\bdelta^0$: one-year-ahead forecasts of four-year excess bond returns ($h=12$)}
\label{tab:emp12}
{\small
\setlength{\tabcolsep}{4pt}
\begin{tabular}{lcccccc}
\hline\hline
 & $\hat f_1$ & $\hat f_2$ & $\hat f_3$ & $\hat f_4$ & $\hat f_5$ & CP \\
\hline
\multicolumn{7}{c}{1976M6--2004M1\quad ($T=313$;\ rel.\ gaps $0.64,\ 0.37,\ 0.29,\ 0.04$)} \\[2pt]
LS estimate                              & 0.95 & 0.43 & 0.08 & 0.28 & 1.05 & 2.41 \\
\hspace{1.5em}$p$-value: normal          & [0.22] & [0.47] & [0.85] & [0.25] & [0.11] & [0.00] \\
\hspace{1.5em}$p$-value: bootstrap   & [0.32] & [0.58] & [0.87] & [0.52] & [0.12] & [0.00] \\[2pt]
Analytically bias-corrected estimate     & 0.95 & 0.43 & 0.14 & -0.28 & 1.30 & 2.43 \\
\hspace{1.5em}$p$-value: normal          & [0.21] & [0.48] & [0.75] & [0.26] & [0.05] & [0.00] \\
\hspace{1.5em}$p$-value: bootstrap   & [0.32] & [0.59] & [0.82] & [0.62] & [0.12] & [0.00] \\[2pt]
Bootstrap bias-corrected estimate        & 0.94 & 0.42 & 0.11 & 0.26 & 1.19 & 2.41 \\
\hspace{1.5em}$p$-value: normal          & [0.22] & [0.48] & [0.80] & [0.28] & [0.07] & [0.00] \\
\hspace{1.5em}$p$-value: bootstrap   & [0.34] & [0.61] & [0.87] & [0.62] & [0.16] & [0.00] \\
$\widehat{\mathrm{RU}}_j$                           & 0.01 & 0.04 & 0.18 & 1.53 & 0.19 & -- \\
\hline
\multicolumn{7}{c}{2004M2--2024M10\quad ($T=249$;\ rel.\ gaps $0.66,\ 0.58,\ 0.36,\ 0.11$)} \\[2pt]
LS estimate                              & 0.38 & 1.39 & 0.03 & 0.07 & 0.18 & 0.14 \\
\hspace{1.5em}$p$-value: normal          & [0.10] & [0.01] & [0.94] & [0.80] & [0.18] & [0.78] \\
\hspace{1.5em}$p$-value: bootstrap   & [0.23] & [0.06] & [0.90] & [0.86] & [0.19] & [0.75] \\[2pt]
Analytically bias-corrected estimate     & 0.38 & 1.39 & 0.05 & 0.07 & 0.20 & 0.13 \\
\hspace{1.5em}$p$-value: normal          & [0.10] & [0.01] & [0.88] & [0.80] & [0.14] & [0.78] \\
\hspace{1.5em}$p$-value: bootstrap   & [0.23] & [0.06] & [0.89] & [0.86] & [0.20] & [0.76] \\[2pt]
Bootstrap bias-corrected estimate        & 0.37 & 1.39 & 0.04 & 0.07 & 0.20 & 0.15 \\
\hspace{1.5em}$p$-value: normal          & [0.10] & [0.01] & [0.90] & [0.79] & [0.14] & [0.77] \\
\hspace{1.5em}$p$-value: bootstrap   & [0.24] & [0.07] & [0.90] & [0.88] & [0.22] & [0.74] \\
$\widehat{\mathrm{RU}}_j$                           & 0.04 & 0.01 & 0.09 & 0.22 & 0.20 & -- \\
\hline\hline
\end{tabular}
}
\begin{minipage}{0.98\textwidth}
\footnotesize
\noindent Notes: Point estimates and two-sided $p$-values are reported for the fixed population parameter $\bdelta^0$, with PC signs normalized so that the uncorrected LS factor coefficients are positive and the same orientation retained for the bias-corrected estimates.
Normal $p$-values use the standard normal approximation, while bootstrap $p$-values are equal-tail $p$-values from the corresponding recentered studentized bootstrap distribution.
Inference uses the prewhitened Newey--West standard errors and dependent resampling described above.
Rel.\ gaps are $1-\hat\lambda_{k+1}/\hat\lambda_k$, $k=1,\dots,4$, and $\widehat{\mathrm{RU}}_j$ is defined in \eqref{eq:Shat}.
\end{minipage}
\end{table}

Table~\ref{tab:emp12} shows that bootstrap inference gives a substantially more parsimonious picture of bond-return predictability.
In the pre-2004 subsample, CP is the only predictor significant even at the 10\% level under bootstrap inference, while none of the five macro factors is significant.
The contrast with the normal approximation is most visible for the fifth factor, whose analytically bias-corrected normal and bootstrap $p$-values are $0.05$ and $0.12$, respectively, while the bootstrap bias-corrected estimator gives $p=0.16$.

The pre-2004 fourth factor illustrates rotation uncertainty: its relative eigengap is only $0.04$ and $\widehat{\mathrm{RU}}_4=1.53$, substantially larger than for the other factors.
The analytical correction also changes its coefficient from $0.28$ to $-0.28$ under the common PC orientation, revealing a large correction that would be obscured by renormalizing the sign of each corrected estimate separately.\footnote{The corresponding coefficient is reported as positive in \citet{jiang2024Mw} because of the different sign normalization used therein.}
Despite this instability, the fourth factor is clearly insignificant under bootstrap inference.

In the post-2004 subsample, the second macro factor is the only predictor significant at the 10\% level under bootstrap inference, with $p$-values of $0.06$--$0.07$, compared with about $0.01$ under the normal approximation.
Because $\widehat{\mathrm{RU}}_2=0.01$, this discrepancy reflects finite-sample distributional approximation rather than rotation uncertainty, an interpretation reinforced by the block-length sensitivity analysis in Section~\ref{app:empblock}.
At the two-year horizon, Section~\ref{app:emp24} reports no coefficient significant at conventional levels under bootstrap inference.
Overall, dependence-aware bootstrap inference confines the robust pre-2004 evidence to the CP factor and the post-2004 evidence to weak significance of the second macro factor at the 10\% level.

\section{Conclusion}\label{sec:con}

We have developed bootstrap inference for factor-augmented regressions under weak factor models with possibly heterogeneous signal rates $N^{\alpha_k}$, $0<\alpha_k\leq1$.
Because the least squares estimator targets a rotation of the latent parameter vector, we establish bootstrap validity for the conventional data-dependent rotation $\hat{\bH}$, the alternative $\hat{\bH}_q$ with smaller asymptotic bias under weaker rate conditions, and the unique (up to sign) signal-dependent rotation $\bH$ defining the fixed parameter $\bdelta^0$.
We construct the bootstrap analogue of the rotation-specific target directly, avoiding the additional randomness introduced by the pre-rotation step in existing procedures.

Under weak factors, the studentized least squares statistic can have nonzero first-order bias, so we consider analytical and bootstrap bias corrections for all three targets and bootstrap the corresponding studentized statistics.
The proposed bootstrap correction avoids direct estimation of the analytical bias, eigengap quantities, and divergence rates $\alpha_k$.
After first-order bias correction, the studentized statistics are asymptotically pivotal with centered Gaussian limits, so conventional critical values are first-order valid. 
We also establish bootstrap validity for the corresponding bias-corrected statistics, so the bootstrap consistently approximates their distributions and thereby provides scope for refinement beyond first-order asymptotic inference.

The experiments show that the proposed correction generally reduces the pronounced bias associated with $\hat{\bH}$ more accurately than the existing pre-rotation correction with virtually identical RMSE, and that bootstrap distributional approximation can substantially improve target-relative inference.
They also show that accurate target-relative inference need not imply accurate inference on an individual fixed population coefficient in small samples. However, under the fixed target, the proposed bootstrap controls these distortions effectively as $T$ increases, whereas sample-dependent targets naturally reflect ongoing rotation uncertainty. 
To assess this distinction, we introduce the bootstrap rotation-uncertainty diagnostic $\widehat{\mathrm{RU}}_j$, which quantifies uncertainty in the $j$-th factor coefficient arising from close but distinct signal eigenvalues.
In the bond-return application, dependent resampling accommodates serial and cross-sectional dependence, and the bootstrap procedures together with $\widehat{\mathrm{RU}}_j$ distinguish findings sensitive to rotation uncertainty from those sensitive to finite-sample distributional approximation.

Two extensions are particularly relevant.
First, extending the bootstrap prediction intervals of \citet{GoncalvesPerronDjogbenou2017} to our weak-factor, rotation-specific setting would incorporate both bias and rotation uncertainty in out-of-sample forecasting, with our dependent resampling scheme providing a natural basis for serial and cross-sectional dependence.
Second, $\widehat{\mathrm{RU}}_j$ is used as a diagnostic rather than a pretest.
Under the maintained eigengap conditions, $\widehat{\mathrm{RU}}_j=o_p(1)$; a formal analysis under shrinking eigengaps is pursued separately.

\section*{Acknowledgment}
We are grateful to Yoshimasa Uematsu and Yukitoshi Matsushita for helpful discussions and useful comments. 

\section*{Funding}
This work was supported by JSPS KAKENHI (grant numbers 23H00804, 24K16343, 25K00625, 25K05036 and 25H00544).

\section*{Generative AI disclosure}
Generative AI tools, including ChatGPT (OpenAI) and Claude (Anthropic), were used for language editing and coding assistance in the Monte Carlo and empirical analyses.
The authors conducted all research design, methodological development, analysis, interpretation, and verification and take full responsibility for the paper and accompanying code.

\section*{Disclosure Statement}

The authors do not have conflicts of interest with regard to this work.

\bibliographystyle{agsm}
\setlength{\bibsep}{0pt}
\bibliography{references_wfr}

\newpage
\appendix

\setcounter{page}{1}
\hypersetup{pageanchor=false}
\setcounter{section}{0}

\renewcommand{\thesection}{\Alph{section}}
\renewcommand{\thesubsection}{\thesection.\arabic{subsection}}

\counterwithin{equation}{section}
\counterwithin{table}{section}
\counterwithin{figure}{section}

\counterwithin{thm}{section}
\counterwithin{lem}{section}
\counterwithin{cor}{section}
\counterwithin{prop}{section}
\counterwithin{rem}{section}
\counterwithin{defi}{section}
\counterwithin{con}{section}
\counterwithin{proc}{section}

\counterwithin{ass}{section}
\counterwithin{assb}{section}

\begin{center}
	{\Large \bf Supplementary Material for \\[7mm]
	{\LARGE \bf Bootstrap inference in factor-augmented regressions with weak factors}} \\[10mm]

  \textsc{\large Peiyun Jiang$^*$,} 
	\textsc{\large Takashi Yamagata$^\dagger$} \\[5mm]
 $^*$\textit{\large Center for Data Science, Waseda University} \\[1mm]
  $\dagger$\textit{\large Department of Economics and Related Studies, University of York} \\[1mm]
   $\dagger$\textit{\large Graduate School of Economics and Management, Tohoku University} \\[1mm]
   
\end{center}

\section{ Auxiliary results and proofs for the bootstrap theory}
\label{sec:bootstrap-proofs}

 Section~\ref{app:rotation-specific-bootstrap} states the rotation-specific bootstrap limits (Theorems~\ref{thm:bias_Hhat}--\ref{thm:bias_H} and Proposition~\ref{prop:Q-I}) and derives Theorem~\ref{thm:bootstrap-all-rotations} from them. Section~\ref{app:bc-proofs} proves Theorems~\ref{cor:abc}, \ref{thm:abc-bootstrap}, \ref{thm:bc-all-rotations} and \ref{thm:db-all-rotations} for the bias-corrected estimators. Section~\ref{app:rotation-specific-proofs} proves the rotation-specific limits, and the technical lemmas invoked there are collected in Section~\ref{app:aux-lemmas}, with the lemmas taken from \citet{jiang2023revisiting} and \citet{jiang2024Mw} grouped at its end.
\subsection*{Original-sample assumptions}
\phantomsection\label{app:original-assumptions}
Assumptions~\ref{ass:eigen}--\ref{ass:signal} are retained in Section~\ref{sec:4}. The remaining original-sample Assumptions~\ref{ass:errors}--\ref{ass:Aug_errors} and their discussion are collected below.

\begingroup
\renewcommand{\theass}{\arabic{ass}}
\setcounter{ass}{2}

\begin{ass}[Idiosyncratic errors]\normalfont\label{ass:errors} ${ }^{ }$\\
(i) $\E[e_{t,i}]=0$ and $\E[e_{t,i}^4]\le M$ for all $i$ and $t$;\\
(ii) For all $i$, $\left|\E[e_{s,i} e_{t,i}]\right| \leq\left|\gamma_{s, t}\right|$ for some $\gamma_{s, t}$ such that $\sum_{t=1}^T\left|\gamma_{ s, t}\right| \leq M$;\\
(iii)  For all $t$, $\left|\E[e_{t,i} e_{t,j}]\right| \leq\left|\tau_{i, j}\right|$ for some $\tau_{i, j}$ such that $\sum_{j=1}^N\left|\tau_{ i,j}\right| \leq M$;\\
(iv) $\left\|\bE \right\|_2^2=O_p(\max{\{N,T}\})$.

\end{ass}

As discussed earlier, the PC estimators $(\hat{\bF},\hat{\bB})$ are viewed as estimators of the pseudo-true parameters $({\bF^0},{\bB^0})$. Accordingly, we impose the following assumptions directly on them.

\begin{ass}[Factors and Loadings]\normalfont\label{ass:factor and loadings}
Denote $\bz^0_t:=(\bff^{0\prime}_t, \bw_t')'$. \\
(i) $\E\|{\bz_t^0}\|_2^4 \le M$ and $\E\|\bb_i^0\|_2^4 \le M$;\\ 
(ii) $\E\| \bN^{-\frac{1}{2}}\sum_{i=1}^{N} \bb^0_ie_{t,i}   \|_2^2 \le M$ for each $t$; \\
(iii) $\E\| T^{-\frac{1}{2}}\sum_{t=1}^{T}{\bz_t^0}e_{t,i}\|_2^2 \le M$ for each $i$; \\
(iv) $\E\| T^{-\frac{1}{2}}\bN^{-\frac{1}{2}} \sum_{t=1}^{T}\sum_{i=1}^{N}\bb_i^{0}e_{t,i}{\bz_t^{0\prime}} \|_{{2}}^2 \le M$;\\
(v) 
$T^{-1} \sum_{t=1}^{T}  ( \bN^{-\frac{1}{2}}\sum_{i=1}^{N} \bb_i^0e_{t,i})  ( \bN^{-\frac{1}{2}}\sum_{i=1}^{N} \bb_i^0e_{t,i})'\CP \bGamma  $,\\
where $\bGamma := \lim_{N, T \rightarrow \infty} T^{-1} \sum_{t=1}^T \bGamma_t>0$, and $\bGamma_t :=   \Var(\bN^{-\frac{1}{2}}\sum_{i=1}^{N} \bb_i^0e_{t,i})$.
\end{ass}
The moment restrictions in Assumption 4 (iii), (iv) %
are essentially similar to Assumptions D, F2 %
in \cite{Bai2003}, and Assumption 4 (ii) %
is similar moment restriction related for $\bb_i^0$.
Assumption (v) is similar to Assumption 3(e) in \cite{GoncalvesPerron2014}. 

Now we impose assumptions on the pseudo-true augmented model \eqref{pseudomodel}:

\begin{ass}[Weak dependence between idiosyncratic errors and regression errors]\normalfont\label{ass:2errors} 
$\E||T^{-\frac{1}{2}}\bN^{-\frac{1}{2}} \sum_{t=1}^{T} \sum_{i=1}^N \bb_i^0 e_{t,i} \epsilon_{t+h}||_2^2 \leq M$.
\end{ass}

\begin{ass}[Moments, parameters and CLT]%
\normalfont\label{ass:Aug_errors}${ }^{ }$\\
(i) $\E[\epsilon_{t+h}] =0$ and $\E|\epsilon_{t+h}|^2 < M $;\\
(ii) $||\bdelta^0||_{2}\leq M$ and $\bH\CP \bH_0$ which is fixed and invertible;\\
(iii) 
$T^{-1/2} {\bZ}^{0\prime}\bepsilon \CD N(\mathbf{0},\bSigma_{\bZ^0 \bepsilon})$, $T^{-1} {\bZ}^{0\prime}\bZ^0 \CP \bSigma_{\bZ^0 \bZ^0}$, where $\bSigma_{\bZ^0\bepsilon}$ and $\bSigma_{\bZ^0 \bZ^0}$ are fixed, positive definite and bounded;\\
(iv) there exists an estimator $\hat{\bSigma}_{\bZ^0\bepsilon}$ of $\bSigma_{\bZ^0\bepsilon}$, a function of $\hat{\bZ}$ and $\hat{\bepsilon}$, such that $\hat{\bSigma}_{\bZ^0\bepsilon}-\bSigma_{\bZ^0\bepsilon}=o_p(1)$.
\end{ass}

Assumptions \ref{ass:2errors} and \ref{ass:Aug_errors} are similar to Assumption 4 in \cite{GoncalvesPerron2014} and Assumption E in \cite{BaiNg2006}, respectively. Under Assumption \ref{ass:eigen}, $\bH$ is bounded in probability. 
Assumption \ref{ass:Aug_errors}(ii) further guarantees that its probability limit exists and coincides with that of other four data-dependent rotation matrices considered in \citet{jiang2023revisiting}.
\endgroup

\subsubsection*{Second-level bootstrap construction}
\phantomsection\label{app:second-level-construction}
The detailed construction and alternative targets for Section~\ref{subsec:double-bootstrap} are reproduced here.
\begingroup
Conditional on the $b$-th first-level bootstrap sample, define the second-level bootstrap population by $\bF^{0\ddag(b)}:=\hat{\bF}^{\dag(b)}$, $\bB^{0\ddag(b)}:=\hat{\bB}^{\dag(b)}$, and $\bdelta^{0\ddag(b)}:=\hat{\bdelta}^{\dag(b)}$.
For $\ell=1,\dots,B_2$, resample $\bE^{\ddag(b,\ell)}$ from $\hat{\bE}^{\dag(b)}$ and $\bepsilon^{\ddag(b,\ell)}$ from $\hat{\bepsilon}^{\dag(b)}$ using the same resampling scheme as at the first level, and generate the second-level bootstrap sample $\bX^{\ddag(b,\ell)} = \bF^{0\ddag(b)} \bB^{0\ddag(b)\prime} + \bE^{\ddag(b,\ell)}, \quad \by^{\ddag(b,\ell)} = \bZ^{0\ddag(b)} \bdelta^{0\ddag(b)} + \bepsilon^{\ddag(b,\ell)}, \qquad \bZ^{0\ddag(b)} = (\bF^{0\ddag(b)},\bW).$ From $\bX^{\ddag(b,\ell)}$, compute the second-level PC estimators $(\hat{\bF}^{\ddag(b,\ell)},\hat{\bB}^{\ddag(b,\ell)})$, with the same sign normalization as at the first level but relative to $(\bF^{0\ddag(b)},\bB^{0\ddag(b)})$, and construct the second-level rotation
\begin{align}\label{eq:Hhat-ddag}
\tilde{\bH}^{\ddag(b,\ell)}
&=
\bB^{0\ddag(b)\prime}\bB^{0\ddag(b)}
(
T^{-1}\bF^{0\ddag(b)\prime}
\hat{\bF}^{\ddag(b,\ell)}
)
(
\hat{\bB}^{\ddag(b,\ell)\prime}
\hat{\bB}^{\ddag(b,\ell)}
)^{-1}.
\end{align}
The corresponding second-level rotation-dependent parameter is 
$\bdelta_{\hat{\bH}^{\ddag(b,\ell)}}
=
\bPhi_{\tilde{\bH}^{\ddag(b,\ell)}}^{-1}
\bdelta^{0\ddag(b)}$. 
Using $\hat{\bZ}^{\ddag(b,\ell)}=(\hat{\bF}^{\ddag(b,\ell)},\bW)$, compute $\hat{\bdelta}^{\ddag(b,\ell)}=(\hat{\bZ}^{\ddag(b,\ell)\prime}\hat{\bZ}^{\ddag(b,\ell)})^{-1}\hat{\bZ}^{\ddag(b,\ell)\prime}\by^{\ddag(b,\ell)}$.

\begin{rem}\label{rem:other-targets}
The same construction applies to the other two targets. 
For $\hat{\bH}_q$, the first-level target is $\bdelta_{\hat{\bH}_q^{\dag(b)}}=\bPhi_{\tilde{\bH}_q^{\dag(b)}}^{-1}\bdelta^{0\dag}$, where $\tilde{\bH}_q^{\dag(b)}=(T^{-1}\hat{\bF}^{\dag(b)\prime}\bF^{0\dag})^{-1}$. 
At the second level, replace \eqref{eq:Hhat-ddag} by $\tilde{\bH}_q^{\ddag(b,\ell)}=(T^{-1}\hat{\bF}^{\ddag(b,\ell)\prime}\bF^{0\ddag(b)})^{-1}$ and use $\bdelta_{\hat{\bH}_q^{\ddag(b,\ell)}}=\bPhi_{\tilde{\bH}_q^{\ddag(b,\ell)}}^{-1}\bdelta^{0\ddag(b)}$ in the second-level bias estimate. 
For the fixed target $\bdelta^0$ associated with $\bH$, no rotation needs to be constructed at either bootstrap level. 
The first-level and second-level targets are $\bdelta^{0\dag}=\hat{\bdelta}$ and $\bdelta^{0\ddag(b)}=\hat{\bdelta}^{\dag(b)}$, respectively, so that $\hat{\bb}_{\bH}^{\ddag(b)}=B_2^{-1}\sum_{\ell=1}^{B_2}(\hat{\bdelta}^{\ddag(b,\ell)}-\bdelta^{0\ddag(b)})$.
\end{rem}

\endgroup

\begin{rem}\label{rem:no-new-estimator}
The estimator whose distribution is approximated is the bootstrap bias-corrected estimator $\hat{\bdelta}_\bR^{bc}$ of Section~\ref{subsec:bc-estimators}. 
The second-level bootstrap defines no new original-sample estimator; it only constructs, in each first-level world, the analogue $\hat{\bdelta}_\bR^{\dag,bc}$ of the same estimator. 
Moreover, the second level is used only to estimate the conditional mean defining the bias correction, rather than to approximate a conditional distribution or quantile. 
This makes a relatively small number of second-level replications natural in implementation. 
A related result of \citet{ChangHall2015} shows, for smooth-function bootstrap bias correction, that the number of second-level draws can remain fixed without changing the order of the Monte Carlo error. 
Although their result does not directly apply to our weak-factor setting, it supports the use of substantially fewer second-level than first-level replications. 
\end{rem}

\begin{rem}[Bootstrap stochastic orders]\label{rem:bootstrap-orders}
All statements below hold conditionally on the original sample, in probability. For bootstrap statistics $X^{\dag}$: $X^{\dag}=o_{p^{\dag}}(1)$ if $\Pro^{\dag}(\|X^{\dag}\|_2>\varepsilon)\CP0$ for every $\varepsilon>0$; $X^{\dag}=O_{p^{\dag}}(1)$ if for every $\varepsilon>0$ there is an $M_\varepsilon<\infty$ with $\Pro[\Pro^{\dag}(\|X^{\dag}\|_2>M_\varepsilon)>\varepsilon]\to0$; $X^{\dag}\CPb c$ if $X^{\dag}-c=o_{p^{\dag}}(1)$; and $X^{\dag}\CDb X$ with law $\mathcal{L}$ if $d_{\mathrm{BL}}(\mathcal{L}^{\dag}(X^{\dag}),\mathcal{L})\CP0$, where $\mathcal{L}^{\dag}(X^{\dag})$ is the conditional law of $X^{\dag}$ and $d_{\mathrm{BL}}$ the bounded Lipschitz metric \citep[cf.][]{GoncalvesPerron2014}; for continuous $\mathcal{L}$ the last is equivalent to uniform convergence of the conditional distribution function. Functions of the original sample alone ($\E^{\dag}$, $\Var^{\dag}$, $\bGamma^{\dag}$, $\bD^{\dag}$) carry the orders $o_p$, $O_p$, $\CP$; quantities varying across bootstrap draws ($\hat{\bdelta}^{\dag}$, $\bM^{\dag}$) carry $o_{p^{\dag}}$, $O_{p^{\dag}}$, $\CPb$.
\end{rem}

\subsection{Rotation-specific bootstrap limits and proof of Theorem~\ref{thm:bootstrap-all-rotations}}
\label{app:rotation-specific-bootstrap}

This subsection proves Theorem~\ref{thm:bootstrap-all-rotations}. 
We first establish the rotation-specific bootstrap limits for $\hat{\bH}$, $\hat{\bH}_q$, and $\bH$, whose first-order bias components and rate conditions differ. 
These results are then combined with the corresponding original-sample asymptotic distributions to establish the unified bootstrap validity result.

\begin{thm}\label{thm:bias_Hhat}
Suppose Assumptions \ref{ass:eigen}--\ref{ass:Aug_errors} and \ref{ass:errors_b}--\ref{ass:bootstrap_b} hold.
If $\alpha_r>1/2$, $N^{1-\alpha_r}/\sqrt{T}\to0$, and $\sqrt{T}N^{\frac12\alpha_1-\frac32\alpha_r}\to c_1\in[0,\infty)$, then
\[
\sqrt{T}\left(
\hat{\bdelta}^{\dag}
-
\bdelta_{\hat{\bH}^{\dag}}
\right)
\CDb
N\left(
c_1\bkappa_{\bdelta^*},
\bSigma_{\bdelta}
\right)
\]
in probability, where
\[
\bkappa_{\bdelta^*}
:=
-\bSigma_{\bZ^0\bZ^0}^{-1}
\binom{
\bG+\nu\bD^{-1}\bGamma\bD^{-1}
}{
\bSigma_{\bW\bF^0}\bG
}
\bH_0^{-1}\bgamma^*,
\]
$c_1\bG:=\lim_{N,T\to\infty}\sqrt{T}\bN^{1/2}\bGamma\bD^{-2}\bN^{-3/2}$, $\nu:=\lim_{N\to\infty}N^{-\frac12(\alpha_1-\alpha_r)}$, and $\bSigma_{\bW\bF^0}:=\plim T^{-1}\bW'\bF^0$.
\end{thm}

\begin{thm}\label{thm:bias_H3}
Suppose Assumptions \ref{ass:eigen}--\ref{ass:Aug_errors} and \ref{ass:errors_b}--\ref{ass:bootstrap_b} hold.
If $\alpha_r>1/2$, $N^{1-\alpha_r}/\sqrt{T}\to0$, and $\sqrt{T}N^{-\alpha_r}\to c_2\in[0,\infty)$, then
\[
\sqrt{T}\left(
\hat{\bdelta}^{\dag}
-
\bdelta_{\hat{\bH}_q^{\dag}}
\right)
\CDb
N\left(
c_2\bar{\bkappa}_{\bdelta^*},
\bSigma_{\bdelta}
\right)
\]
in probability, where
\[
\bar{\bkappa}_{\bdelta^*}
:=
\bSigma_{\bZ^0\bZ^0}^{-1}
\binom{
\mathbf{0}
}{
\bSigma_{\bW\bF^0}\bar{\bG}
}
\bH_0^{-1}\bgamma^*,
\]
and $c_2\bar{\bG}:=\lim_{N,T\to\infty}\sqrt{T}\bN^{-1/2}\bD^{-1}\bGamma\bD^{-1}\bN^{-1/2}$.
\end{thm}

The bias under $\hat{\bH}_q$ arises at rate $\sqrt{T}N^{-\alpha_r}$ and is generally smaller than that under $\hat{\bH}$.
Moreover, it vanishes when $\bF^0$ and $\bW$ are asymptotically uncorrelated.

For the fixed target, decompose
\[
\sqrt{T}\left(
\hat{\bdelta}^{\dag}
-
\bdelta^{0\dag}
\right)
=
\sqrt{T}\left(
\hat{\bdelta}^{\dag}
-
\bdelta_{\hat{\bH}_q^{\dag}}
\right)
+
\sqrt{T}\left(
\bdelta_{\hat{\bH}_q^{\dag}}
-
\bdelta^{0\dag}
\right).
\]
The first term is covered by Theorem~\ref{thm:bias_H3}, while the following proposition characterizes the additional bootstrap rotation-gap component.

\begin{prop}\label{prop:Q-I}
Suppose Assumptions \ref{ass:eigen}--\ref{ass:Aug_errors} and \ref{ass:errors_b}--\ref{ass:bootstrap_b} hold, with $\alpha_r>1/2$, $N^{1-\alpha_r}/\sqrt{T}\to0$, and $\sqrt{T}N^{-\alpha_r}\to c_2\in[0,\infty)$.
Then,
\[
\sqrt{T}\left(
\bdelta_{\hat{\bH}_q^{\dag}}
-
\bdelta^{0\dag}
\right)
\CPb
c_2\bh_{\bgamma^*}
\]
in probability, where
\[
\bh_{\bgamma^*}
:=
\binom{
\frac{1}{2}
\left(
\bM'-\bM-\mathbf{1}_{r\times r}
\right)
\odot
\bar{\bG}
\bH_0^{-1}\bgamma^*
}{
\mathbf{0}_p
},
\]
$\bM=(m_{kl})$, $m_{kl}=\plim\{\lambda_k/(\hat{\lambda}_l-\lambda_k)\}$ for $k\neq l$, and $m_{kk}=0$.
\end{prop}

\begin{thm}\label{thm:bias_H}
Suppose Assumptions \ref{ass:eigen}--\ref{ass:Aug_errors} and \ref{ass:errors_b}--\ref{ass:bootstrap_b} hold, with $\alpha_r>1/2$, $N^{1-\alpha_r}/\sqrt{T}\to0$, and $\sqrt{T}N^{-\alpha_r}\to c_2\in[0,\infty)$.
Then,
\[
\sqrt{T}\left(
\hat{\bdelta}^{\dag}
-
\bdelta^{0\dag}
\right)
\CDb
N\left(
c_2\bh_{\bgamma^*}
+
c_2\bar{\bkappa}_{\bdelta^*},
\bSigma_{\bdelta}
\right)
\]
in probability.
\end{thm}

\begin{proof}[Proof of Theorem \ref{thm:bootstrap-all-rotations}]

For $\bR=\hat{\bH}$, $\hat{\bH}_q$, and $\bH$, respectively,
Theorems~\ref{thm:bias_Hhat}, \ref{thm:bias_H3}, and
\ref{thm:bias_H} give, under the corresponding rate conditions,
\[
\sqrt{T}\bigl(\hat{\bdelta}^{\dag}-\bdelta_{\bR}^{\dag}\bigr)
\CDb
N(\bmu_{\bR},\bSigma_{\bdelta})
\]
in probability.
The limiting means and covariance matrices in these
rotation-specific results are expressed in the original-sample
quantities by Assumption~\ref{ass:bootstrap_b} and \eqref{eq:D-dag-cons}--\eqref{eq:M-dag-cons}.

By Theorem~\ref{thm:asymptotic-all-rotations}, under the same
corresponding rate conditions,
\[
\sqrt{T}\bigl(\hat{\bdelta}-\bdelta_{\bR}\bigr)
\CD
N(\bmu_{\bR},\bSigma_{\bdelta}).
\]

Let $G_{\bR}$ denote the distribution function of
$N(\bmu_{\bR},\bSigma_{\bdelta})$. Since $\bSigma_{\bdelta}$
is positive definite, $G_{\bR}$ is continuous. Hence,
by the conditional multivariate P\'olya theorem,
\[
\sup_{\bx\in\mathbb R^{r+p}}
\left|
\Pro^\dag\left[
\sqrt{T}\bigl(\hat{\bdelta}^{\dag}
-\bdelta_{\bR}^{\dag}\bigr)\le\bx
\right]
-G_{\bR}(\bx)
\right|
=o_p(1),
\]
while the multivariate P\'olya theorem gives
\[
\sup_{\bx\in\mathbb R^{r+p}}
\left|
\Pro\left[
\sqrt{T}\bigl(\hat{\bdelta}-\bdelta_{\bR}\bigr)\le\bx
\right]
-G_{\bR}(\bx)
\right|
=o(1).
\]
The triangle inequality therefore yields
\eqref{eq:bootstrap-valid-all-R}.
\end{proof}

\subsection{Proofs for the bias-corrected bootstrap inference}
\label{app:bc-proofs}

\begin{rem}[Normal and bootstrap approximation after bias correction]
\label{rem:biasmap}
After bias correction and studentization, the estimation error is asymptotically pivotal.
The standard normal distribution therefore provides a first-order approximation to its distribution.
We additionally approximate this distribution by the bootstrap, with the aim of improving finite-sample accuracy.
Classical bootstrap theory shows that bootstrapping an asymptotically pivotal statistic can yield a higher-order refinement over direct use of the limiting distribution; see, for example, \citet{Beran1988}.
We do not establish such a higher-order result in the present weak-factor setting.
The same argument applies whether the underlying estimator is analytically or bootstrap bias corrected.
\end{rem}

This subsection proves the results for the analytically and bootstrap bias-corrected estimators, Theorems~\ref{cor:abc}, \ref{thm:abc-bootstrap}, \ref{thm:bc-all-rotations} and \ref{thm:db-all-rotations}, building on Theorems~\ref{thm:asymptotic-all-rotations} and \ref{thm:bootstrap-all-rotations} and the rotation-specific limits of Section~\ref{app:rotation-specific-bootstrap}.

\begin{proof}[Proof of Theorem \ref{cor:abc}]
Write 
$\sqrt{T}(\hat{\bdelta}_{\bR}^{abc}-\bdelta_{\bR})
=
\sqrt{T}(\hat{\bdelta}-\bdelta_{\bR})
-
\sqrt{T}\hat{\bb}_{\bR}^{a}$.
The first term converges to $N(\bmu_{\bR},\bSigma_{\bdelta})$ under the conditions of Theorem~\ref{thm:asymptotic-all-rotations}, while $\sqrt{T}\hat{\bb}_{\bR}^{a}\CP\bmu_{\bR}$.
The result follows by Slutsky's theorem.
\end{proof}

\begin{proof}[Proof of Theorem \ref{thm:abc-bootstrap}]
Apply the same decomposition in the bootstrap world:
$\sqrt{T}(\hat{\bdelta}_{\bR}^{abc\dag}-\bdelta_{\bR}^{\dag})
=
\sqrt{T}(\hat{\bdelta}^{\dag}-\bdelta_{\bR}^{\dag})
-
\sqrt{T}\hat{\bb}_{\bR}^{a\dag}$.
By Theorem~\ref{thm:bootstrap-all-rotations},
$\sqrt{T}(\hat{\bdelta}^{\dag}-\bdelta_{\bR}^{\dag})\CDb N(\bmu_{\bR},\bSigma_{\bdelta})$
in probability, while
$\sqrt{T}\hat{\bb}_{\bR}^{a\dag}\CPb\bmu_{\bR}$
in probability.

The first claim follows by the conditional version of Slutsky's theorem.
For the second claim, let $G_{\bdelta}(\bx)$ denote the distribution function of $N(\bzero,\bSigma_{\bdelta})$.
By Theorem~\ref{cor:abc} and the multivariate version of P\'olya's theorem,
\[
\sup_{\bx\in\bbR^{r+p}}
\left|
\Pro\!\left[\sqrt{T}\left(\hat{\bdelta}_{\bR}^{abc}-\bdelta_{\bR}\right)\leq\bx\right]
-G_{\bdelta}(\bx)
\right|
\longrightarrow0,
\]
while the first claim and the conditional version of P\'olya's theorem give
\[
\sup_{\bx\in\bbR^{r+p}}
\left|
\Pro^{\dag}\!\left[\sqrt{T}\left(\hat{\bdelta}_{\bR}^{abc\dag}-\bdelta_{\bR}^{\dag}\right)\leq\bx\right]
-G_{\bdelta}(\bx)
\right|
=o_p(1).
\]
The triangle inequality then yields \eqref{eq:abc-bootstrap-valid}, exactly as in the proof of Theorem~\ref{thm:bootstrap-all-rotations}.
\end{proof}

\begin{proof}[Proof of Theorem \ref{thm:bc-all-rotations}]

Fix $\bR\in\cR$. With the population bias $\bb_{\bR}^{\dag}$ in place of its Monte Carlo average $\hat{\bb}_{\bR}^{\dag}$, the bootstrap bias-corrected estimator is
\[
\hat{\bdelta}_{\bR}^{bc}
=
\hat{\bdelta}
-
\bb_{\bR}^{\dag}.
\]
Therefore,
\begin{align}
\sqrt{T}
\left(
\hat{\bdelta}_{\bR}^{bc}
-
\bdelta_{\bR}
\right)
&=
\sqrt{T}
\left(
\hat{\bdelta}
-
\bdelta_{\bR}
\right)
-
\sqrt{T}\bb_{\bR}^{\dag}.
\label{eq:proof-bc-decomp}
\end{align}

We consider the two terms on the right-hand side separately.

First, under the rate conditions of the theorem, the corresponding
non-bootstrap asymptotic result of Theorem~\ref{thm:asymptotic-all-rotations} for each rotation target gives
\begin{align}
\sqrt{T}
\left(
\hat{\bdelta}
-
\bdelta_{\bR}
\right)
\CD
N
\left(
\bmu_{\bR},
\bSigma_{\bdelta}
\right),
\label{eq:proof-bc-uncorrected}
\end{align}
where $\bmu_{\bR}$ is defined in \eqref{eq:muR}.
Thus, the only rotation-specific feature of the first-order limiting
distribution relevant here is its asymptotic bias $\bmu_{\bR}$; the covariance
matrix is $\bSigma_{\bdelta}$ for all three targets.

Second, Assumption~\ref{ass:bc-highlevel}(i) gives
\begin{align}
\sqrt{T}\bb_{\bR}^{\dag}
\CP
\bmu_{\bR}.
\label{eq:proof-bc-bias}
\end{align}
Hence the bootstrap bias estimate consistently estimates exactly the asymptotic bias term appearing in
\eqref{eq:proof-bc-uncorrected}.

Combining \eqref{eq:proof-bc-decomp}--\eqref{eq:proof-bc-bias} and
applying Slutsky's theorem yields
\begin{align}
\sqrt{T}
\left(
\hat{\bdelta}_{\bR}^{bc}
-
\bdelta_{\bR}
\right)
\CD
N
\left(
\boldsymbol{0},
\bSigma_{\bdelta}
\right).
\label{eq:proof-bc-centered}
\end{align}
In particular, subtracting the estimated bias removes the asymptotic bias but does not alter the asymptotic covariance matrix.

It remains to account for studentization. By Assumption~\ref{ass:Aug_errors}(iv) and $T^{-1}\hat{\bZ}'\hat{\bZ}\CP\bSigma_{\bZ^0\bZ^0}$ \citep[proof of Theorem 1]{jiang2024Mw},
\[
\hat{\bSigma}_{\bdelta}
\CP
\bSigma_{\bdelta}.
\]
Since $\bSigma_{\bdelta}$ is positive definite under the maintained
assumptions, continuity of the inverse square-root map gives
\[
\hat{\bSigma}_{\bdelta}^{-1/2}
\CP
\bSigma_{\bdelta}^{-1/2}.
\]
Applying Slutsky's theorem once more to
\eqref{eq:proof-bc-centered},
\begin{align*}
\hat{\bSigma}_{\bdelta}^{-1/2}
\sqrt{T}
\left(
\hat{\bdelta}_{\bR}^{bc}
-
\bdelta_{\bR}
\right)
&\CD
N
\left(
\boldsymbol{0},
\bSigma_{\bdelta}^{-1/2}
\bSigma_{\bdelta}
\bSigma_{\bdelta}^{-1/2}
\right)\\
&=
N
\left(
\boldsymbol{0},
\bI_{r+p}
\right).
\end{align*}

The argument applies to each
$\bR\in\cR=\{\hat{\bH},\hat{\bH}_q,\bH\}$ under its corresponding
rate conditions. Since $\cR$ is finite, all of the probability
convergences above hold simultaneously over $\bR\in\cR$; for example,
if $a_{\bR,NT}=o_p(1)$ for every $\bR\in\cR$, then
\[
\max_{\bR\in\cR}
\|a_{\bR,NT}\|
=
o_p(1).
\]
This establishes the stated result jointly over the finite set $\cR$.
\end{proof}

\begin{proof}[Proof of Theorem \ref{thm:db-all-rotations}]

Fix $\bR\in\cR$. With the second-level population bias $\bb_{\bR}^{\ddag}$ in place of its Monte Carlo average $\hat{\bb}_{\bR}^{\ddag(b)}$, the first-level bootstrap analogue of the bootstrap bias-corrected estimator is
\[
\hat{\bdelta}_{\bR}^{\dag,bc}
=
\hat{\bdelta}^{\dag}
-
\bb_{\bR}^{\ddag},
\]
where
\[
\bb_{\bR}^{\ddag}
=
\E^{\ddag}
\left(
\hat{\bdelta}^{\ddag}
-
\bdelta_{\bR}^{\ddag}
\right)
\]
is the second-level estimate of the bias within the first-level bootstrap
world. Therefore,
\begin{align}
\sqrt{T}
\left(
\hat{\bdelta}_{\bR}^{\dag,bc}
-
\bdelta_{\bR}^{\dag}
\right)
&=
\sqrt{T}
\left(
\hat{\bdelta}^{\dag}
-
\bdelta_{\bR}^{\dag}
\right)
-
\sqrt{T}\bb_{\bR}^{\ddag}.
\label{eq:proof-db-decomp}
\end{align}

We again analyze the two terms separately.

First, Theorem~\ref{thm:bootstrap-all-rotations} establishes, under the corresponding rate conditions, 
\begin{align}
\sqrt{T}
\left(
\hat{\bdelta}^{\dag}
-
\bdelta_{\bR}^{\dag}
\right)
\CDb
N
\left(
\bmu_{\bR},
\bSigma_{\bdelta}
\right),
\label{eq:proof-db-uncorrected}
\end{align}
in probability. Thus, conditional on the original sample, the
first-level bootstrap reproduces both the asymptotic bias 
$\bmu_{\bR}$ and the covariance matrix $\bSigma_{\bdelta}$ of the
uncorrected statistic.

Second, Assumption~\ref{ass:ddag} ensures that, conditional on the
first-level sample, the second-level bootstrap is constructed under the
same regularity conditions as the first-level bootstrap, with the
first-level quantities playing the role of the population quantities.
Assumption~\ref{ass:bc-highlevel}(ii) then gives
\begin{align}
\sqrt{T}\bb_{\bR}^{\ddag}
=
\bmu_{\bR}
+
o_{p^{\dag}}(1),
\label{eq:proof-db-bias}
\end{align}
in probability. Hence the inner bootstrap bias correction consistently
reproduces, within each first-level bootstrap world, the same
rotation-specific asymptotic bias term that appears in
\eqref{eq:proof-db-uncorrected}.

Substituting \eqref{eq:proof-db-bias} into
\eqref{eq:proof-db-decomp} gives
\begin{align*}
\sqrt{T}
\left(
\hat{\bdelta}_{\bR}^{\dag,bc}
-
\bdelta_{\bR}^{\dag}
\right)
&=
\left\{
\sqrt{T}
\left(
\hat{\bdelta}^{\dag}
-
\bdelta_{\bR}^{\dag}
\right)
-
\bmu_{\bR}
\right\}
+
o_{p^{\dag}}(1).
\end{align*}
By \eqref{eq:proof-db-uncorrected} and the conditional version of
Slutsky's theorem,
\begin{align}
\sqrt{T}
\left(
\hat{\bdelta}_{\bR}^{\dag,bc}
-
\bdelta_{\bR}^{\dag}
\right)
\CDb
N
\left(
\boldsymbol{0},
\bSigma_{\bdelta}
\right),
\label{eq:proof-db-centered}
\end{align}
in probability.

For the bootstrap studentization, let $\hat{\bSigma}_{\bZ^0\bepsilon}^{\dag}$ be the estimator of Assumption~\ref{ass:Aug_errors}(iv) computed from $(\hat{\bZ}^{\dag},\hat{\bepsilon}^{\dag})$. We use its bootstrap analogue, $\hat{\bSigma}_{\bZ^0\bepsilon}^{\dag}-\bSigma_{\bZ^0\bepsilon}^{\dag}=o_{p^{\dag}}(1)$ in probability, which, together with $T^{-1}\hat{\bZ}^{\dag\prime}\hat{\bZ}^{\dag}=T^{-1}\hat{\bZ}'\hat{\bZ}+o_{p^{\dag}}(1)$ in probability (Section~\ref{app:rotation-specific-proofs}), $T^{-1}\hat{\bZ}'\hat{\bZ}\CP\bSigma_{\bZ^0\bZ^0}$ and Assumption~\ref{ass:bootstrap_b}, gives
\[
\hat{\bSigma}_{\bdelta}^{\dag}
=
\bSigma_{\bdelta}
+
o_{p^{\dag}}(1),
\]
in probability. Since $\bSigma_{\bdelta}$ is positive definite,
\[
\left(
\hat{\bSigma}_{\bdelta}^{\dag}
\right)^{-1/2}
=
\bSigma_{\bdelta}^{-1/2}
+
o_{p^{\dag}}(1),
\]
in probability. Applying conditional Slutsky's theorem to
\eqref{eq:proof-db-centered} therefore yields
\begin{align}
\left(
\hat{\bSigma}_{\bdelta}^{\dag}
\right)^{-1/2}
\sqrt{T}
\left(
\hat{\bdelta}_{\bR}^{\dag,bc}
-
\bdelta_{\bR}^{\dag}
\right)
\CDb
N
\left(
\boldsymbol{0},
\bI_{r+p}
\right),
\label{eq:proof-db-studentized}
\end{align}
in probability. This proves the first statement of the theorem.

It remains to establish the distributional approximation in the second
statement. Define
\begin{align*}
\bT_{\bR}
&:=
\hat{\bSigma}_{\bdelta}^{-1/2}
\sqrt{T}
\left(
\hat{\bdelta}_{\bR}^{bc}
-
\bdelta_{\bR}
\right),\\
\bT_{\bR}^{\dag}
&:=
\left(
\hat{\bSigma}_{\bdelta}^{\dag}
\right)^{-1/2}
\sqrt{T}
\left(
\hat{\bdelta}_{\bR}^{\dag,bc}
-
\bdelta_{\bR}^{\dag}
\right),
\end{align*}
and let $\Phi_{r+p}(\bx)$ denote the distribution function of
$N(\boldsymbol{0},\bI_{r+p})$, with inequalities understood
componentwise.

By Theorem~\ref{thm:bc-all-rotations},
\[
\bT_{\bR}
\CD
N(\boldsymbol{0},\bI_{r+p}).
\]
Since the limiting multivariate normal distribution has a continuous
distribution function, the multivariate version of Pólya's theorem gives
\begin{align}
\sup_{\bx\in\bbR^{r+p}}
\left|
\Pro
\left(
\bT_{\bR}\leq\bx
\right)
-
\Phi_{r+p}(\bx)
\right|
\longrightarrow
0.
\label{eq:proof-db-polya-sample}
\end{align}

Similarly, \eqref{eq:proof-db-studentized} implies the conditional
version
\begin{align}
\sup_{\bx\in\bbR^{r+p}}
\left|
\Pro^{\dag}
\left(
\bT_{\bR}^{\dag}\leq\bx
\right)
-
\Phi_{r+p}(\bx)
\right|
=
o_p(1).
\label{eq:proof-db-polya-bootstrap}
\end{align}
The latter follows from conditional weak convergence in probability to
the continuous limiting distribution.

Using the triangle inequality,
\begin{align*}
&
\sup_{\bx\in\bbR^{r+p}}
\left|
\Pro^{\dag}
\left(
\bT_{\bR}^{\dag}\leq\bx
\right)
-
\Pro
\left(
\bT_{\bR}\leq\bx
\right)
\right|
\\
&\quad\leq
\sup_{\bx\in\bbR^{r+p}}
\left|
\Pro^{\dag}
\left(
\bT_{\bR}^{\dag}\leq\bx
\right)
-
\Phi_{r+p}(\bx)
\right|
\\
&\qquad\quad+
\sup_{\bx\in\bbR^{r+p}}
\left|
\Pro
\left(
\bT_{\bR}\leq\bx
\right)
-
\Phi_{r+p}(\bx)
\right|.
\end{align*}
The first term on the right-hand side is $o_p(1)$ by
\eqref{eq:proof-db-polya-bootstrap}, while the second converges to zero
by \eqref{eq:proof-db-polya-sample}. Hence
\[
\sup_{\bx\in\bbR^{r+p}}
\left|
\Pro^{\dag}
\left(
\bT_{\bR}^{\dag}\leq\bx
\right)
-
\Pro
\left(
\bT_{\bR}\leq\bx
\right)
\right|
=
o_p(1),
\]
which is the desired bootstrap consistency result.

Finally, $\cR$ contains only the three rotation targets
$\hat{\bH}$, $\hat{\bH}_q$, and $\bH$. Therefore all of the preceding
conditional convergences hold simultaneously over $\bR\in\cR$, since
the maximum of finitely many $o_p(1)$ terms is again $o_p(1)$. This
establishes the result jointly over the finite set $\cR$.
\end{proof}

\subsection{Proofs of the rotation-specific bootstrap limits}
\label{app:rotation-specific-proofs}

This subsection proves Theorems~\ref{thm:bias_Hhat}--\ref{thm:bias_H} and Proposition~\ref{prop:Q-I} stated in Section~\ref{app:rotation-specific-bootstrap}. The technical lemmas invoked in the proofs are stated and proved in Section~\ref{app:aux-lemmas}.

To proceed with the proof, we introduce new rotation matrices defined as follows: 
\begin{align*}
   &  \tilde{\bH}_b^{\dag}:= \bB^{0\dag\prime}\hat{\bB}^{\dag}\left( {\hat{\bB}}^{\dag\prime}\hat{\bB}^{\dag}\right)^{-1},\; \tilde{\bQ}^{\dag}:=  \frac{1}{T}{\hat{\bF}}^{\dag\prime}\bF^{0\dag}.
\end{align*}
The bootstrap model $\bX^{\dag}=\bF^{0\dag}\bB^{0\dag\prime} +\bE^{\dag}$ satisfies the PC1 conditions:
\[
\frac{1}{T}\bF^{0\dag\prime}\bF^{0\dag}=\bI_r, \quad \bB^{0\dag\prime}\bB^{0\dag} \in \bD(r).
\]
Therefore, the framework of \cite{jiang2023revisiting} applies. It follows from their Proposition 1 that these bootstrap rotation matrices denoted with “tilde” and “$\dag$” are all asymptotically equivalent to $\bI_r$ in probability. 
Let
$\bLambda^\dag:=\bB^{0\dag\prime}\bB^{0\dag}=\operatorname{diag}(\lambda_1^\dag,\ldots,\lambda_r^\dag)$,
$\lambda_k^\dag=d_k^{\dag}N^{\alpha_k}$, and
$\bD^\dag:=\bLambda^\dag\bN^{-1}=\diag(d_1^\dag,\ldots,d_r^\dag)$.
Since $\bB^{0\dag}=\hat{\bB}$ and $\hat{\bB}'\hat{\bB}=\hat{\bLambda}$, we have $\bLambda^\dag=\hat{\bLambda}$ and $\bD^\dag=\hat{\bLambda}\bN^{-1}$, which are functions of the original sample.
The original-sample PC eigenvalue consistency $\|\bN^{-1/2}(\hat{\bLambda}-\bLambda)\bN^{-1/2}\|_{\F}=o_p(1)$ \citep[Proposition 1(ii)]{jiang2023revisiting}, which holds under Assumptions \ref{ass:eigen}--\ref{ass:Aug_errors} and $N^{1-\alpha_r}/T\to0$, therefore gives
\begin{equation}\label{eq:D-dag-cons}
\bD^\dag-\bD=(\hat{\bLambda}-\bLambda)\bN^{-1}=o_p(1).
\end{equation}
Hence the bootstrap-population signal strengths $d_k^\dag$ inherit the boundedness and eigengap separation of $d_k$ in Assumption~\ref{ass:signal}, so that the bootstrap DGP satisfies the analogue of Assumption~\ref{ass:signal}.

More generally, under Assumptions \ref{ass:errors_b}--\ref{ass:bootstrap_b}, with \eqref{eq:D-dag-cons} supplying the bootstrap analogue of Assumption~\ref{ass:signal}, the bootstrap sample satisfies the bootstrap analogues of the conditions imposed on the original sample in \citet{jiang2023revisiting} and \citet{jiang2024Mw}, so that their principal component convergence rates hold conditionally on the original sample. We use them in this and the next subsection without restating them: the consistency of the bootstrap signal eigenvalues, $\bN^{-1}\hat{\bLambda}^{\dag}=\bD^{\dag}+o_{p^{\dag}}(1)$ with $\hat{\lambda}_k^{\dag}\asymp N^{\alpha_k}$, the stochastic boundedness of the bootstrap rotation matrices, and the moment bounds $\|\bF^{0\dag\prime}\bE^{\dag}\|_{\F}=O_{p^{\dag}}(\sqrt{NT})$ and $\|\bE^{\dag}\bB^{0\dag}\bN^{-\frac{1}{2}}\|_{\F}=O_{p^{\dag}}(\sqrt{T})$.

Throughout this and the next subsection, we use the following quantities. The deterministic rates are
\[
\Delta_{NT}
:=
\frac{N^{1-\alpha_r}}{T}+ N^{\frac{1}{2}\alpha_1- \alpha_r} \frac{N^{1- \alpha_r}}{T} + N^{\frac{1}{2}\alpha_{1}-\frac{3}{2}\alpha_{r}} ,
\qquad
\tilde{\Delta}_{NT}:=\frac{N^{1-\alpha_r}}{T}+  N^{-\alpha_{r}} .
\]

Let $\hat\lambda_\ell^\dag$ denote the $\ell$-th diagonal element of $\hat{\bLambda}^\dag:=\hat{\bB}^{\dag\prime}\hat{\bB}^{\dag}$, and let $\bM^\dag:=(m_{k\ell}^\dag)$ and $\bM_0^\dag:=(m_{0,k\ell}^\dag)$, where
\[
m_{k\ell}^\dag
:=
\frac{\lambda_k^\dag}{\hat\lambda_\ell^\dag-\lambda_k^\dag},
\qquad
m_{0,k\ell}^\dag
:=
\frac{\lambda_k^\dag}{\lambda_\ell^\dag-\lambda_k^\dag},
\qquad k\neq\ell,
\]
and $m_{kk}^\dag=m_{0,kk}^\dag:=0$. The matrix $\bM^\dag$ varies across bootstrap draws through $\hat\lambda_\ell^\dag$, whereas $\bM_0^\dag$ is its bootstrap-population counterpart and a function of the original sample.
For the eigengap matrices, $\lambda_k^\dag=\hat\lambda_k$ implies $\bM_0^\dag=\hat{\bM}$, where $\hat{\bM}=(\hat m_{k\ell})$ with $\hat m_{k\ell}=\hat\lambda_k/(\hat\lambda_\ell-\hat\lambda_k)$ is the estimator of $\bM$ in \citet[Theorem 3]{jiang2024Mw}. Original-sample eigenvalue consistency and the eigengap separation give $\bM_0^\dag-\bM=o_p(1)$, while the conditional bootstrap eigenvalue consistency and the same separation give $\bM^\dag-\bM_0^\dag=o_{p^\dag}(1)$ in probability. Therefore
\begin{equation}\label{eq:M-dag-cons}
\bM^\dag-\bM=(\bM^\dag-\bM_0^\dag)+(\bM_0^\dag-\bM)=o_{p^\dag}(1),\quad\text{in probability.}
\end{equation}
Define
\begin{align*}
&  \mathcal G^\dag
 :=
\sqrt T
\bN^{1/2}
\bGamma^\dag\bD^{\dag-2}
\bN^{-3/2},
\quad
\bar{\mathcal G}^\dag
:=
\sqrt T
\bN^{-1/2}
\bD^{\dag-1}\bGamma^\dag\bD^{\dag-1}
\bN^{-1/2},\\
& \hat{\bSigma}_{\bW\hat{\bF}}
:=
T^{-1}\bW'\hat{\bF}
=
T^{-1}\bW'\bF^{0\dag}.
\end{align*}

\begin{proof}[Proof of Theorem \ref{thm:bias_Hhat}] Replacing $\hat{\mathbf{F}}^\dag$ with $\mathbf{F}^{0\dag}$, we have
    \begin{align*}
 \frac{1}{T}  \hat\bZ^{\dag\prime}\hat\bZ^{\dag}   & =
  \frac{1}{T}\bZ^{0\dag\prime}\bZ^{0\dag} 
+
 O_{p^{\dag}}\left(\Delta_{NT}\right)+O_{p^{\dag}}\left(\frac{1}{\sqrt{T}}\right)
\end{align*}
in probability, because Lemmas \ref{lem:Rotation}(iv)(v) and \eqref{eq:bias-Hhat-W} imply
\begin{align*}
& \frac{1}{T} \hat{\bF}^{\dag\prime}\bF^{0\dag}-\bI_r=O_{p^{\dag}}\left( \tilde{\Delta}_{NT}\right), \quad 
\frac{1}{T} \bW'(\hat\bF^{\dag}-\bF^{0\dag})=O_{p^{\dag}}\left( \Delta_{NT} \right)+O_{p^{\dag}}\left(\frac{1}{\sqrt{T}}\right),
\end{align*}
in probability. As shown by \citet[Proof of Theorem 1]{jiang2024Mw} that $\frac{1}{T}  \hat\bZ^{\prime}\hat\bZ    =
  \frac{1}{T}\bZ^{0\prime}\bZ^{0} 
+o_p(1)$, the analogous result within bootstrap holds, and therefore
$ \frac{1}{T}  \hat\bZ^{\dag\prime}\hat\bZ^{\dag}
    =
  \frac{1}{T}\bZ^{0\dag\prime}\bZ^{0\dag} 
+o_{p^{\dag}}(1)$ in probability. In addition, $ \plim \bN^{-1}\hat{\bLambda} = \bD $ and $ \plim \bGamma^{\dag}=\bGamma$ imply that
$\mathcal G^\dag
\CP
c_1\bG $ and $  \bD^{\dag-1} \bGamma^{\dag} \bD^{\dag-1} \CP  \bD^{ -1} \bGamma  \bD^{ -1}$. Furthermore, $\plim \bgamma^{0\dag} = \bH_0^{-1} \bgamma^*$ and $\plim \hat{\bSigma}_{\bW \hat{\bF} } =\bSigma_{\bW \bF^0}   $.     
Combining these results with Lemmas \ref{lem:F*epsilon} and \ref{lem:bias_Hhat},
  \begin{align*}
& \sqrt{T}(\hat\bdelta^{\dag} - \bdelta_{\hat{\bH}^{\dag}})  \\
&= (T^{-1}\hat\bZ^{\dag\prime}\hat\bZ^{\dag})^{-1} T^{-\frac{1}{2}}\hat\bZ^{\dag\prime}\bepsilon^{\dag}
- (T^{-1}\hat\bZ^{\dag\prime}\hat\bZ^{\dag})^{-1} T^{-\frac{1}{2}}\hat\bZ^{\dag\prime}(\hat\bF^{\dag}-\bF^{0\dag}\tilde{\bH}^{\dag})\tilde{\bH}^{{\dag}-1}\bgamma^{0\dag}  \\
&  = (T^{-1} \bZ^{0\dag\prime} \bZ^{0\dag})^{-1} T^{-\frac{1}{2}}  \bZ^{0\dag\prime}\bepsilon^{\dag}
-(T^{-1} \bZ^{0\dag\prime} \bZ^{0\dag})^{-1} T^{-\frac{1}{2}}\hat\bZ^{\dag\prime}(\hat\bF^{\dag}-\bF^{0\dag}\tilde{\bH}^{\dag})\tilde{\bH}^{{\dag}-1}\bgamma^{0\dag}  \\
& 
\quad \quad +
O_{p^{\dag}}(\Delta_{NT})+
O_{p^{\dag}} \left(\sqrt{T} N^{-\frac{3}{2}\alpha_r}\right)+O_{p^{\dag}}\left( \frac{N^{1-\alpha_r}}{\sqrt{T}} \right)+O_{p^{\dag}}\left( N^{\frac{1}{2}- \alpha_r} \right)\\
& \CDb
N\left(c_1 \bkappa_{\bdelta^*}, \bSigma_{\bdelta}\right)
\end{align*}
in probability, where $\bSigma_{\bdelta}= \bSigma_{\bZ^0 \bZ^0}^{-1} \bSigma_{\bZ^0 \bepsilon} \bSigma_{\bZ^0 \bZ^0}^{-1}$ and
\begin{align*}
   \bkappa_{\bdelta^*} = -\bSigma_{\bZ^0 \bZ^0}  ^{-1}\binom{\bG+\nu \bD^{-1}\bGamma\bD^{-1}}{\bSigma_{\bW \bF^0}  {\bG}}\bH_0^{-1} \bgamma^*. 
\end{align*} 
\end{proof}

\begin{proof}[Proof of Theorem \ref{thm:bias_H3}]
By the definition of $\tilde{\bH}_q^{\dag}$, we have
\begin{align*}
    &\hat\bF^{\dag\prime}(\hat\bF^{\dag}-\bF^{0\dag}\tilde{\bH}_q^{\dag})\tilde{\bH}_q^{\dag-1}\bgamma^{0\dag}= \mathbf{0}.
\end{align*}
Combining with Lemmas \ref{lem:F*epsilon} and \ref{lem:bias_Hhat3},
  \begin{align*}
& \sqrt{T}(\hat\bdelta^{\dag} - \bdelta_{\hat{\bH}_q^{\dag}})  \\
&= (T^{-1}\hat\bZ^{\dag\prime}\hat\bZ^{\dag})^{-1} T^{-\frac{1}{2}}\hat\bZ^{\dag\prime}\bepsilon^{\dag}
- (T^{-1}\hat\bZ^{\dag\prime}\hat\bZ^{\dag})^{-1} T^{-\frac{1}{2}}\hat\bZ^{\dag\prime}(\hat\bF^{\dag}-\bF^{0\dag}\tilde{\bH}_q^{\dag})\tilde{\bH}_q^{\dag-1}\bgamma^{0\dag}  \\
    &  = (T^{-1} \bZ^{0\dag\prime} \bZ^{0\dag})^{-1} T^{-\frac{1}{2}}  \bZ^{0\dag\prime}\bepsilon^{\dag}
-  (T^{-1} \bZ^{0\dag\prime} \bZ^{0\dag})^{-1} T^{-\frac{1}{2}}\hat\bZ^{\dag\prime}(\hat\bF^{\dag}-\bF^{0\dag}\tilde{\bH}_q^{\dag})\tilde{\bH}_q^{\dag-1}\bgamma^{0\dag} \\
& 
\quad \quad +
O_{p^{\dag}}(\Delta_{NT})+
O_{p^{\dag}} \left(\sqrt{T} N^{-\frac{3}{2}\alpha_r}\right)+O_{p^{\dag}}\left( \frac{N^{1-\alpha_r}}{\sqrt{T}} \right)+O_{p^{\dag}}\left( N^{\frac{1}{2}- \alpha_r} \right)\\
& \CDb
N\left(c_2 \bar{\bkappa}_{\bdelta^*}, \bSigma_{\bdelta}\right)
\end{align*}
in probability, where 
\begin{align*}
   \bar{\bkappa}_{\bdelta^*} = \bSigma_{\bZ^0 \bZ^0}^{-1}\binom{\mathbf{0}}{\bSigma_{\bW \bF^0}  \bar{\bG}}\bH_0^{-1} \bgamma^*. 
\end{align*}  

\end{proof}

{\bf Notation and preliminary definitions.}
Define
\[
\hat\bU^\dag:=\hat\bF^\dag/\sqrt T,\qquad
\bU^{0\dag}:=\bF^{0\dag}/\sqrt T,\qquad
\bQ^\dag:=\bU^{0\dag\prime}\hat\bU^\dag .
\]
Then
\[
\tilde{\bQ}^\dag
=
T^{-1}\hat\bF^{\dag\prime}\bF^{0\dag}
=
\hat\bU^{\dag\prime}\bU^{0\dag}
=
\bQ^{\dag\prime}.
\]
Let \(\bU_\perp^\dag\) be an orthonormal basis for the orthogonal complement of
\(\mathrm{span}(\bU^{0\dag})\), and set
$
\bS^\dag:=\bU_\perp^{\dag\prime}\hat\bU^\dag .
$
To analyze \(\bQ^\dag-\bI_r\), consider the polar decomposition of \(\bQ^\dag\). Let
\[
\bQ^\dag=\bU^\dag\bPsi^\dag\bV^{\dag\prime},
\qquad
\bPsi^\dag\succeq\bzero,
\]
be its singular value decomposition, and define
\[
\bO^{\ast\dag}:=\bU^\dag\bV^{\dag\prime},
\qquad
\bC^{\ast\dag}:=(\bQ^{\dag\prime}\bQ^\dag)^{1/2}
=
\bV^\dag\bPsi^\dag\bV^{\dag\prime}.
\]
Then
\[
\bQ^\dag=\bO^{\ast\dag}\bC^{\ast\dag},
\]
and hence
\[
\bQ^\dag-\bI_r
=
(\bO^{\ast\dag}-\bI_r)
+
(\bC^{\ast\dag}-\bI_r)
+
(\bO^{\ast\dag}-\bI_r)(\bC^{\ast\dag}-\bI_r).
\]
Lemma \ref{lem:Q-I dec} gives the bootstrap expansions of
\(\bC^{\ast\dag}-\bI_r\) and \(\bO^{\ast\dag}-\bI_r\), which are then combined in Proposition \ref{prop:Q-I}.

\begin{proof}[Proof of Proposition \ref{prop:Q-I}]
By the polar decomposition,
\[
\bQ^\dag-\bI_r
=
(\bO^{\ast\dag}-\bI_r)
+
(\bC^{\ast\dag}-\bI_r)
+
(\bO^{\ast\dag}-\bI_r)(\bC^{\ast\dag}-\bI_r).
\]
Lemma \ref{lem:Q-I dec} gives
\[
\sqrt T(\bC^{\ast\dag}-\bI_r)
=
-\frac12\bar{\mathcal{G}}^{\dag}
+
o_{p^\dag}(1),
\]
and
\[
\sqrt T(\bO^{\ast\dag}-\bI_r)
=
\frac{1}{2}
(\bM^\dag-\bM^{\dag\prime})
\odot\bar{\mathcal{G}}^{\dag}
+
o_{p^\dag}(1),
\]
in probability.
Moreover,
\[
\sqrt T
(\bO^{\ast\dag}-\bI_r)
(\bC^{\ast\dag}-\bI_r)
=
O_{p^\dag}(\sqrt T N^{-2\alpha_r})
=
o_{p^\dag}(1),
\]
in probability. Therefore,
\begin{align}
\sqrt T(\bQ^\dag-\bI_r)
=
\frac{1}{2}
(\bM^\dag-\bM^{\dag\prime}-\mathbf 1_{r\times r})
\odot\bar{\mathcal{G}}^{\dag}
+
o_{p^\dag}(1),
\label{sqrtT(Q-I)}
\end{align}
in probability.
Since \(\bM^\dag \CPb \bM\) in probability and \(\bar{\mathcal{G}}^{\dag} \CP c_2\bar{\bG}\), \eqref{sqrtT(Q-I)} implies
\[
\sqrt T(\bQ^{\dag\prime}-\bI_r)
=
\frac{1}{2}
(\bM'-\bM-\mathbf 1_{r\times r})
\odot c_2\bar{\bG}
+
o_{p^\dag}(1)
\]
in probability. Recall that
\[
\tilde{\bH}_q^{\dag-1}
=
\tilde{\bQ}^{\dag}
=
T^{-1}\hat\bF^{\dag\prime}\bF^{0\dag}
=
\bQ^{\dag\prime}.
\]
Thus,
\[
\sqrt T
(\bdelta_{\hat{\bH}_q^{\dag}}-\bdelta^{0\dag})
=
\binom{
\sqrt T(\bQ^{\dag\prime}-\bI_r)\bgamma^{0\dag}
}
{\mathbf 0_p}.
\]
Since
\[
\bgamma^{0\dag}=\hat\bgamma\CP \bH_0^{-1}\bgamma^*,
\]
we obtain
\[
\sqrt T
(\bdelta_{\hat{\bH}_q^{\dag}}-\bdelta^{0\dag})
\CPb
c_2\bh_{\bgamma^*}
\]
in probability, where
\[
c_2\bh_{\bgamma^*}
=
\binom{
\frac{c_2}{2}
\left[
(\bM'-\bM-\mathbf 1_{r\times r})
\odot\bar{\bG}
\right]
\bH_0^{-1}\bgamma^*
}
{\mathbf 0_p}.
\]
\end{proof}

\begin{proof}[Proof of Theorem \ref{thm:bias_H}]
We consider
\begin{align*}
& \sqrt{T}( \hat{\bdelta}^{\dag}- \bdelta^{0\dag})=
  \left( 
\begin{array}{c}
\sqrt{T}(  \tilde{\bH}_q^{\dag-1}- \bI_r)\bgamma^{0\dag} \\
\nonumber
    \mathbf{0}
\end{array}
\right)
 + \sqrt{T}(\hat\bdelta^{\dag} - \bdelta_{\hat{\bH}_q^{\dag}}).
  \label{delta_est}
\end{align*} 
Proposition~\ref{prop:Q-I} implies that the first term
$(\sqrt{T}\bgamma^{0\dag\prime}
(\tilde{\bH}_q^{\dag\prime-1}-\bI_r),\mathbf{0}^{\prime})'
\CPb c_2\bh_{\bgamma^*}$ in probability. The second bias is the same as
that in
$\sqrt{T}(\hat{\bdelta}^{\dag}
-\bdelta_{\hat{\bH}_q^{\dag}})$,
given by $c_2\bar{\bkappa}_{\bdelta^*}$. Thus,
\[
\sqrt{T}(\hat{\bdelta}^{\dag}-\bdelta^{0\dag})
\CDb
N\left(
c_2\bh_{\bgamma^*}
+c_2\bar{\bkappa}_{\bdelta^*},
\bSigma_{\bdelta}
\right)
\]
in probability.
\end{proof}

\subsection{Auxiliary lemmas}
\label{app:aux-lemmas}

This subsection collects the technical lemmas. 
Lemmas~ \ref{lem:bias_Hhat}, \ref{lem:bias_Hhat3} and \ref{lem:Q-I dec} contain the new derivations of this paper; the lemmas at the end of the subsection are taken from the literature and stated for completeness.

\begin{lem}\label{lem:bias_Hhat}
Suppose Assumptions \ref{ass:eigen}--\ref{ass:Aug_errors} and
\ref{ass:errors_b}--\ref{ass:bootstrap_b} hold. If
\[
\alpha_r>\frac12,\qquad
\frac{N^{1-\alpha_r}}{\sqrt T}\to0,\qquad
\sqrt T N^{\frac12\alpha_1-\frac32\alpha_r}\to c_1\in[0,\infty),
\]
then the following statements hold in probability, 
\begin{align}
\frac{1}{\sqrt T}
\hat{\bF}^{\dag\prime}
(\hat{\bF}^{\dag}-\bF^{0\dag}\tilde{\bH}^{\dag})
&=
\mathcal G^\dag+\bar{\mathcal G}^\dag+o_{p^\dag}(1),
\label{eq:bias-Hhat-F}\\
\frac{1}{\sqrt T}
\bW'
(\hat{\bF}^{\dag}-\bF^{0\dag}\tilde{\bH}^{\dag})
&=
\hat{\bSigma}_{\bW\hat{\bF}}\mathcal G^\dag
+o_{p^\dag}(1),
\label{eq:bias-Hhat-W}
\end{align}
where $\mathcal G^\dag$, $\bar{\mathcal G}^\dag$ and $\hat{\bSigma}_{\bW\hat{\bF}}$ are as defined at the beginning of Section~\ref{app:rotation-specific-proofs}.
Moreover, under Assumption \ref{ass:bootstrap_b} and \eqref{eq:D-dag-cons},
\[
\mathcal G^\dag
\CP
c_1\bG,
\qquad
\bar{\mathcal G}^\dag
\CP
c_1\nu \bD^{-1}\bGamma\bD^{-1},
\]
where
\[
\bG:=(\be_{\alpha_1}\be_{\alpha_r}')\odot \bGamma\bD^{-2},
\qquad
\nu:=\lim_{N\to\infty}N^{-(\alpha_1-\alpha_r)/2}.
\]
Here \(\be_{\alpha_1}\) and \(\be_{\alpha_r}\) are the indicator vectors for the entries of
\((\alpha_1,\ldots,\alpha_r)'\) equal to \(\alpha_1\) and \(\alpha_r\), respectively.
In particular, \(\nu=1\) if \(\alpha_1=\alpha_r\), and \(\nu=0\) if
\(\alpha_1>\alpha_r\).
\end{lem}

\begin{proof}[Proof of Lemma \ref{lem:bias_Hhat}]
The proof follows the bootstrap analogue of \citet[Lemma A.3]{jiang2024Mw}.
We give the main steps in order to clarify the limiting objects. For \eqref{eq:bias-Hhat-F}, write
\begin{align*}
\frac{1}{\sqrt T}
\hat{\bF}^{\dag\prime}
(\hat{\bF}^{\dag}-\bF^{0\dag}\tilde{\bH}^{\dag})
&=
\frac{1}{\sqrt T}
\hat{\bF}^{\dag\prime}
\left(
\frac{1}{T}\bE^\dag\bE^{\dag\prime}\hat{\bF}^\dag\hat\bLambda^{\dag-1}
+
\frac{1}{T}\bE^\dag\bB^{0\dag}
\bF^{0\dag\prime}\hat{\bF}^\dag\hat\bLambda^{\dag-1}
\right.\\
&\qquad\qquad\left.
+
\frac{1}{T}\bF^{0\dag}
\bB^{0\dag\prime}\bE^{\dag\prime}\hat{\bF}^\dag\hat\bLambda^{\dag-1}
\right)\\
&=:\bA_1+\bA_2+\bA_3 .
\end{align*}
The first term satisfies
\[
\|\bA_1\|_{\F}
=
O_{p^\dag}
\left(
\frac{N^{1-\alpha_r}}{\sqrt T}
+
\sqrt T N^{-3\alpha_r/2}
+
N^{1/2-\alpha_r}
\right)
=
o_{p^\dag}(1),
\]
under the maintained rate conditions. Thus the dominant terms arise from
\(\bA_2\) and \(\bA_3\). Expanding
\[
\hat{\bF}^\dag
=
(\hat{\bF}^\dag-\bF^{0\dag}\tilde{\bH}_b^\dag)
+
\bF^{0\dag}\tilde{\bH}_b^\dag
\]
in \(\bA_2\) and \(\bA_3\), and using the bootstrap analogues of the preliminary PC
bounds, yields two leading components:
\begin{align}
&\sqrt{T} \bN^{\frac{1}{2}} \hat{\bLambda}^{{\dag}-1} \left(\bN^{-\frac{1}{2}} \frac{\hat{\bF}^{\dag\prime} \bF^{0\dag}}{T} \bN^{\frac{1}{2}}\right) \frac{\bN^{-\frac{1}{2}} \bB^{0\dag\prime} \bE^{\dag\prime} \bE^{\dag} \bB^{0\dag} \bN^{-\frac{1}{2}}}{T}\left(\bN^{\frac{1}{2}} \frac{\bF^{0\dag\prime} \hat{\bF}^{\dag}}{T} \bN^{-\frac{1}{2}}\right) \bN^{\frac{1}{2}} \hat{\bLambda}^{{\dag}-1}
\notag\\
&
=
\bar{\mathcal G}^\dag+o_{p^\dag}(1),
\label{eq:barGcal-dag-lead}
\\
&\sqrt{T} \bN^{\frac{1}{2}} \left(\bN^{-\frac{1}{2}} \frac{\hat{\bF}^{\dag\prime} \bF^{0\dag}}{T} \bN^{\frac{1}{2}}\right) \frac{\bN^{-\frac{1}{2}} \bB^{0\dag\prime} \bE^{\dag\prime} \bE^{\dag} \bB^{0\dag} \bN^{-\frac{1}{2}}}{T}\left(\bN^{\frac{1}{2}} \frac{\bF^{0\dag \prime} \hat{\bF}^{\dag}}{T} \bN^{-\frac{1}{2}}\right) \bN^{\frac{1}{2}} \hat{\bLambda}^{{\dag}-2}
\notag\\
&
=
\mathcal G^\dag+o_{p^\dag}(1).
\label{eq:Gcal-dag-lead}
\end{align}
Indeed, by bootstrap eigenvalue consistency,
\[
\hat\bLambda^\dag\bN^{-1}
=
\bD^\dag+o_{p^\dag}(1),
\qquad
\bN\hat\bLambda^{\dag-1}
=
\bD^{\dag-1}+o_{p^\dag}(1),
\]
in probability. Moreover, Assumption \ref{ass:factor and loadings_b}(iv) implies
\[
\frac{
\bN^{-1/2}
\bB^{0\dag\prime}\bE^{\dag\prime}\bE^\dag\bB^{0\dag}
\bN^{-1/2}
}{T}
=
\bGamma^\dag+o_{p^\dag}(1).
\]
Therefore, \eqref{eq:barGcal-dag-lead} and \eqref{eq:Gcal-dag-lead}
give
\[
\frac{1}{\sqrt T}
\hat{\bF}^{\dag\prime}
(\hat{\bF}^{\dag}-\bF^{0\dag}\tilde{\bH}^{\dag})
=
\mathcal G^\dag+\bar{\mathcal G}^\dag+o_{p^\dag}(1),
\]
which proves \eqref{eq:bias-Hhat-F}.

We next derive the probability limits of the scaled bootstrap matrices. \\
Case 1: When \(\alpha_1=\alpha_r=\alpha\), the scaled bootstrap matrices reduce to
\[
\bar{\mathcal{G}}^\dag
=
\frac{\sqrt T}{N^\alpha}
\bD^{\dag-1}\bGamma^\dag\bD^{\dag-1},
\qquad
\mathcal{G}^\dag
=
\frac{\sqrt T}{N^\alpha}
\bGamma^\dag\bD^{\dag-2}.
\]

Under \(\sqrt T/N^\alpha\to c_1\), Assumption \ref{ass:bootstrap_b} and \eqref{eq:D-dag-cons},
\[
\bar{\mathcal{G}}^\dag
\CP
c_1\bD^{-1}\bGamma\bD^{-1},
\qquad
\mathcal G^\dag
\CP
c_1\bGamma\bD^{-2}
=c_1(\be_{\alpha_1}\be_{\alpha_r}')
\odot
\bGamma\bD^{-2},
\]
since we have
\(\be_{\alpha_1}=\be_{\alpha_r}=\mathbf 1_r\).\\
Case 2: \(\alpha_1>\alpha_r\). In this case, the term associated with
\(\mathcal G^\dag\) dominates, whereas the term associated with
\(\bar{\mathcal G}^\dag\) is asymptotically negligible under the scaling
\(\sqrt T N^{\frac12\alpha_1-\frac32\alpha_r}\). If
\[
\sqrt{T}N^{\frac12\alpha_1-\frac32\alpha_r}
\to c_1\in[0,\infty),
\]
then
\[
\frac{1}{\sqrt{T}}
\hat{\bF}^{\dag\prime}
(\hat{\bF}^{\dag}-\bF^{0\dag}\tilde{\bH}^{\dag})
=
\mathcal G^\dag
+
o_{p^\dag}(1)
\]
in probability. Indeed, the \((k,\ell)\)-th element of \(\mathcal G^\dag\) is
\[
\mathcal G^\dag_{k\ell}
=
\left(
\sqrt T N^{\frac12\alpha_1-\frac32\alpha_r}
\right)
N^{\frac12(\alpha_k-\alpha_1)-\frac32(\alpha_\ell-\alpha_r)}
[\bGamma^\dag\bD^{\dag-2}]_{k\ell}.
\]
The exponent is non-positive and equals zero only when
\(\alpha_k=\alpha_1\) and \(\alpha_\ell=\alpha_r\). Hence, under Assumption \ref{ass:bootstrap_b} and \eqref{eq:D-dag-cons},
\[
\mathcal G^\dag
\CP
c_1
(\be_{\alpha_1}\be_{\alpha_r}')
\odot
\bGamma\bD^{-2}.
\]
For \eqref{eq:bias-Hhat-W}, the same expansion gives the leading term
\[
\frac{1}{\sqrt T}
\bW'
(\hat{\bF}^\dag-\bF^{0\dag}\tilde{\bH}^\dag)
=
\hat{\bSigma}_{\bW\hat{\bF}}\mathcal G^\dag
+o_{p^\dag}(1),
\]
where $\hat{\bSigma}_{\bW\hat{\bF}}$ is as defined at the beginning of Section~\ref{app:rotation-specific-proofs}.
This proves \eqref{eq:bias-Hhat-W}.
\end{proof}

\begin{lem} 
\label{lem:bias_Hhat3}
 Suppose Assumptions \ref{ass:eigen}–\ref{ass:Aug_errors} and \ref{ass:errors_b}--\ref{ass:bootstrap_b} hold. If $\alpha_r>\frac{1}{2}$, $\frac{N^{1-\alpha_r}}{\sqrt{T}} \to  0$, and $\sqrt{T}N^{-\alpha_r} \to c_2 \in [0,\infty)$, then, the following statements hold in probability, as $N, T \to \infty$,
\begin{align*}
     &   \frac{1}{\sqrt{T}}  \bW'(\hat{\bF}^{\dag}-\bF^{0\dag}\tilde{\bH}_q^{\dag})   =   -\hat{\bSigma}_{\bW \hat{\bF} } \bar{\mathcal{G}}^\dag+o_{p^{\dag}}(1),
\end{align*}
 where $\bar{\mathcal{G}}^\dag$ is as defined at the beginning of Section~\ref{app:rotation-specific-proofs}. Moreover, under Assumption \ref{ass:bootstrap_b} and \eqref{eq:D-dag-cons}, 
$\bar{\mathcal{G}}^\dag \CP c_2\bar{\bG} := c_2(\be_{\alpha_r} \be_{\alpha_r}')
\odot\bD^{-1}\bGamma  \bD^{-1}$ . Here \(\be_{\alpha_r}\) is the indicator vector for the entries of
\((\alpha_1,\ldots,\alpha_r)'\) equal to \(\alpha_r\). When
\(\alpha_1=\alpha_r\), \(\be_{\alpha_r}=\mathbf 1_r\), and hence this reduces to
$
\bar{\mathcal G}^\dag
\CP
c_2\bD^{-1}\bGamma\bD^{-1}$.
\end{lem}

\begin{proof}[Proof of Lemma \ref{lem:bias_Hhat3}] We follow the argument of \citet[Lemma A.4]{jiang2024Mw}. 
If $\sqrt{T}N^{-\alpha_r} \to c_2$, the dominant term is
\begin{align*}
 &-\sqrt{T} \left(  \frac{\bW' \bF^{0\dag}}{T}  \right)     \tilde{\bQ}^{{\dag}-1}  \hat{\bLambda}^{{\dag}-1} \bN^{\frac{1}{2}}  \left(\bN^{-\frac{1}{2}} \tilde{\bQ}^{\dag}\bN^{ \frac{1}{2}}\right)    \frac{\bN^{-\frac{1}{2}} \bB^{0\dag\prime} \bE^{\dag\prime} \bE^{\dag} \bB^{0\dag}\bN^{-\frac{1}{2}}}{T} \bN^{\frac{1}{2}} \tilde{\bQ}^{\dag\prime}\bN^{-\frac{1}{2}}  \bN  \hat{\bLambda}^{{\dag}-1} \bN^{-\frac{1}{2}}\\
 & = -  \hat{\bSigma}_{\bW \hat{\bF} }  \bar{\mathcal{G}}^\dag +o_{p^{\dag}}(1),
\end{align*}
in probability. It remains to identify the probability limit of
\(\bar{\mathcal G}^\dag\). Its \((k,\ell)\)-th element is
\[
\bar{\mathcal G}^\dag_{k\ell}
=
\sqrt T
N^{-(\alpha_k+\alpha_\ell)/2}
\left[
\bD^{\dag-1}\bGamma^\dag\bD^{\dag-1}
\right]_{k\ell}.
\]
Since
\[
\sqrt T N^{-\alpha_r}\to c_2,
\]
we can rewrite this as
\[
\bar{\mathcal G}^\dag_{k\ell}
=
\left(\sqrt T N^{-\alpha_r}\right)
N^{\alpha_r-(\alpha_k+\alpha_\ell)/2}
\left[
\bD^{\dag-1}\bGamma^\dag\bD^{\dag-1}
\right]_{k\ell}.
\]
Because \(\alpha_k,\alpha_\ell\ge \alpha_r\), the exponent
\[
\alpha_r-\frac{\alpha_k+\alpha_\ell}{2}
\]
is non-positive and equals zero if and only if
\[
\alpha_k=\alpha_\ell=\alpha_r.
\]
Hence,
\[
\bar{\mathcal G}^\dag
\CP
c_2
(\be_{\alpha_r}\be_{\alpha_r}')
\odot
\bD^{-1}\bGamma\bD^{-1}.
\]
When \(\alpha_1=\alpha_r\), we have
\(\be_{\alpha_r}=\mathbf 1_r\), and therefore
\[
(\be_{\alpha_r}\be_{\alpha_r}')
\odot
\bD^{-1}\bGamma\bD^{-1}
=
\bD^{-1}\bGamma\bD^{-1}.
\]
This proves the stated result.
\end{proof}

\begin{lem}\label{lem:Q-I dec}
    Suppose that Assumptions \ref{ass:eigen}--\ref{ass:Aug_errors} and \ref{ass:errors_b}--\ref{ass:bootstrap_b} hold. If $\alpha_r>\frac{1}{2}$, $\frac{N^{1-\alpha_r}}{\sqrt{T}} \to  0$, $\sqrt{T}N^{-\alpha_r} \to c_2 \in [0,\infty)$ as $N, T \to \infty$,  then, conditionally on the original sample,
\begin{align*}
&(i)~ \sqrt T(\bC^{\ast\dag}-\bI_r)
=
-\frac12 \bar{\mathcal{G}}^{\dag}
+
o_{p^\dag}(1),\\
&(ii)~
\sqrt T(\bO^{\ast\dag}-\bI_r)
=
\frac{1}{2}
(\bM^\dag-\bM^{\dag\prime})
\odot\bar{\mathcal{G}}^{\dag}
+
o_{p^\dag}(1),
\end{align*}
in probability, where $\bar{\mathcal{G}}^\dag$ is as defined at the beginning of Section~\ref{app:rotation-specific-proofs},
and $\bM^\dag=(m_{k\ell}^\dag)$ is as defined there.
Moreover, under Assumption \ref{ass:bootstrap_b} and \eqref{eq:D-dag-cons}--\eqref{eq:M-dag-cons},
\[
\bar{\mathcal{G}}^{\dag} \CP c_2\bar{\bG} := c_2(\be_{\alpha_r} \be_{\alpha_r}')
\odot\bD^{-1}\bGamma  \bD^{-1},
\] 
\[
\bM^\dag \CPb \bM , \;\text{ in probability.}
\]
\end{lem}
\begin{proof}[Proof of Lemma \ref{lem:Q-I dec}]
The proof is the conditional analogue of the original-sample polar-decomposition argument in \citet[Lemma A.5]{jiang2024Mw}.
We provide the main steps for completeness. Using the decomposition
\[
\bU^{0\dag}\bU^{0\dag\prime}
+
\bU_\perp^\dag\bU_\perp^{\dag\prime}
=
\bI_T,
\]
we can write
\[
\hat\bU^\dag
=
(\bU^{0\dag}\bU^{0\dag\prime}
+
\bU_\perp^\dag\bU_\perp^{\dag\prime})
\hat\bU^\dag
=
\bU^{0\dag}\bQ^\dag
+
\bU_\perp^\dag\bS^\dag,
\qquad
\bQ^{\dag\prime}\bQ^\dag-\bI_r
=
-\bS^{\dag\prime}\bS^\dag .
\]
Hence, by the Fréchet expansion of the matrix square root,
\begin{align}
    \bC^{\ast\dag}-\bI_r
=
-\frac12\bS^{\dag\prime}\bS^\dag
+
O_{p^\dag}(\|\bS^\dag\|_2^4).
\label{C-I}
\end{align}
It remains to identify the leading term of \(\bS^{\dag\prime}\bS^\dag\). Let
\[
\hat\bA^\dag=T^{-1}\bX^\dag\bX^{\dag\prime},\qquad
\bA^{0\dag}=\bU^{0\dag}\bLambda^\dag\bU^{0\dag\prime},
\qquad
\bDelta_{\bA}^\dag=\hat\bA^\dag-\bA^{0\dag}.
\]
From
\[
(\bA^{0\dag}+\bDelta_{\bA}^\dag)\hat\bU^\dag
=
\hat\bU^\dag\hat\bLambda^\dag
\]
and \(\bU_\perp^{\dag\prime}\bA^{0\dag}=\mathbf{0}\), we obtain
\[
\bS^\dag
=
\bU_\perp^{\dag\prime}
\bDelta_{\bA}^\dag
\hat\bU^\dag
\hat\bLambda^{\dag-1}.
\]
Expanding \(\bDelta_{\bA}^\dag\) gives
\[
\bS^\dag
=
\frac{1}{\sqrt T}
\bU_\perp^{\dag\prime}
\bE^\dag\bB^{0\dag}
\bQ^\dag
\hat\bLambda^{\dag-1}
+
O_{p^\dag}
\left(
\frac{N^{1-\alpha_r}}{T}
+
\frac{N^{\frac12-\alpha_r}}{\sqrt T}
\right),
\]
where the remainder is negligible under the stated rate conditions. Consequently,
\begin{align}
\bS^{\dag\prime}\bS^\dag
=
\frac1T
\hat\bLambda^{\dag-1}
\bQ^{\dag\prime}
\bB^{0\dag\prime}\bE^{\dag\prime}\bE^\dag\bB^{0\dag}
\bQ^\dag
\hat\bLambda^{\dag-1}
+
O_{p^\dag}\left(\frac{1}{TN^{\alpha_r}} \right)
+
o_{p^\dag}
\left(
\frac{N^{1-\alpha_r}}{T}
+
\frac{N^{\frac12-\alpha_r}}{\sqrt T}
\right).\\
\label{SS}
\end{align}
Combining \eqref{C-I} and \eqref{SS}, and using
\[
\sqrt T
\hat\bLambda^{\dag-1}
\bQ^{\dag\prime}
\left(
\frac{\bB^{0\dag\prime}\bE^{\dag\prime}\bE^\dag\bB^{0\dag}}{T}
\right)
\bQ^\dag
\hat\bLambda^{\dag-1}
=
\bar{\mathcal{G}}^{\dag}
+
o_{p^\dag}(1),
\]
yields
\[
\sqrt T(\bC^{\ast\dag}-\bI_r)
=
-\frac12\bar{\mathcal{G}}^{\dag}
+
o_{p^\dag}(1).
\]
(ii) We next consider \(\bO^{\ast\dag}-\bI_r\). Since
\[
\bQ^\dag=\bI_r+\bDelta_{\bQ}^\dag,
\qquad
\|\bDelta_{\bQ}^\dag\|_2=o_{p^\dag}(1),
\]
the Fréchet expansion of the inverse square root gives
\begin{align}
\bO^{\ast\dag}-\bI_r
=
\operatorname{skew}(\bQ^\dag-\bI_r)
+
O_{p^\dag}(\|\bQ^\dag-\bI_r\|_2^2).
\label{R*-I}
\end{align}
Here $\operatorname{skew}(\bA):=(\bA-\bA')/2$ denotes the skew-symmetric part of a square matrix $\bA$.
The quadratic remainder is negligible after multiplication by \(\sqrt T\).

To evaluate the skew-symmetric part, use the bootstrap eigen-equation
\begin{align}
\bQ^\dag
-
\bLambda^{\dag-1}\bQ^\dag\hat\bLambda^\dag
=
-
\bLambda^{\dag-1}
\bU^{0\dag\prime}\bDelta_{\bA}^\dag\hat\bU^\dag .
\label{Qlead}
\end{align}
The leading term of $\sqrt{T}\bLambda^{\dag-1}\bU^{0\dag\prime}\bDelta_{\bA}^\dag\hat\bU^\dag$ is
\[
\frac{1}{\sqrt T}
\bLambda^{\dag-1}
\bB^{0\dag\prime}
\bE^{\dag\prime}\bE^\dag\bB^{0\dag}
\bQ^\dag
\hat\bLambda^{\dag-1}
=
\bar{\mathcal{G}}^{\dag}
+
o_{p^\dag}(1).
\]
Multiplying \eqref{Qlead} by $\sqrt T$ gives, for $k\neq\ell$,
\[
\sqrt T\Bigl(1-\frac{\hat\lambda_\ell^\dag}{\lambda_k^\dag}\Bigr)\bQ_{k\ell}^\dag
=
-\bar{\mathcal{G}}^{\dag}_{k\ell}
+
o_{p^\dag}(1).
\]
Thus, for \(k\neq \ell\),
\[
\sqrt T \bQ_{k\ell}^\dag
=
m_{k\ell}^\dag \bar{\mathcal{G}}^{\dag}_{k\ell}
+
o_{p^\dag}(1),
\qquad
m_{k\ell}^\dag
=
\frac{\lambda_k^\dag}{\hat\lambda_\ell^\dag-\lambda_k^\dag}.
\]
Since \(\bar{\mathcal{G}}^{\dag}\) is symmetric,
\[
\sqrt T\operatorname{skew}(\bQ^\dag-\bI_r)
=
\frac{1}{2}
(\bM^\dag-\bM^{\dag\prime})
\odot\bar{\mathcal{G}}^{\dag}
+
o_{p^\dag}(1).
\]
Together with \eqref{R*-I}, this proves
\[
\sqrt T(\bO^{\ast\dag}-\bI_r)
=
\frac{1}{2}
(\bM^\dag-\bM^{\dag\prime})
\odot\bar{\mathcal{G}}^{\dag}
+
o_{p^\dag}(1).
\]

Finally, Lemma \ref{lem:bias_Hhat3} implies
$
\bar{\mathcal{G}}^{\dag}
\CP
c_2\bar{\bG},
$
while \eqref{eq:M-dag-cons} gives
\[
\bM^\dag
\CPb
\bM,
\quad\text{in probability}.
\]
This completes the proof.
\end{proof}
\begin{rem}\label{rem:mkl-dag}
The matrix \(\bM^\dag\) summarizes the limiting effect of the relative eigengaps in the bootstrap sample. In the strictly ordered case
$\alpha_1>\cdots>\alpha_r$,
\[
m_{k\ell}^\dag
\CPb
-1,\qquad k<\ell,
\]
and
\[
m_{k\ell}^\dag
\CPb
0,\qquad k>\ell,
\]
in probability, while $m_{kk}^\dag=0$ by definition.
Hence $\bQ^\dag$ is asymptotically upper triangular under the
bootstrap law. If two factors have the same divergence rate, the corresponding entries of \(\bM^\dag\) are determined by the finite nonzero limits of the associated eigengap ratios
\[
m_{k\ell}^\dag
=
\frac{\lambda_k^\dag}{\hat\lambda_\ell^\dag-\lambda_k^\dag}.
\]
\end{rem}

The following Lemmas are analogous to results in \citet[Lemma 2, Lemma B.3 and Proposition 1]{jiang2023revisiting} and \citet[Lemmas A.1 -- A.2]{jiang2024Mw}. We provide only the main arguments, as the complete derivations closely follow those in the cited literature.

\begin{lem}
	\label{lem:Rotation} Let $\Delta_{NT}$ and $\tilde{\Delta}_{NT}$ be as defined at the beginning of Section~\ref{app:rotation-specific-proofs}.
  Suppose that Assumptions \ref{ass:eigen}–\ref{ass:Aug_errors} and \ref{ass:errors_b}--\ref{ass:bootstrap_b} hold. If $\frac{N^{1-\alpha_r}}{T}\to 0$, then, the following statements hold in probability, as $N, T \to \infty$,
		\begin{flalign*}
			&(i) ~		\left\|  \frac{1}{T} \bE^{\dag \prime}(\hat{\bF}^{\dag }- \bF^{0\dag}\tilde{\bH}_b^{\dag }) \right\|_{\F} 
			=  O_{p^{\dag}} \left( \frac{N^{1-\frac{1}{2}\alpha_r}}{T}  \right)+ O_{p^{\dag}} \left(  N^{-\frac{1}{2}\alpha_{r}} \right), & \\
  & (ii) ~		\left\|  \frac{1}{T} \bB^{0\dag\prime}\bE^{\dag \prime}(\hat{\bF}^{\dag }- \bF^{0\dag}\tilde{\bH}_b^{\dag }) \right\|_{\F} 
			= O_{p^{\dag}} \left( N^{\frac{1}{2}\alpha_{1}-\frac{1}{2}\alpha_{r}}\right)+ O_{p^{\dag}} \left(  N^{\frac{1}{2}\alpha_1-\frac{1}{2}\alpha_r} \frac{N^{1-\frac{1}{2}\alpha_r}}{T} \right), & \\
   &(iii) ~	 	\left\|  \frac{1}{T} \bN^{-\frac{1}{2}} \bB^{0\dag\prime}\bE^{\dag \prime}(\hat{\bF}^{\dag }- \bF^{0\dag}\tilde{\bH}_b^{\dag }) \right\|_{\F} 
			= O_{p^{\dag}} \left( N^{ -\frac{1}{2}\alpha_{r}}\right)+ O_{p^{\dag}} \left(    \frac{N^{1- \alpha_r}}{T} \right), & \\
    &(iv) ~		\left\| \frac{1}{T}\hat{\bF}^{\dag \prime}(\hat{\bF}^{\dag }-\bF^{0\dag}\tilde{\bH}^{\dag }) \right\|_{\F} 
		= O_{p^{\dag}} \left(\Delta_{NT} \right), & \\
   &(v)~ \left\|	 \frac{{\hat{\bF}}^{\dag \prime}\bF^{0\dag}}{T}-\bI_r \right\|_{\F} = O_{p^{\dag}}(\tilde{\Delta}_{NT}) .
	\end{flalign*}	
\end{lem}

\begin{proof}[Proof of Lemma \ref{lem:Rotation}] The proof follows directly from \citet[Lemmas B.3, Proposition 1]{jiang2023revisiting}.
\end{proof}

\begin{lem} Suppose that Assumptions \ref{ass:eigen}–\ref{ass:Aug_errors} and \ref{ass:errors_b}--\ref{ass:bootstrap_b} hold.  If $\frac{N^{1-\alpha_r}}{T}\to 0$, then, the following statements hold in probability, as $N, T \to \infty$,
\label{lem:F*epsilon}
\begin{flalign}
     &(i) ~ \left\| \frac{1}{\sqrt{T}} (\hat{\bF}^{\dag }-\bF^{0\dag}\tilde{\bH}^{\dag })'\bepsilon^{\dag }  \right\|_{\F}=O_{p^{\dag}} \left(\sqrt{T} N^{-\frac{3}{2}\alpha_r}\right) +O_{p^{\dag}}\left( \frac{N^{1-\alpha_r}}{\sqrt{T}} \right)+O_{p^{\dag}}\left( N^{\frac{1}{2}- \alpha_r} \right),\label{eq:F*ep_H}&\\
   &(ii) ~ \left\| \frac{1}{\sqrt{T}} (\hat{\bF}^{\dag }-\bF^{0\dag}\tilde{\bH}_q^{\dag })'\bepsilon^{\dag }  \right\|_{\F}=O_{p^{\dag}} \left(\sqrt{T} N^{-\frac{3}{2}\alpha_r}\right) +O_{p^{\dag}}\left( \frac{N^{1-\alpha_r}}{\sqrt{T}} \right)+O_{p^{\dag}}\left( N^{\frac{1}{2}- \alpha_r} \right).\label{eq:F*ep_H3}&
\end{flalign}
\end{lem}
\begin{proof}[Proof of Lemma \ref{lem:F*epsilon}] The proof follows directly from \citet[Lemma A.2]{jiang2024Mw}.
\end{proof}

\section{Additional Monte Carlo results}\label{app:addMC}

This appendix supplements the Monte Carlo evidence in Section~\ref{sec: MC}.
We first report fixed-restriction and rotation-invariant results for the heterogeneous-rate designs that are omitted from the main text.
We then report additional bias and target-relative coverage results for the equal-rate designs, followed by stronger equal-rate factors, further rotation-invariant checks, movement of the data-dependent targets, bootstrap critical values, rotation diagnostics, RMSE, and a validation of the dependent bootstrap scheme used in the empirical application.
Except for the dependent-bootstrap experiment in Section~\ref{app:depboot}, the design, methods, bootstrap implementation, and notation are as in Section~\ref{subsec:MCdesign}.

\subsection{Fixed-restriction inference under heterogeneous weak factors}
\label{app:fixed-difA}

We first examine inference on the fixed population restriction $H_0:\gamma_2^0=0$ under heterogeneous divergence rates. 
For this experiment, we set $\bgamma^0=(1,0)'$ and $\bbeta=(1,1)'$, so that the fixed null holds in every Monte Carlo replication. 
For each $\bR\in\{\hat{\bH},\hat{\bH}_q,\bH\}$, the observed studentized statistic is formed against the fixed null value zero, while bootstrap critical values are obtained from the corresponding rotation-specific recentered bootstrap distribution. 
Figure~\ref{fig:size.zero.difA} reports empirical rejection frequencies for nominal 5\% two-sided tests. 
For the procedures based on $\hat{\bH}$ and $\hat{\bH}_q$, rejection frequencies exceed the nominal level, particularly when the second factor is weaker and the sample size is small. 
The distortion remains even when the proposed bootstrap bias correction is combined with bootstrap critical values, although it decreases as $T$ increases. 
In some cases, normal critical values yield rejection frequencies closer to the nominal level than bootstrap critical values. 
This does not indicate a more accurate normal approximation: the bootstrap approximates the estimation error relative to the rotation-specific target, whereas the test concerns the fixed null $H_0:\gamma_2^0=0$. 
The apparent advantage of the normal approximation therefore reflects an offsetting finite-sample error rather than better approximation of the relevant fixed-null distribution. 
In contrast, the procedure based on the fixed-target construction $\bH$ produces rejection frequencies close to 5\% throughout the heterogeneous-rate designs, with the proposed double-bootstrap procedure ranging from 4.4\% to 6.4\%. 
This difference is consistent with the decomposition in \eqref{eq:fixed-R-decomp}: accurate bootstrap approximation of the estimation error relative to a rotation-specific target does not by itself ensure accurate inference on $\gamma_2^0$, which also depends on the gap between that target and the fixed population coefficient. 

\begin{figure}[!htb]
\centering
\begin{subfigure}[b]{0.32\textwidth}\centering
\includegraphics[width=\textwidth]{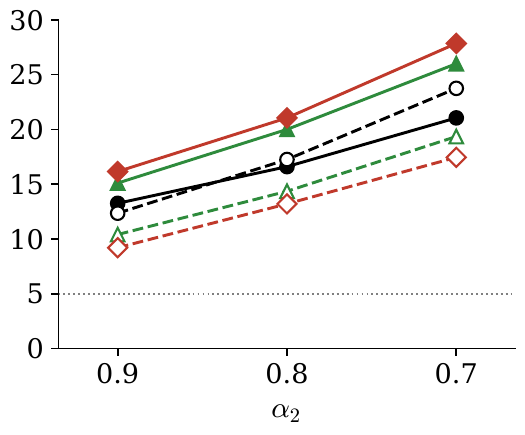}
\caption{$\hat{\bH}$, $T=50$}
\end{subfigure}\hfill
\begin{subfigure}[b]{0.32\textwidth}\centering
\includegraphics[width=\textwidth]{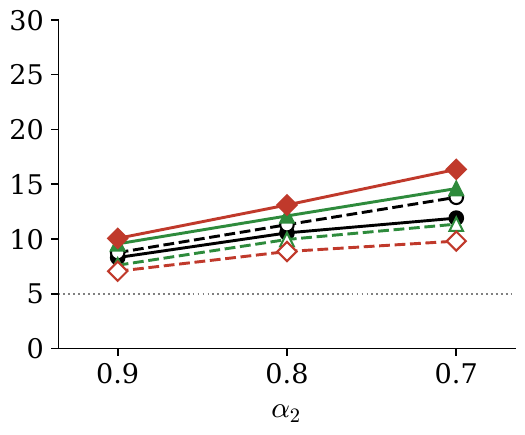}
\caption{$\hat{\bH}$, $T=100$}
\end{subfigure}\hfill
\begin{subfigure}[b]{0.32\textwidth}\centering
\includegraphics[width=\textwidth]{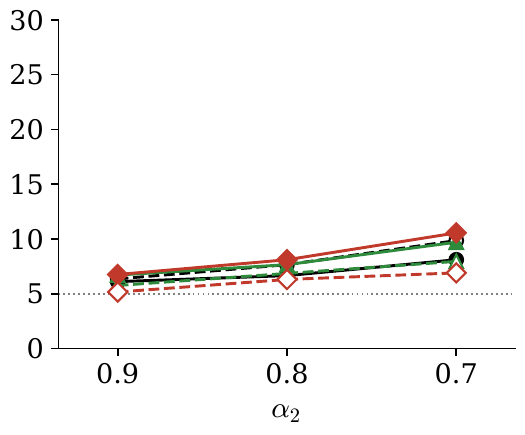}
\caption{$\hat{\bH}$, $T=200$}
\end{subfigure}
\begin{subfigure}[b]{0.32\textwidth}\centering
\includegraphics[width=\textwidth]{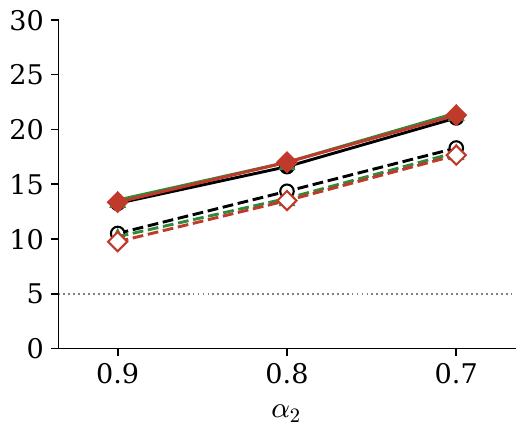}
\caption{$\hat{\bH}_q$, $T=50$}
\end{subfigure}\hfill
\begin{subfigure}[b]{0.32\textwidth}\centering
\includegraphics[width=\textwidth]{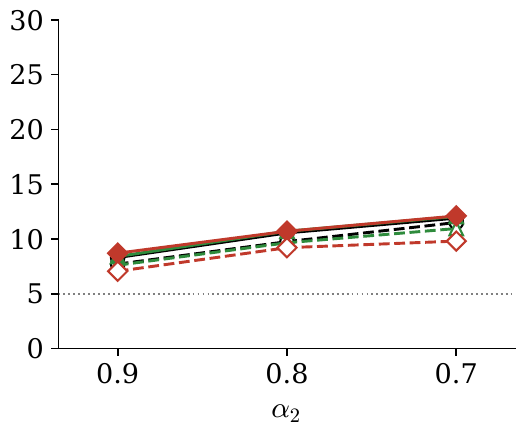}
\caption{$\hat{\bH}_q$, $T=100$}
\end{subfigure}\hfill
\begin{subfigure}[b]{0.32\textwidth}\centering
\includegraphics[width=\textwidth]{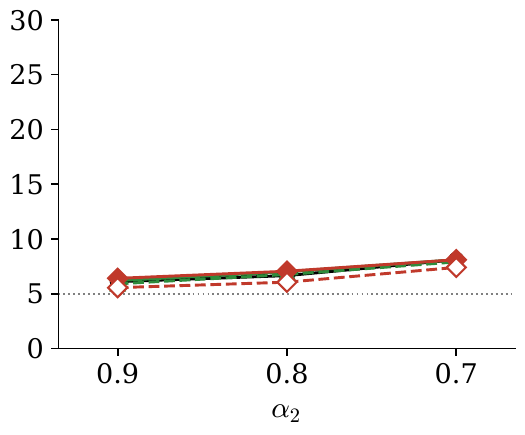}
\caption{$\hat{\bH}_q$, $T=200$}
\end{subfigure}
\begin{subfigure}[b]{0.32\textwidth}\centering
\includegraphics[width=\textwidth]{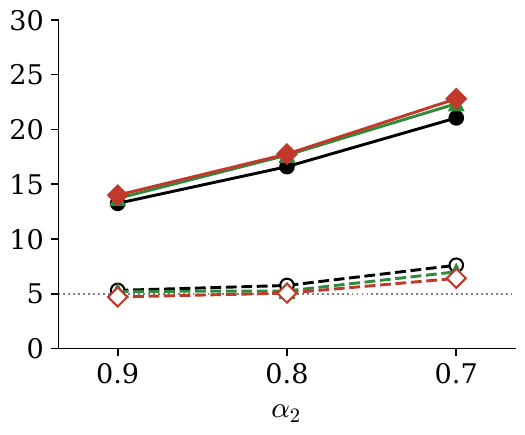}
\caption{$\bH$, $T=50$}
\end{subfigure}\hfill
\begin{subfigure}[b]{0.32\textwidth}\centering
\includegraphics[width=\textwidth]{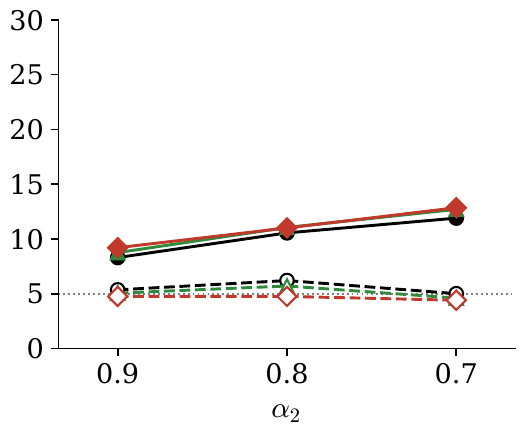}
\caption{$\bH$, $T=100$}
\end{subfigure}\hfill
\begin{subfigure}[b]{0.32\textwidth}\centering
\includegraphics[width=\textwidth]{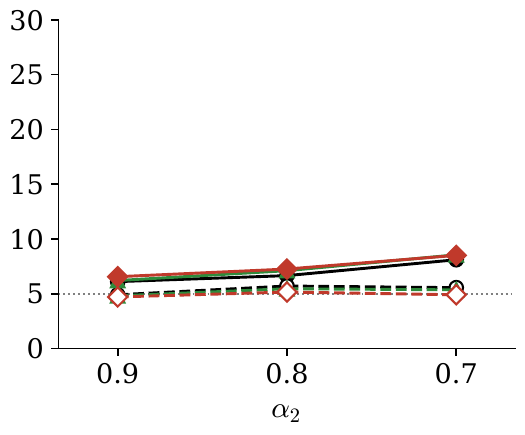}
\caption{$\bH$, $T=200$}
\end{subfigure}
\includegraphics[width=0.95\textwidth]{figs/legend_zero_methods.pdf}
\caption{Size for $H_0:\gamma_2^0=0$}
\label{fig:size.zero.difA}
\vspace{2pt}
\begin{minipage}{0.95\textwidth}
\footnotesize
\justifying
\noindent Notes: Rejection frequencies are percentages for nominal 5\% two-sided tests of the fixed population null $H_0:\gamma_2^0=0$ under $\bgamma^0=(1,0)'$.
The top, middle, and bottom rows use the bias correction and bootstrap construction associated with $\hat{\bH}$, $\hat{\bH}_q$, and $\bH$, respectively; the row labels indicate the rotation used in constructing the estimator and its bootstrap distribution, while the null hypothesis is $H_0:\gamma_2^0=0$ in every panel.
For each method, the observed studentized statistic is formed relative to zero using the corresponding uncorrected, analytically bias-corrected, or bootstrap bias-corrected estimator.
Tests using normal critical values reject when the observed statistic lies outside $[-1.96,1.96]$.
Tests using bootstrap critical values reject when the observed statistic lies below the 2.5th percentile or above the 97.5th percentile of the corresponding recentered studentized bootstrap distribution. 
For the bootstrap bias-corrected estimator, these bootstrap critical values are obtained from the corresponding double-bootstrap statistic.
The horizontal dotted line denotes the nominal 5\% level.
\end{minipage}
\end{figure}

\subsection{Rotation-invariant checks under heterogeneous weak factors}
\label{app:beta-difA}

To separate the fixed-coordinate effect in Figure~\ref{fig:size.zero.difA} from a general failure of the bootstrap approximation, we repeat the fixed-null experiment for the rotation-invariant restriction $H_0:\beta_1=0$, setting $\bgamma^0=(1,1)'$ and $\bbeta=(0,1)'$.
Figure~\ref{fig:size.zero.beta.difA} shows rejection frequencies close to the nominal level when bootstrap critical values are used, including in the weaker-factor designs.
Since $\beta_1$ is unchanged by factor rotation, this comparison indicates that the distortions in Figure~\ref{fig:size.zero.difA} arise from the distinction between the data-dependent factor coordinates and the fixed population coordinate rather than from a general failure of the bootstrap distributional approximation.

\begin{figure}[!htb]
\centering
\begin{subfigure}[b]{0.32\textwidth}\centering
\includegraphics[width=\textwidth]{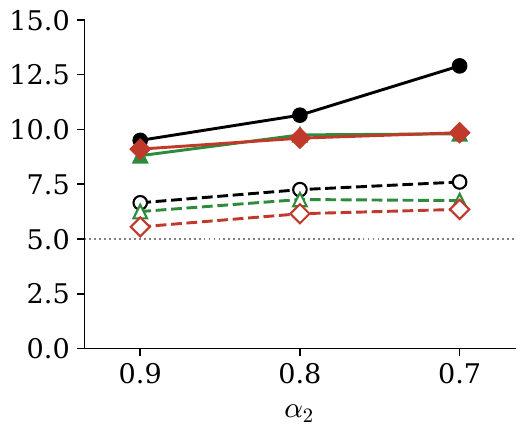}
\caption{$T=50$}
\end{subfigure}\hfill
\begin{subfigure}[b]{0.32\textwidth}\centering
\includegraphics[width=\textwidth]{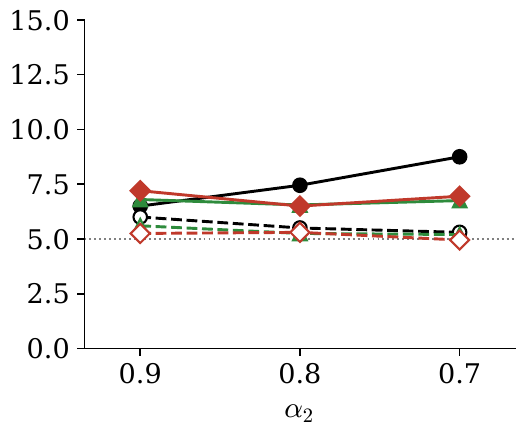}
\caption{$T=100$}
\end{subfigure}\hfill
\begin{subfigure}[b]{0.32\textwidth}\centering
\includegraphics[width=\textwidth]{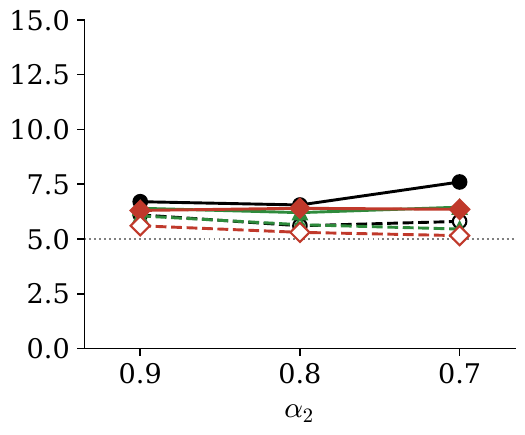}
\caption{$T=200$}
\end{subfigure}
\includegraphics[width=0.95\textwidth]{figs/legend_zero_methods.pdf}
\caption{Size for $H_0:\beta_1=0$}
\label{fig:size.zero.beta.difA}
\vspace{2pt}
\begin{minipage}{0.95\textwidth}
\footnotesize
\justifying
\noindent Notes: Rejection frequencies are percentages for nominal 5\%
two-sided tests of $H_0:\beta_1=0$ under
$\bgamma^0=(1,1)'$ and $\bbeta=(0,1)'$. 
Since $\beta_1$ is invariant to factor rotation, results are reported using the $\bH$ construction.
Tests using normal critical values reject when the observed statistic lies outside $[-1.96,1.96]$, while tests using bootstrap critical values reject when the observed statistic lies below the 2.5th percentile or above the 97.5th percentile of the corresponding recentered studentized bootstrap distribution.
For the bootstrap bias-corrected estimator, these bootstrap critical values are obtained from the corresponding double-bootstrap statistic.
The horizontal dotted line denotes the nominal 5\% level.
\end{minipage}
\end{figure}

We next examine coverage for the coefficient $\beta_1$ under the distributional-approximation design.
Figure~\ref{fig:coverage.beta.difA} reports coverage of nominal 95\% confidence intervals.
Although $\beta_1$ itself is unaffected by factor rotation, factor-estimation error still enters its estimator.
Normal intervals show moderate undercoverage in the smaller samples, particularly when the second factor is weak, whereas studentized bootstrap intervals bring coverage closer to the nominal level.
Thus, the improvement from bootstrap distributional approximation is not confined to factor coefficients.

\begin{figure}[!htb]
\centering
\begin{subfigure}[b]{0.32\textwidth}\centering
\includegraphics[width=\textwidth]{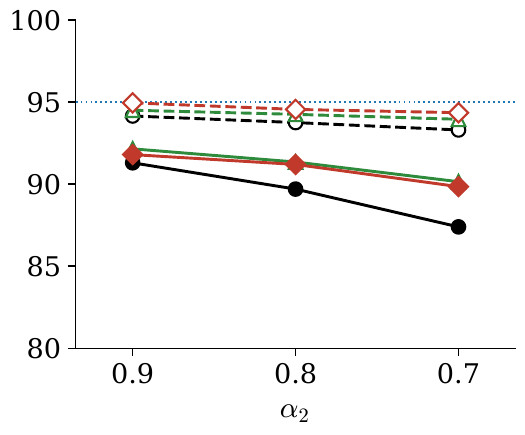}
\caption{$T=50$}
\end{subfigure}\hfill
\begin{subfigure}[b]{0.32\textwidth}\centering
\includegraphics[width=\textwidth]{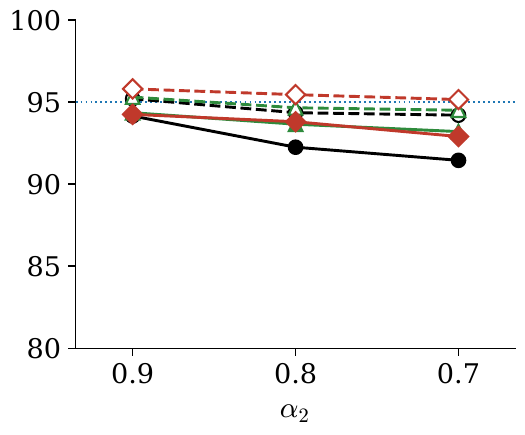}
\caption{$T=100$}
\end{subfigure}\hfill
\begin{subfigure}[b]{0.32\textwidth}\centering
\includegraphics[width=\textwidth]{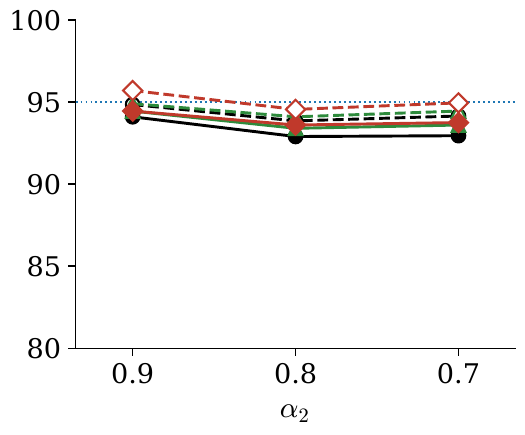}
\caption{$T=200$}
\end{subfigure}
\includegraphics[width=0.95\textwidth]{figs/legend_methods.pdf}
\caption{Coverage of 95\% confidence intervals for $\beta_1$}
\label{fig:coverage.beta.difA}
\vspace{2pt}
\begin{minipage}{0.95\textwidth}
\footnotesize
\justifying
\noindent Notes: Since $\beta_1$ is invariant to factor rotation, the results are reported using the $\bH$ representation.
Studentized bootstrap intervals use the 2.5th and 97.5th percentiles of the corresponding bootstrap distribution. 
The horizontal dotted line denotes the nominal 95\% coverage probability.
\end{minipage}
\end{figure}

\subsection{Bias under equal weak-factor strengths}\label{app:bias-eqA08}

Figure~\ref{fig:bias.eqA08} reports the bias of the four estimators relative to the conventional target $\gamma_{\hat{\bH},2}$ in the equal-rate designs.
The uncorrected estimator exhibits substantial negative bias, particularly in the smaller samples.
All three corrections reduce the bias substantially.
The proposed bootstrap correction and the GP correction are generally close, with the proposed correction often giving the smaller absolute bias, while the analytical correction leaves somewhat greater residual bias.
The differences diminish rapidly as the sample size increases.

\begin{figure}[!htb]
\centering
\begin{subfigure}[b]{0.32\textwidth}\centering
\includegraphics[width=\textwidth]{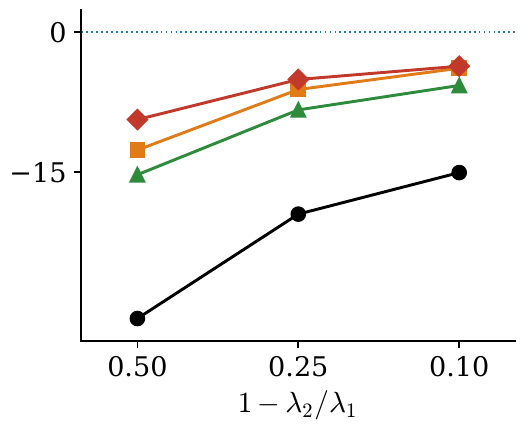}
\caption{$T=50$}
\end{subfigure}\hfill
\begin{subfigure}[b]{0.32\textwidth}\centering
\includegraphics[width=\textwidth]{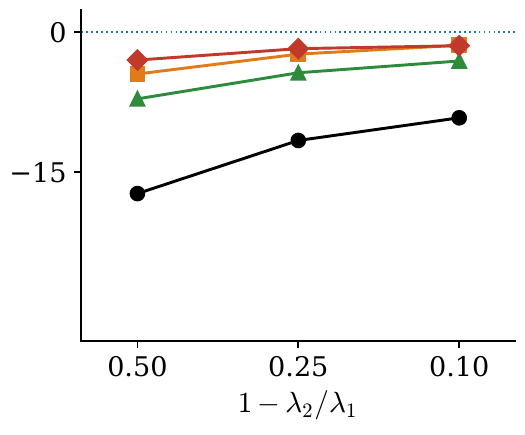}
\caption{$T=100$}
\end{subfigure}\hfill
\begin{subfigure}[b]{0.32\textwidth}\centering
\includegraphics[width=\textwidth]{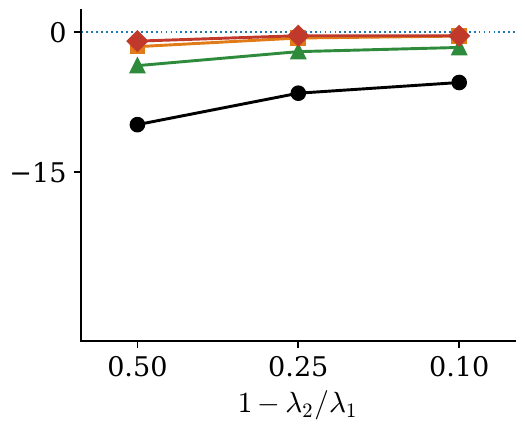}
\caption{$T=200$}
\end{subfigure}
\includegraphics[width=0.95\textwidth]{figs/legend_estimators.pdf}
\caption{Bias relative to $\gamma_{\hat{\bH},2}$}
\label{fig:bias.eqA08}
\vspace{2pt}
\begin{minipage}{0.95\textwidth}
\footnotesize
\justifying
\noindent Notes: Bias is multiplied by 100.
These designs set $\alpha_1=\alpha_2=0.8$ and have relative signal eigengaps $0.50$, $0.25$, and $0.10$, respectively.
The horizontal line denotes zero bias.
\end{minipage}
\end{figure}

\subsection{Target-relative coverage under close eigengaps}\label{app:coverage-eqA08}

Figure~\ref{fig:coverage.eqA08} reports coverage of nominal 95\% confidence intervals in the equal-rate designs for the three targets considered in the main text.
For the conventional target $\gamma_{\hat{\bH},2}$, normal intervals can undercover severely, while bootstrap critical values provide a large improvement.
For the proposed bootstrap bias-corrected estimator with bootstrap critical values, coverage ranges from 90.8\% to 95.2\% across these designs and sample sizes.
For the alternative target $\gamma_{\hat{\bH}_q,2}$, coverage is remarkably stable, with the proposed procedure ranging from 94.7\% to 96.1\%.
Thus, even when the estimated PC direction is unstable across samples, estimation relative to the realized data-dependent target remains well behaved.

For the fixed target $\gamma_2^0$, the effect of the eigengap is much stronger.
The proposed double-bootstrap procedure gives coverage close to nominal at relative eigengaps $0.50$ and $0.25$, but at $0.10$ coverage is 80.6\%, 86.8\%, and 94.3\% at $T=50$, $100$, and $200$, respectively.
This deterioration is distinct from the target-relative behavior under $\hat{\bH}$ and $\hat{\bH}_q$ and reflects the additional variation between the estimated PC direction and the fixed population orientation.
The improvement with $T$ is nevertheless substantial, with coverage in the close design approaching the nominal level by $T=200$.

\begin{figure}[!htb]
\centering
\begin{subfigure}[b]{0.32\textwidth}\centering
\includegraphics[width=\textwidth]{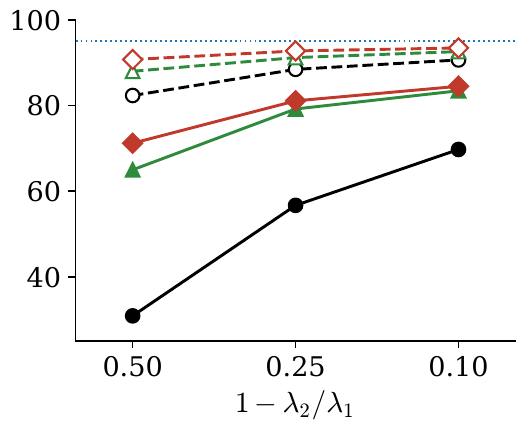}
\caption{$\gamma_{\hat{\bH},2}$, $T=50$}
\end{subfigure}\hfill
\begin{subfigure}[b]{0.32\textwidth}\centering
\includegraphics[width=\textwidth]{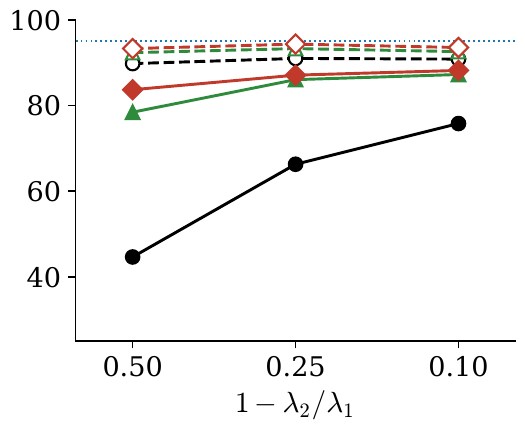}
\caption{$\gamma_{\hat{\bH},2}$, $T=100$}
\end{subfigure}\hfill
\begin{subfigure}[b]{0.32\textwidth}\centering
\includegraphics[width=\textwidth]{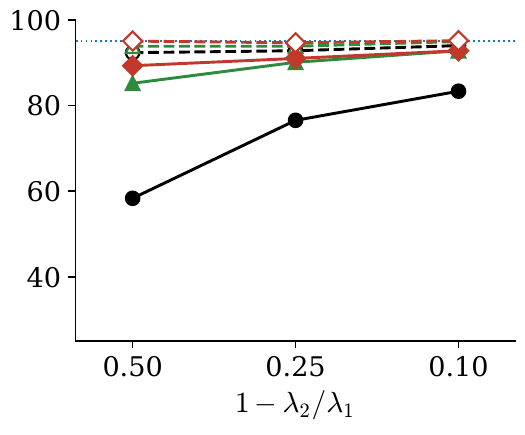}
\caption{$\gamma_{\hat{\bH},2}$, $T=200$}
\end{subfigure}
\begin{subfigure}[b]{0.32\textwidth}\centering
\includegraphics[width=\textwidth]{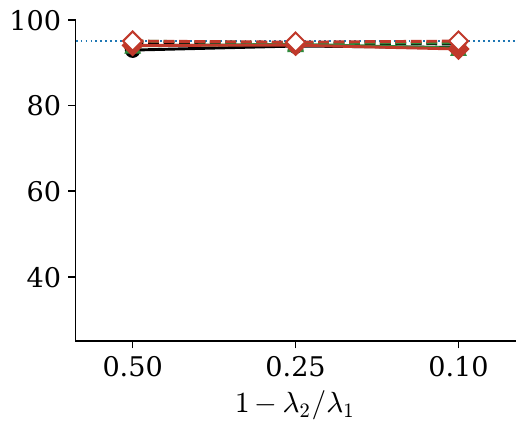}
\caption{$\gamma_{\hat{\bH}_q,2}$, $T=50$}
\end{subfigure}\hfill
\begin{subfigure}[b]{0.32\textwidth}\centering
\includegraphics[width=\textwidth]{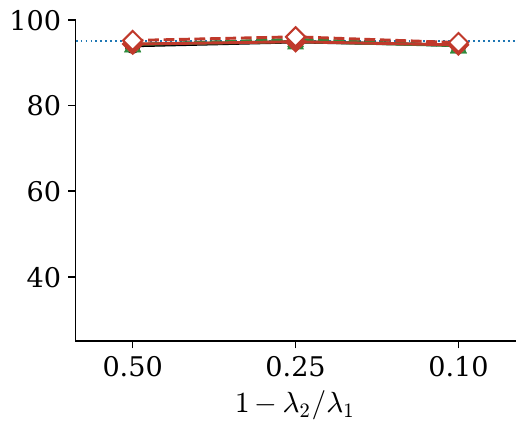}
\caption{$\gamma_{\hat{\bH}_q,2}$, $T=100$}
\end{subfigure}\hfill
\begin{subfigure}[b]{0.32\textwidth}\centering
\includegraphics[width=\textwidth]{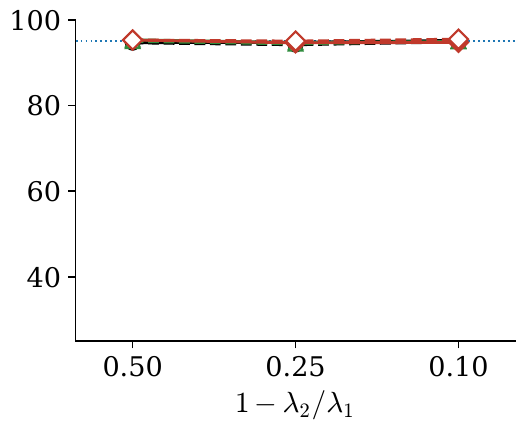}
\caption{$\gamma_{\hat{\bH}_q,2}$, $T=200$}
\end{subfigure}
\begin{subfigure}[b]{0.32\textwidth}\centering
\includegraphics[width=\textwidth]{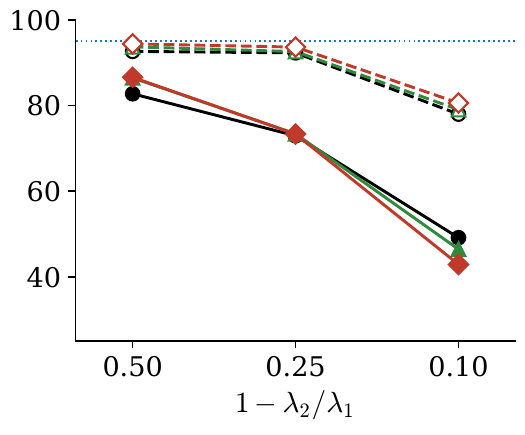}
\caption{$\gamma_2^0$, $T=50$}
\end{subfigure}\hfill
\begin{subfigure}[b]{0.32\textwidth}\centering
\includegraphics[width=\textwidth]{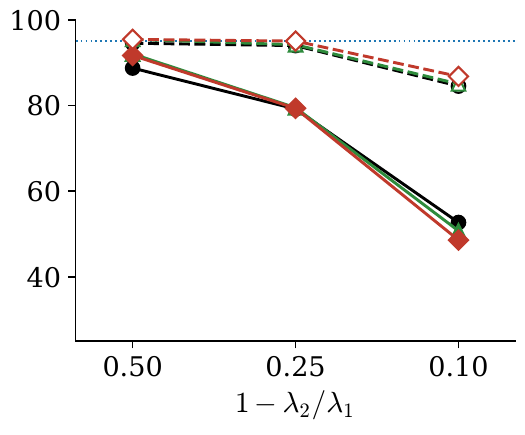}
\caption{$\gamma_2^0$, $T=100$}
\end{subfigure}\hfill
\begin{subfigure}[b]{0.32\textwidth}\centering
\includegraphics[width=\textwidth]{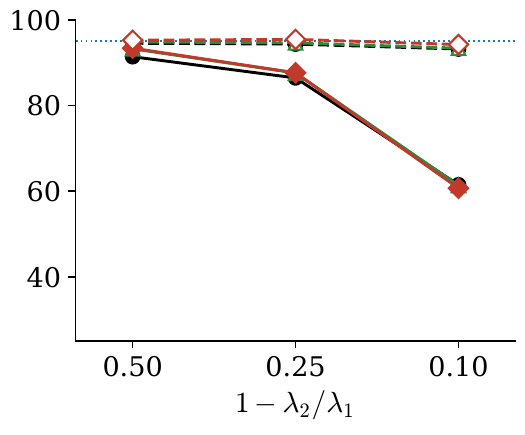}
\caption{$\gamma_2^0$, $T=200$}
\end{subfigure}
\includegraphics[width=0.95\textwidth]{figs/legend_methods.pdf}
\caption{Coverage under close eigengaps}
\label{fig:coverage.eqA08}
\vspace{2pt}
\begin{minipage}{0.95\textwidth}
\footnotesize
\justifying
\noindent Notes: The top, middle, and bottom rows correspond to $\gamma_{\hat{\bH},2}$, $\gamma_{\hat{\bH}_q,2}$, and $\gamma_2^0$, respectively.
For the two data-dependent rotations, coverage is evaluated relative to the realized rotation-specific target in each Monte Carlo replication.
For $\bH$, coverage is evaluated relative to the fixed population coefficient $\gamma_2^0$.
The horizontal dotted line denotes 95\% coverage.
\end{minipage}
\end{figure}

\subsection{Stronger equal-rate factors}\label{app:eqA10}

We next repeat the equal-rate experiment with stronger factors, setting $\alpha_1=\alpha_2=1$ and $(d_1,d_2)=(1,0.50),(1,0.75),(1,0.90)$, so that the relative eigengaps remain $0.50$, $0.25$, and $0.10$.
Figure~\ref{fig:coverage.eqA10} shows that the qualitative pattern of the weaker equal-rate designs remains, but the distortions are generally milder.
For $\gamma_{\hat{\bH}_q,2}$, coverage remains close to nominal throughout, while the proposed double-bootstrap coverage for $\gamma_{\hat{\bH},2}$ ranges from 93.5\% to 95.8\%.
For the fixed target, the wide and moderate eigengap designs are also well behaved.
At the smallest relative eigengap, however, coverage for the proposed procedure is 86.2\% at $T=50$, 93.2\% at $T=100$, and 94.6\% at $T=200$.
Thus, stronger factors reduce but do not eliminate the finite-sample difficulty created by a nearly degenerate signal eigenspace.

\begin{figure}[!htb]
\centering
\begin{subfigure}[b]{0.32\textwidth}\centering
\includegraphics[width=\textwidth]{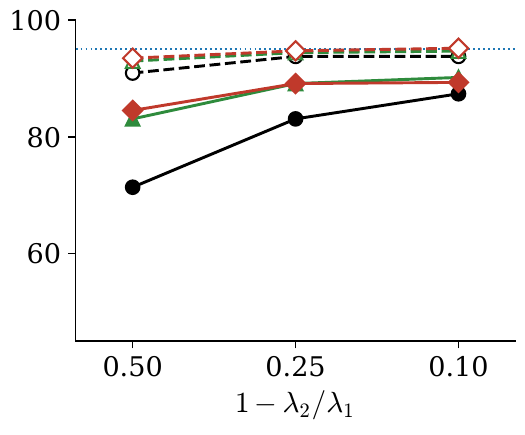}
\caption{$\gamma_{\hat{\bH},2}$, $T=50$}
\end{subfigure}\hfill
\begin{subfigure}[b]{0.32\textwidth}\centering
\includegraphics[width=\textwidth]{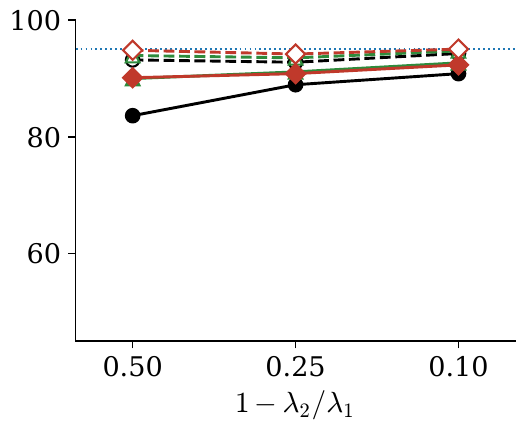}
\caption{$\gamma_{\hat{\bH},2}$, $T=100$}
\end{subfigure}\hfill
\begin{subfigure}[b]{0.32\textwidth}\centering
\includegraphics[width=\textwidth]{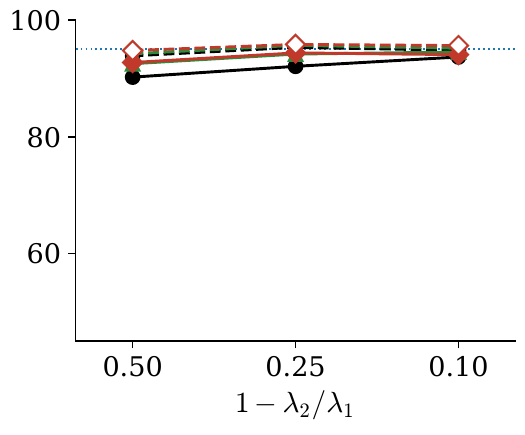}
\caption{$\gamma_{\hat{\bH},2}$, $T=200$}
\end{subfigure}
\begin{subfigure}[b]{0.32\textwidth}\centering
\includegraphics[width=\textwidth]{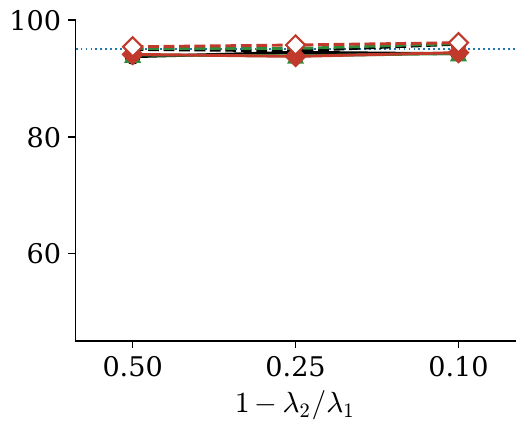}
\caption{$\gamma_{\hat{\bH}_q,2}$, $T=50$}
\end{subfigure}\hfill
\begin{subfigure}[b]{0.32\textwidth}\centering
\includegraphics[width=\textwidth]{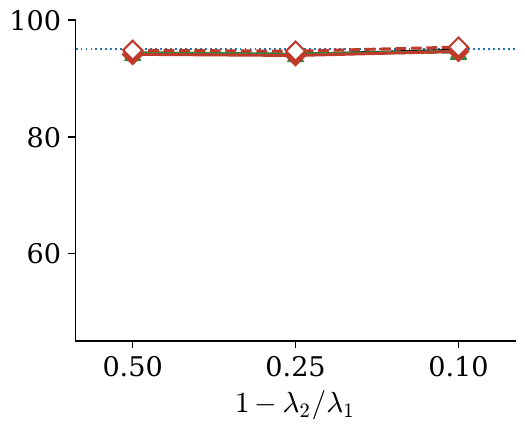}
\caption{$\gamma_{\hat{\bH}_q,2}$, $T=100$}
\end{subfigure}\hfill
\begin{subfigure}[b]{0.32\textwidth}\centering
\includegraphics[width=\textwidth]{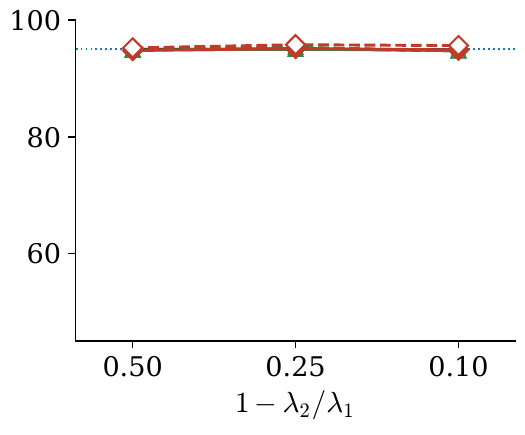}
\caption{$\gamma_{\hat{\bH}_q,2}$, $T=200$}
\end{subfigure}
\begin{subfigure}[b]{0.32\textwidth}\centering
\includegraphics[width=\textwidth]{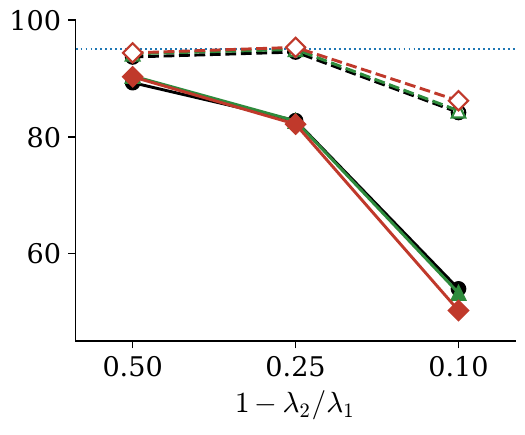}
\caption{$\gamma_2^0$, $T=50$}
\end{subfigure}\hfill
\begin{subfigure}[b]{0.32\textwidth}\centering
\includegraphics[width=\textwidth]{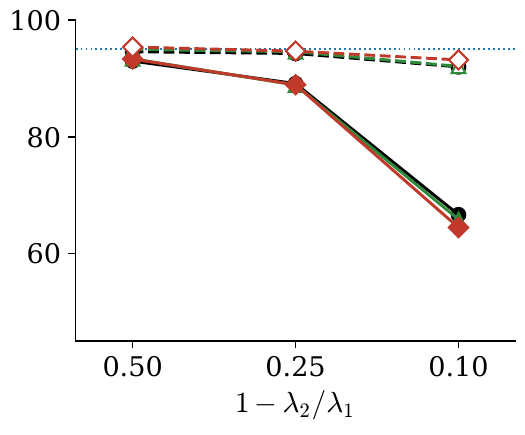}
\caption{$\gamma_2^0$, $T=100$}
\end{subfigure}\hfill
\begin{subfigure}[b]{0.32\textwidth}\centering
\includegraphics[width=\textwidth]{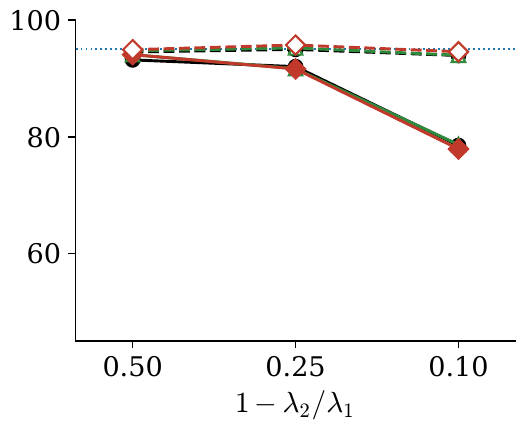}
\caption{$\gamma_2^0$, $T=200$}
\end{subfigure}
\includegraphics[width=0.95\textwidth]{figs/legend_methods.pdf}
\caption{Coverage under stronger equal-rate factors}
\label{fig:coverage.eqA10}
\vspace{2pt}
\begin{minipage}{0.95\textwidth}
\footnotesize
\justifying
\noindent Notes: The designs have $\alpha_1=\alpha_2=1$ and relative eigengaps $0.50$, $0.25$, and $0.10$, respectively.
The target definitions and coverage evaluation are the same as in Figure~\ref{fig:coverage.eqA08}.
The horizontal dotted line denotes 95\% coverage.
\end{minipage}
\end{figure}

\subsection{Rotation-invariant check under close eigengaps}\label{app:beta}

Figure~\ref{fig:size.beta.eqA08} reports the fixed-null experiment for $H_0:\beta_1=0$ under the equal-rate designs.
Because $\beta_1$ is invariant to factor rotation, this experiment provides a control for distortions that arise specifically from inference on individual factor coordinates.
Normal inference overrejects in the smaller samples, with the uncorrected test ranging as high as 13.7\%.
Bootstrap critical values substantially improve size control.
For the proposed bootstrap bias-corrected procedure, rejection frequencies range from 4.3\% to 6.0\% over all equal-rate design and sample-size combinations.
The good performance for $\beta_1$ contrasts sharply with the fixed-null results for $\gamma_2^0$ under a small eigengap and supports interpreting the latter as a consequence of rotation uncertainty rather than a general failure of the bootstrap approximation.

\begin{figure}[!htb]
\centering
\begin{subfigure}[b]{0.32\textwidth}\centering
\includegraphics[width=\textwidth]{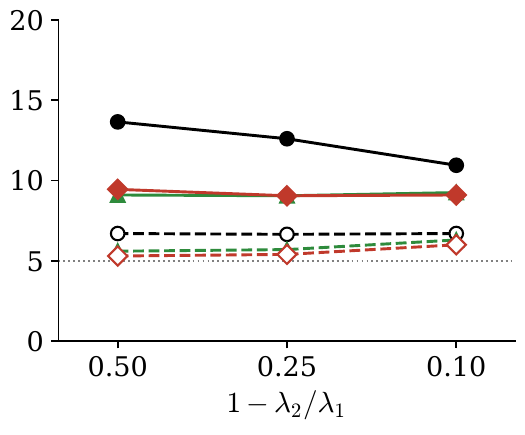}
\caption{$T=50$}
\end{subfigure}\hfill
\begin{subfigure}[b]{0.32\textwidth}\centering
\includegraphics[width=\textwidth]{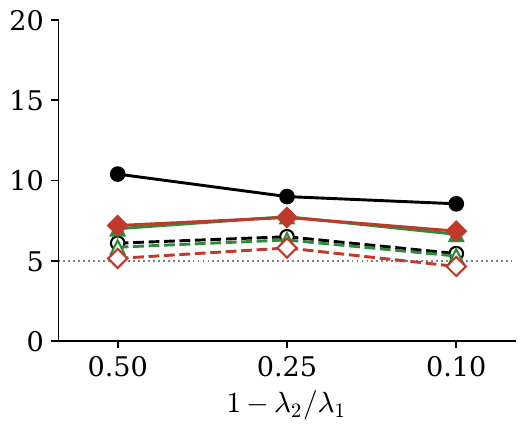}
\caption{$T=100$}
\end{subfigure}\hfill
\begin{subfigure}[b]{0.32\textwidth}\centering
\includegraphics[width=\textwidth]{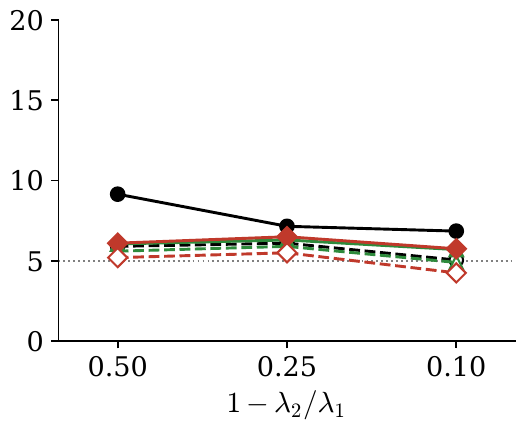}
\caption{$T=200$}
\end{subfigure}
\includegraphics[width=0.95\textwidth]{figs/legend_zero_methods.pdf}
\caption{Size for $H_0:\beta_1=0$ under close eigengaps}
\label{fig:size.beta.eqA08}
\vspace{2pt}
\begin{minipage}{0.95\textwidth}
\footnotesize
\justifying
\noindent Notes: Rejection frequencies are percentages for nominal 5\% two-sided tests under $\bgamma^0=(1,1)'$ and $\bbeta=(0,1)'$.
Since $\beta_1$ is invariant to factor rotation, results are reported using the $\bH$ construction.
The horizontal dotted line denotes the nominal 5\% level.
\end{minipage}
\end{figure}

\subsection{Movement of the rotation-specific targets}\label{app:target-dispersion}

The fixed-null experiment also allows us to examine directly how much the data-dependent target moves away from the fixed population coordinate.
Under $\bgamma^0=(1,0)'$, the population coefficient satisfies $\gamma_2^0=0$, whereas the realized coefficients $\gamma_{\hat{\bH},2}$ and $\gamma_{\hat{\bH}_q,2}$ need not equal zero.
Figure~\ref{fig:target.sd} reports the Monte Carlo standard deviation of these two rotation-specific targets.
Their Monte Carlo means are approximately zero, so the relevant feature is dispersion rather than systematic displacement.

Under heterogeneous divergence rates, the dispersion increases as the second factor becomes weaker and decreases rapidly with the sample size.
The same phenomenon becomes much stronger as the eigengap narrows under equal divergence rates.
At the relative eigengap $0.10$, the standard deviations of the two targets are approximately $0.34$--$0.36$ at $T=50$, $0.22$--$0.23$ at $T=100$, and $0.11$--$0.12$ at $T=200$.
This variation is absent from the fixed population coefficient but enters a procedure whose bootstrap distribution is constructed relative to a data-dependent target.

\begin{figure}[!htb]
\centering
\begin{subfigure}[b]{0.32\textwidth}\centering
\includegraphics[width=\textwidth]{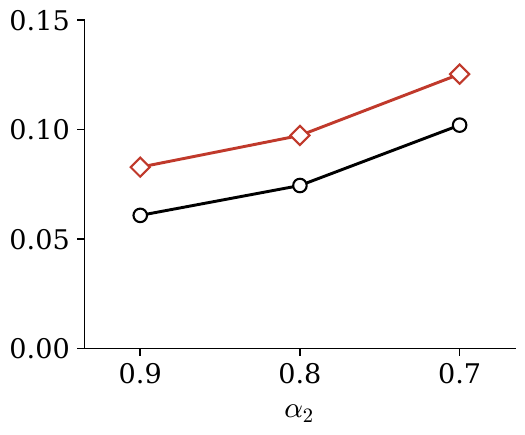}
\caption{Heterogeneous rates, $T=50$}
\end{subfigure}\hfill
\begin{subfigure}[b]{0.32\textwidth}\centering
\includegraphics[width=\textwidth]{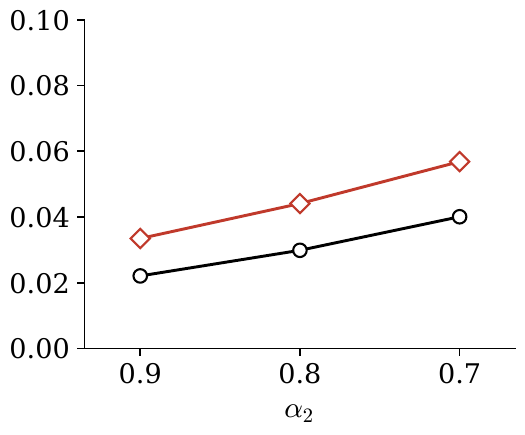}
\caption{Heterogeneous rates, $T=100$}
\end{subfigure}\hfill
\begin{subfigure}[b]{0.32\textwidth}\centering
\includegraphics[width=\textwidth]{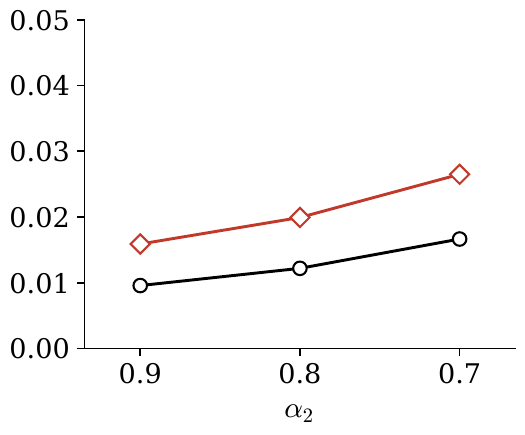}
\caption{Heterogeneous rates, $T=200$}
\end{subfigure}
\begin{subfigure}[b]{0.32\textwidth}\centering
\includegraphics[width=\textwidth]{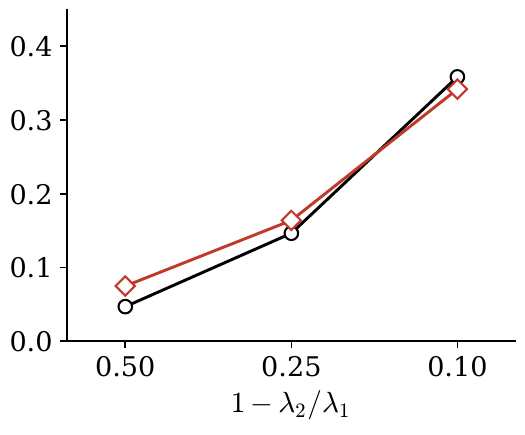}
\caption{Equal rates, $T=50$}
\end{subfigure}\hfill
\begin{subfigure}[b]{0.32\textwidth}\centering
\includegraphics[width=\textwidth]{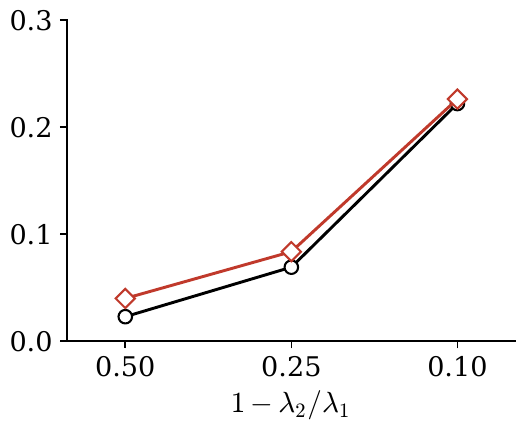}
\caption{Equal rates, $T=100$}
\end{subfigure}\hfill
\begin{subfigure}[b]{0.32\textwidth}\centering
\includegraphics[width=\textwidth]{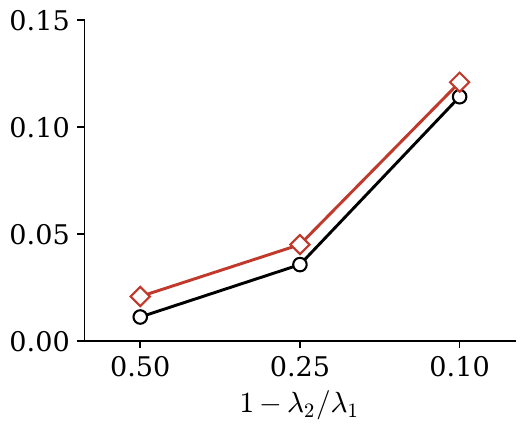}
\caption{Equal rates, $T=200$}
\end{subfigure}
\includegraphics[width=0.55\textwidth]{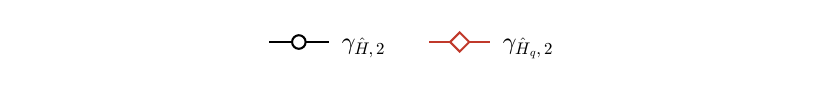}
\caption{Dispersion of rotation-specific targets under $H_0:\gamma_2^0=0$}
\label{fig:target.sd}
\vspace{2pt}
\begin{minipage}{0.95\textwidth}
\footnotesize
\justifying
\noindent Notes: Each panel reports the Monte Carlo standard deviation of $\gamma_{\hat{\bH},2}$ and $\gamma_{\hat{\bH}_q,2}$ under the fixed DGP $\bgamma^0=(1,0)'$.
The top row corresponds to heterogeneous divergence rates and the bottom row to equal divergence rates.
\end{minipage}
\end{figure}

Table~\ref{tab:Shat.difA} complements the equal-rate $\widehat{\mathrm{RU}}_2$ results reported in the main text by reporting the diagnostic under heterogeneous divergence rates.
The diagnostic is below one in most cells and declines with the sample size, although it exceeds one for the weakest heterogeneous-rate design at $T=50$.
The contrast with the relative eigengap $0.10$, where the median remains above two even at $T=200$, is consistent with the much greater instability generated by the close eigengap.

\begin{table}[!htb]
\centering
\caption{Median $\widehat{\mathrm{RU}}_2$ under heterogeneous divergence rates}
\label{tab:Shat.difA}
{\small
\begin{tabular}{lccc}
\hline\hline
$\alpha_2$ & $T=50$ & $T=100$ & $T=200$ \\
\hline
$0.9$ & 0.72 & 0.44 & 0.27 \\
$0.8$ & 0.86 & 0.56 & 0.36 \\
$0.7$ & 1.07 & 0.70 & 0.47 \\
\hline\hline
\end{tabular}}
\\[3pt]
\footnotesize
\justifying
\noindent Notes: Entries are medians across Monte Carlo replications under $H_0:\gamma_2^0=0$.
The diagnostic $\widehat{\mathrm{RU}}_2$ is defined in \eqref{eq:Shat}.
\end{table}

\subsection{Bootstrap critical values under the fixed null}\label{app:critical-values}

Figure~\ref{fig:zero.cv} reports the average lower and upper equal-tail bootstrap critical values for the fixed $\bH$ construction.
The two endpoints are reported separately because the bootstrap distribution need not be symmetric.
Under heterogeneous divergence rates, the critical values move toward the standard-normal values as the sample size increases.
Under equal divergence rates, the bootstrap distribution widens sharply as the eigengap narrows, particularly at the relative eigengap $0.10$.
This widening is economically important for the fixed-coordinate problem because the ordinary regression standard error does not by itself capture the additional variation associated with estimating the population PC direction.
The asymmetry of the two endpoints also supports using the two empirical tails rather than imposing a symmetric critical value.

\begin{figure}[!htb]
\centering
\begin{subfigure}[b]{0.32\textwidth}\centering
\includegraphics[width=\textwidth]{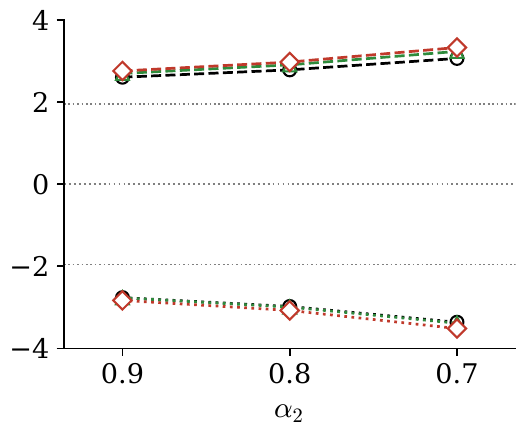}
\caption{Heterogeneous rates, $T=50$}
\end{subfigure}\hfill
\begin{subfigure}[b]{0.32\textwidth}\centering
\includegraphics[width=\textwidth]{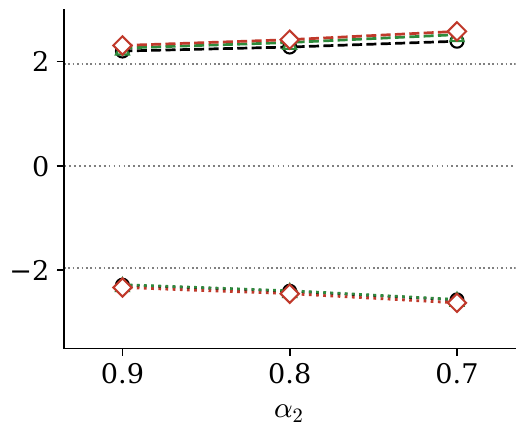}
\caption{Heterogeneous rates, $T=100$}
\end{subfigure}\hfill
\begin{subfigure}[b]{0.32\textwidth}\centering
\includegraphics[width=\textwidth]{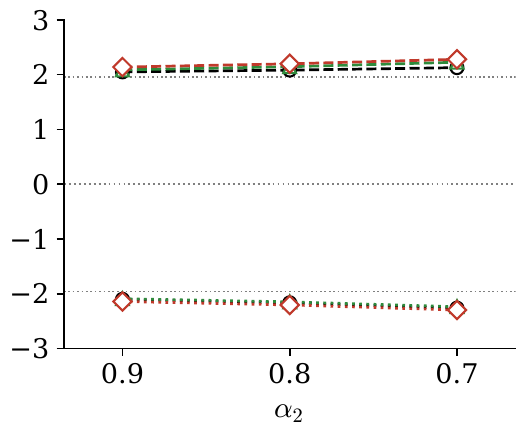}
\caption{Heterogeneous rates, $T=200$}
\end{subfigure}
\begin{subfigure}[b]{0.32\textwidth}\centering
\includegraphics[width=\textwidth]{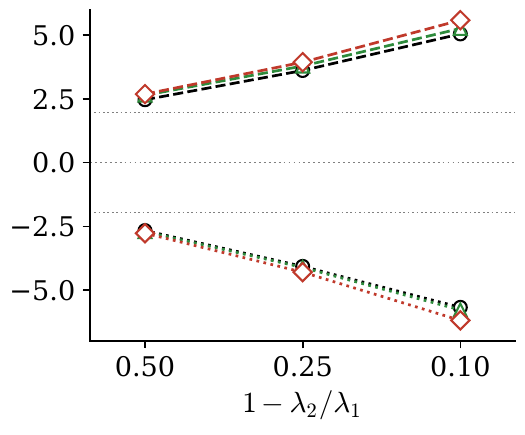}
\caption{Equal rates, $T=50$}
\end{subfigure}\hfill
\begin{subfigure}[b]{0.32\textwidth}\centering
\includegraphics[width=\textwidth]{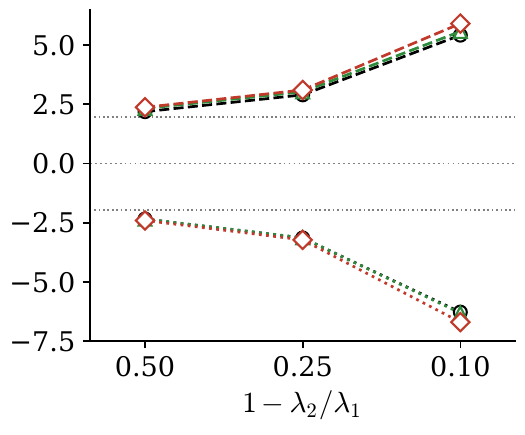}
\caption{Equal rates, $T=100$}
\end{subfigure}\hfill
\begin{subfigure}[b]{0.32\textwidth}\centering
\includegraphics[width=\textwidth]{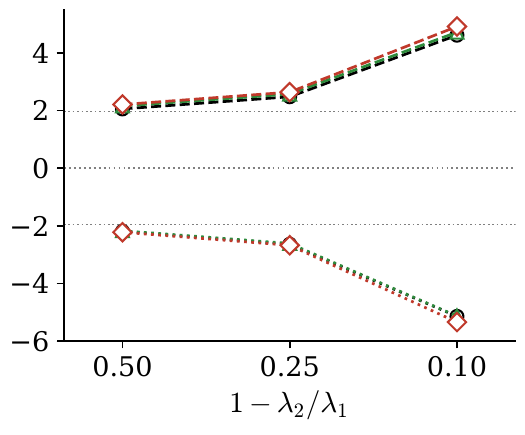}
\caption{Equal rates, $T=200$}
\end{subfigure}
\includegraphics[width=0.95\textwidth]{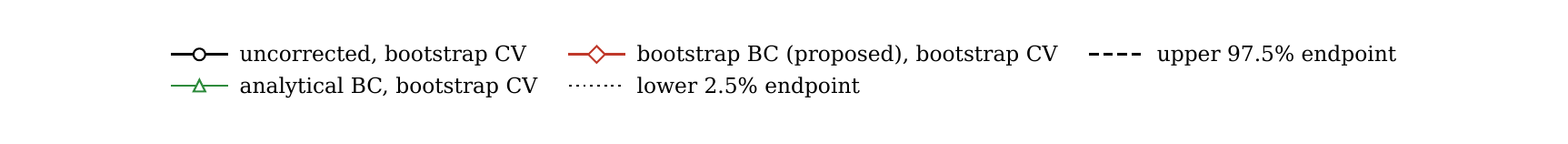}
\caption{Bootstrap critical-value endpoints under $H_0:\gamma_2^0=0$}
\label{fig:zero.cv}
\vspace{2pt}
\begin{minipage}{0.95\textwidth}
\footnotesize
\justifying
\noindent Notes: Lower and upper curves are the Monte Carlo averages of the empirical 2.5th and 97.5th percentiles of the corresponding recentered studentized bootstrap distribution.
The horizontal reference values are $-1.96$ and $1.96$.
The top row corresponds to heterogeneous divergence rates and the bottom row to equal divergence rates.
\end{minipage}
\end{figure}

\subsection{Rotation diagnostics under close eigengaps}\label{app:mechanism}

We next return to the distributional-approximation design to examine the mechanism behind the deterioration of fixed-coordinate inference under small eigengaps. 
Figure~\ref{fig:rotation.diag} provides further evidence on this mechanism.

The top row reports the frequency with which the absolute realized rotation gap exceeds $1.96$ times the ordinary standard error and the frequency of bootstrap label switching.
The bottom row reports the median $\widehat{\mathrm{RU}}_2$ from the distributional Monte Carlo design.
As the eigengap narrows, the rotation gap becomes large relative to conventional coefficient uncertainty.
At the relative eigengap $0.10$, the rotation gap alone exceeds $1.96$ standard errors in approximately 52\%, 45\%, and 34\% of replications at $T=50$, $100$, and $200$, respectively.
Bootstrap label switching, in contrast, falls from about 16\% at $T=50$ to essentially zero at $T=200$.
The persistence of the rotation-gap diagnostic after label switching has disappeared indicates that the problem is not driven simply by discrete relabelling of the PCs but by continuous uncertainty in the estimated direction within the nearly degenerate eigenspace.

\begin{figure}[!htb]
\centering
\begin{subfigure}[b]{0.32\textwidth}\centering
\includegraphics[width=\textwidth]{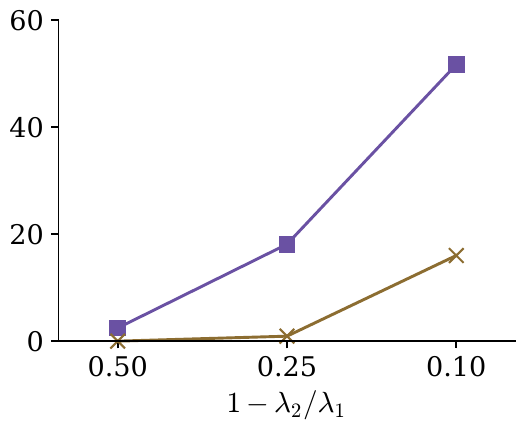}
\caption{$T=50$}
\end{subfigure}\hfill
\begin{subfigure}[b]{0.32\textwidth}\centering
\includegraphics[width=\textwidth]{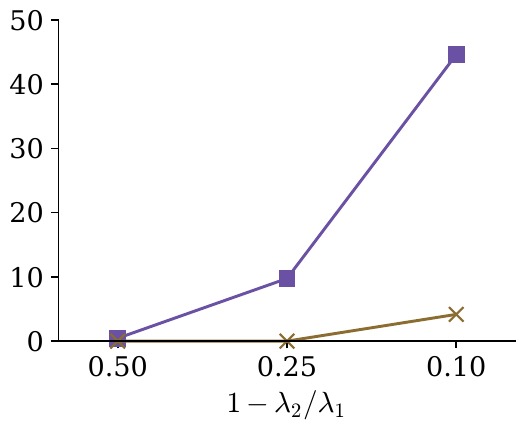}
\caption{$T=100$}
\end{subfigure}\hfill
\begin{subfigure}[b]{0.32\textwidth}\centering
\includegraphics[width=\textwidth]{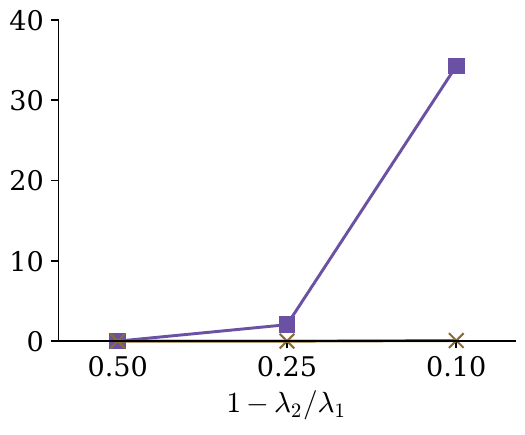}
\caption{$T=200$}
\end{subfigure}
\begin{subfigure}[b]{0.32\textwidth}\centering
\includegraphics[width=\textwidth]{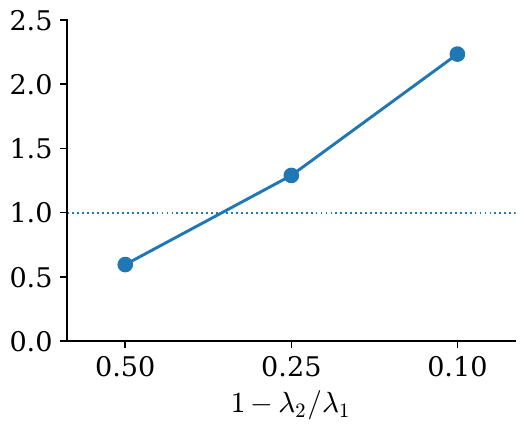}
\caption{$\widehat{\mathrm{RU}}_2$, $T=50$}
\end{subfigure}\hfill
\begin{subfigure}[b]{0.32\textwidth}\centering
\includegraphics[width=\textwidth]{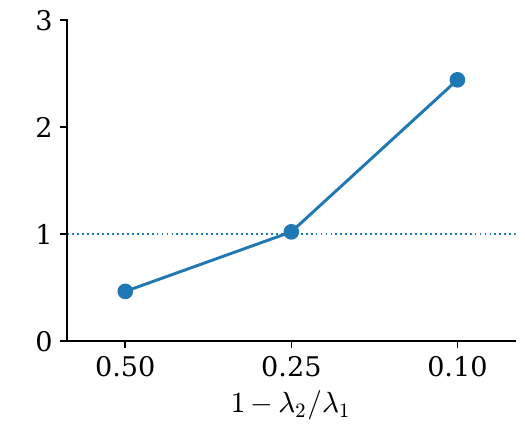}
\caption{$\widehat{\mathrm{RU}}_2$, $T=100$}
\end{subfigure}\hfill
\begin{subfigure}[b]{0.32\textwidth}\centering
\includegraphics[width=\textwidth]{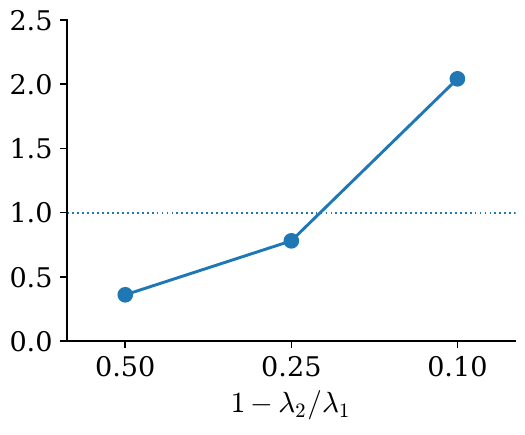}
\caption{$\widehat{\mathrm{RU}}_2$, $T=200$}
\end{subfigure}
\caption{Rotation diagnostics under close eigengaps}
\label{fig:rotation.diag}
\vspace{2pt}
\begin{minipage}{0.95\textwidth}
\footnotesize
\justifying
\noindent Notes: In the top row, square markers report the percentage of Monte Carlo replications for which the absolute rotation gap exceeds $1.96$ times the estimated standard error, while cross markers report first-level bootstrap label-switch frequencies.
The bottom row reports the median $\widehat{\mathrm{RU}}_2$.
These panels use the distributional-approximation DGP with $\bgamma^0=(1,1)'$.
\end{minipage}
\end{figure}

\subsection{RMSE under close eigengaps}\label{app:rmse}

Figure~\ref{fig:rmse.appendix} reports RMSEs for the second factor coefficient at $T=N=100$.
Relative to the data-dependent targets, bias correction reduces RMSE for the conventional rotation and has little effect for $\hat{\bH}_q$, whose uncorrected estimator already has small bias.
For the fixed population target, the RMSE rises sharply as the eigengap narrows.
At the relative eigengap $0.10$, the uncorrected RMSE is approximately $25.4$ when multiplied by 100 relative to $\gamma_2^0$, compared with $8.6$ relative to $\gamma_{\hat{\bH}_q,2}$.
This difference reflects the dispersion of the rotation gap and therefore cannot be removed by a bias correction.
Relative to $\gamma_{\hat{\bH},2}$, the proposed and GP bootstrap corrections have very similar RMSE, while the proposed procedure generally has the smaller residual bias.

\begin{figure}[!htb]
\centering
\begin{subfigure}[b]{0.32\textwidth}\centering
\includegraphics[width=\textwidth]{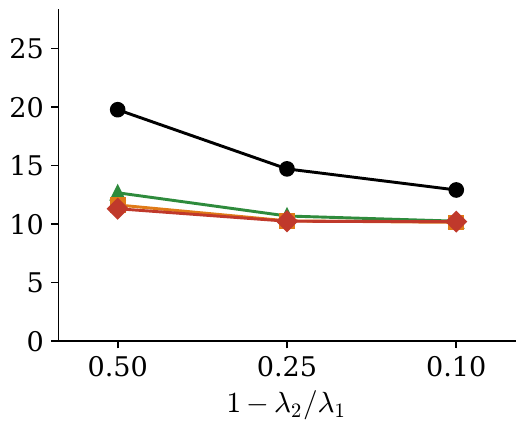}
\caption{$\gamma_{\hat{\bH},2}$}
\end{subfigure}\hfill
\begin{subfigure}[b]{0.32\textwidth}\centering
\includegraphics[width=\textwidth]{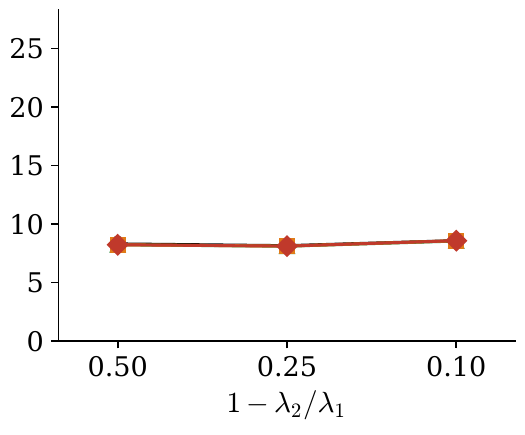}
\caption{$\gamma_{\hat{\bH}_q,2}$}
\end{subfigure}\hfill
\begin{subfigure}[b]{0.32\textwidth}\centering
\includegraphics[width=\textwidth]{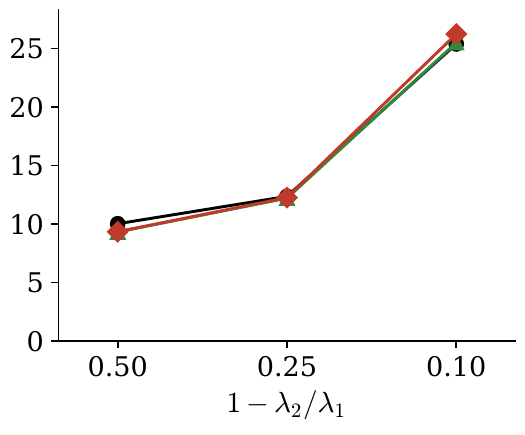}
\caption{$\gamma_2^0$}
\end{subfigure}
\includegraphics[width=0.85\textwidth]{figs/legend_estimators.pdf}
\caption{RMSE under close eigengaps}
\label{fig:rmse.appendix}
\vspace{2pt}
\begin{minipage}{0.95\textwidth}
\footnotesize
\justifying
\noindent Notes: RMSE is multiplied by 100.
The horizontal-axis values correspond to relative eigengaps $0.50$, $0.25$,
and $0.10$.
All panels use $T=N=100$.
\end{minipage}
\end{figure}

\subsection{Validation of the dependent bootstrap scheme}\label{app:depboot}

The empirical application of Section~\ref{sec:empirics} resamples the idiosyncratic errors by POET-whitened block bootstrap and the regression errors by blocks in time, with prewhitened HAC studentization. 
This subsection validates that scheme in a simulation mirroring its intended conditions. 
The design has $r=2$ factors with $(\alpha_1,\alpha_2)=(0.9,0.8)$ and $d_1=d_2=1$. 
The idiosyncratic errors follow a first-order autoregressive process with autoregressive coefficient $0.2$, Toeplitz cross-sectional correlation $0.5^{|i-j|}$, heteroskedastic scales, and standardized $\chi^2(4)$ innovations. 
The regression errors follow a first-order autoregressive process with autoregressive coefficient $0.2$ and standardized $t(5)$ innovations, with $\rho_{fw}=0.6$. 
Estimation and resampling follow the empirical implementation exactly, with the POET threshold selected by cross-validation, block length $\ell=\lceil T^{1/3}\rceil$, $B=299$, $B_2=19$, and HAC studentization at all levels. 
Bootstrap critical values are computed using the empirical order-statistic quantiles employed in the main simulations. 
The experiment uses 1{,}000 Monte Carlo replications.

\begin{table}[!htb]
\centering
\caption{Size under the dependent bootstrap scheme}
\label{tab:depboot}
{\small
\begin{tabular}{lccccc}
\hline\hline
 & uncorrected, & uncorrected, & analytical BC, & bootstrap BC, & bootstrap BC, \\
 & normal & bootstrap & normal & normal & bootstrap \\
\hline
$T=N=100$ & 10.9 & 6.2 & 10.9 & 11.0 & 5.7 \\
$T=N=200$ & 7.5 & 6.2 & 7.6 & 7.8 & 5.7 \\
\hline\hline
\end{tabular}}
\\[3pt]
\begin{minipage}{0.95\textwidth}
\footnotesize
\justifying
\noindent Notes: Entries are empirical rejection frequencies in percent for nominal 5\% tests of $\gamma_2^0=0$. 
Monte Carlo standard errors are at most $1.1$ percentage points. 
The median feasible diagnostic $\widehat{\mathrm{RU}}_2$ is $0.43$ at $T=100$ and $0.28$ at $T=200$.
\end{minipage}
\end{table}

Table~\ref{tab:depboot} provides supportive evidence for the dependent resampling implementation used in the empirical application. 
Normal critical values produce noticeable size distortions: the uncorrected test rejects in 10.9\% and 7.5\% of replications at $T=N=100$ and $200$, respectively, and the corresponding rejection frequencies are 11.0\% and 7.8\% after bootstrap bias correction. 
In contrast, bootstrap critical values reduce the rejection frequencies for the uncorrected estimator to 6.2\% at both sample sizes and for the bootstrap bias-corrected estimator to 5.7\% at both sample sizes. 
Thus, reproducing the sampling distribution of the statistic remains important when both serial and cross-sectional dependence are present.

The diagnostic $\widehat{\mathrm{RU}}_2$ also behaves in the expected direction under the dependent scheme. 
Its median falls from $0.43$ at $T=100$ to $0.28$ at $T=200$, consistent with declining rotation uncertainty as the sample size increases. 
These results are intended as a finite-sample validation of the implementation rather than as a separate theoretical result for the particular dependent resampling scheme.

The whitening-and-recoloring step is sensitive to the conditioning of the covariance estimate. 
With a deliberately under-thresholded POET estimator, the bootstrap bias estimate becomes highly variable and the correction performs poorly. 
We therefore use the cross-validated threshold in this experiment and in the empirical application.

\section{Additional empirical results}\label{app:emp}

This appendix reports the two-year-horizon results and a block-length sensitivity analysis for the empirical application of Section~\ref{sec:empirics}. 
For the two-year-horizon results, the specification, estimation, and bootstrap implementation follow the main text, with $B=999$, $B_2=49$, POET-whitened block resampling with $\ell=h=24$, and prewhitened HAC studentization. 

\subsection{Two-year horizon (\texorpdfstring{$h=24$}{h=24})}\label{app:emp24}

\begin{table}[!htb]
\centering
\caption{Estimation and inference for $\bdelta^0$: two-year-ahead forecasts of four-year excess bond returns ($h=24$)}
\label{tab:emp24}
{\small
\setlength{\tabcolsep}{4pt}
\begin{tabular}{lcccccc}
\hline\hline
 & $\hat f_1$ & $\hat f_2$ & $\hat f_3$ & $\hat f_4$ & $\hat f_5$ & CP \\
\hline
\multicolumn{7}{c}{1976M6--2004M1\quad ($T=313$;\ rel.\ gaps $0.64,\ 0.37,\ 0.29,\ 0.04$)} \\[2pt]
LS estimate                              & 0.56 & 0.71 & 0.16 & 0.11 & 0.54 & 1.81 \\
\hspace{1.5em}$p$-value: normal          & [0.45] & [0.32] & [0.79] & [0.55] & [0.37] & [0.13] \\
\hspace{1.5em}$p$-value: bootstrap       & [0.51] & [0.41] & [0.88] & [0.70] & [0.41] & [0.24] \\[2pt]
Analytically bias-corrected estimate     & 0.56 & 0.72 & 0.11 & -0.18 & 0.67 & 1.81 \\
\hspace{1.5em}$p$-value: normal          & [0.45] & [0.32] & [0.85] & [0.35] & [0.27] & [0.13] \\
\hspace{1.5em}$p$-value: bootstrap       & [0.51] & [0.42] & [0.89] & [0.60] & [0.37] & [0.24] \\[2pt]
Bootstrap bias-corrected estimate        & 0.57 & 0.74 & 0.15 & 0.08 & 0.63 & 1.83 \\
\hspace{1.5em}$p$-value: normal          & [0.45] & [0.30] & [0.81] & [0.66] & [0.30] & [0.13] \\
\hspace{1.5em}$p$-value: bootstrap       & [0.51] & [0.43] & [0.87] & [0.79] & [0.44] & [0.25] \\
$\widehat{\mathrm{RU}}_j$                & 0.01 & 0.02 & 0.09 & 1.05 & 0.10 & -- \\
\hline
\multicolumn{7}{c}{2004M2--2024M10\quad ($T=237$;\ rel.\ gaps $0.66,\ 0.59,\ 0.37,\ 0.10$)} \\[2pt]
LS estimate                              & 0.31 & 1.23 & 0.28 & 0.04 & 0.27 & 0.31 \\
\hspace{1.5em}$p$-value: normal          & [0.23] & [0.31] & [0.59] & [0.90] & [0.17] & [0.72] \\
\hspace{1.5em}$p$-value: bootstrap       & [0.42] & [0.52] & [0.70] & [0.87] & [0.17] & [0.83] \\[2pt]
Analytically bias-corrected estimate     & 0.31 & 1.24 & 0.27 & 0.07 & 0.29 & 0.31 \\
\hspace{1.5em}$p$-value: normal          & [0.22] & [0.31] & [0.60] & [0.83] & [0.15] & [0.72] \\
\hspace{1.5em}$p$-value: bootstrap       & [0.42] & [0.52] & [0.71] & [0.83] & [0.18] & [0.83] \\[2pt]
Bootstrap bias-corrected estimate        & 0.31 & 1.25 & 0.28 & 0.05 & 0.30 & 0.29 \\
\hspace{1.5em}$p$-value: normal          & [0.23] & [0.30] & [0.59] & [0.86] & [0.13] & [0.74] \\
\hspace{1.5em}$p$-value: bootstrap       & [0.42] & [0.53] & [0.71] & [0.86] & [0.20] & [0.83] \\
$\widehat{\mathrm{RU}}_j$                & 0.04 & 0.01 & 0.05 & 0.27 & 0.14 & -- \\
\hline\hline
\end{tabular}
}

\begin{minipage}{0.98\textwidth}
\footnotesize
\noindent Notes: Point estimates and two-sided $p$-values are reported for the fixed population parameter $\bdelta^0$.
The signs of the estimated principal components are normalized so that the corresponding uncorrected LS factor coefficients are positive, and this common orientation is retained for all bias-corrected estimators.
Normal $p$-values are computed from the standard normal distribution, while bootstrap $p$-values are obtained from the corresponding recentered bootstrap studentized statistic using the two empirical tails. 
For the factor-model errors, the bootstrap uses POET whitening, block resampling, and re-coloring, while the regression errors are block-resampled using the same block length $\ell=h$ at both levels.
Standard errors are prewhitened Newey--West with bandwidth at least $h-1$.
Rel.\ gaps are $1-\hat\lambda_{k+1}/\hat\lambda_k$, $k=1,\dots,4$, of the macro panel.
$\widehat{\mathrm{RU}}_j$ is the rotation-uncertainty diagnostic defined in \eqref{eq:Shat}.
\end{minipage}
\end{table}

No coefficient is significant at conventional levels under bootstrap inference in either subsample: the predictability documented at $h=12$ does not extend to the two-year horizon, consistent with the limited effective information in two-year overlapping returns.
The methodological patterns of the main text nonetheless replicate.
The near-degenerate fourth--fifth eigenvalue pair is flagged again, with $\widehat{\mathrm{RU}}_4=1.05$ compared with at most $0.10$ for the other factor coefficients in that subsample; its smaller value than at $h=12$ partly reflects the larger HAC standard error in the denominator at the longer horizon.
For the fourth factor, the analytical correction moves the normal $p$-value from $0.55$ to $0.35$, whereas the bootstrap bias-corrected estimator gives a normal $p$-value of $0.66$ and a bootstrap $p$-value of $0.79$.
Thus, the analytical correction moves this coefficient towards significance, while bootstrap inference continues to place it comfortably away from conventional significance levels.

\subsection{Block-length sensitivity}\label{app:empblock}

The main results set $\ell=h$, so that the block length spans the MA($h-1$) dependence induced by $h$-period overlapping returns.
Table~\ref{tab:empblock} reports the bootstrap $p$-values for the LS, analytically bias-corrected, and bootstrap bias-corrected estimators at $h=12$ for $\ell\in\{4,7,12,24\}$.
For both subsamples, $\ell=7$ corresponds to the $\lceil T^{1/3}\rceil$ rule of thumb, $\ell=12$ is the main specification, and $\ell=24$ provides a longer-block comparison.

\begin{table}[!htb]
\centering
\caption{Block-length sensitivity at $h=12$}
\label{tab:empblock}
{\small
\setlength{\tabcolsep}{4pt}
\begin{tabular}{llcccccc}
\hline\hline
 & & $\hat f_1$ & $\hat f_2$ & $\hat f_3$ & $\hat f_4$ & $\hat f_5$ & CP \\
\hline
\multicolumn{8}{c}{1976M6--2004M1\quad ($T=313$)} \\[2pt]
$\ell=4$ & LS & [0.23] & [0.51] & [0.75] & [0.61] & [0.08] & [0.00] \\
 & Analytical BC & [0.24] & [0.52] & [0.73] & [0.65] & [0.08] & [0.00] \\
 & Bootstrap BC & [0.24] & [0.52] & [0.75] & [0.70] & [0.12] & [0.00] \\[2pt]
$\ell=7$ & LS & [0.28] & [0.50] & [0.78] & [0.59] & [0.09] & [0.00] \\
 & Analytical BC & [0.28] & [0.51] & [0.76] & [0.62] & [0.09] & [0.00] \\
 & Bootstrap BC & [0.29] & [0.52] & [0.76] & [0.68] & [0.15] & [0.00] \\[2pt]
$\ell=12$ & LS & [0.32] & [0.58] & [0.87] & [0.52] & [0.12] & [0.00] \\
 & Analytical BC & [0.32] & [0.59] & [0.82] & [0.62] & [0.12] & [0.00] \\
 & Bootstrap BC & [0.34] & [0.61] & [0.87] & [0.62] & [0.16] & [0.00] \\[2pt]
$\ell=24$ & LS & [0.30] & [0.52] & [0.81] & [0.48] & [0.08] & [0.01] \\
 & Analytical BC & [0.30] & [0.55] & [0.77] & [0.64] & [0.09] & [0.01] \\
 & Bootstrap BC & [0.30] & [0.55] & [0.80] & [0.52] & [0.11] & [0.01] \\[2pt]
\hline
\multicolumn{8}{c}{2004M2--2024M10\quad ($T=249$)} \\[2pt]
$\ell=4$ & LS & [0.38] & [0.04] & [0.88] & [0.74] & [0.20] & [0.79] \\
 & Analytical BC & [0.38] & [0.04] & [0.88] & [0.79] & [0.21] & [0.79] \\
 & Bootstrap BC & [0.38] & [0.04] & [0.88] & [0.77] & [0.22] & [0.78] \\[2pt]
$\ell=7$ & LS & [0.34] & [0.05] & [0.86] & [0.77] & [0.19] & [0.82] \\
 & Analytical BC & [0.35] & [0.05] & [0.86] & [0.78] & [0.20] & [0.82] \\
 & Bootstrap BC & [0.36] & [0.05] & [0.82] & [0.76] & [0.21] & [0.82] \\[2pt]
$\ell=12$ & LS & [0.23] & [0.06] & [0.90] & [0.86] & [0.19] & [0.75] \\
 & Analytical BC & [0.23] & [0.06] & [0.89] & [0.86] & [0.20] & [0.76] \\
 & Bootstrap BC & [0.24] & [0.07] & [0.90] & [0.88] & [0.22] & [0.74] \\[2pt]
$\ell=24$ & LS & [0.24] & [0.09] & [0.88] & [0.79] & [0.16] & [0.82] \\
 & Analytical BC & [0.24] & [0.09] & [0.87] & [0.81] & [0.16] & [0.81] \\
 & Bootstrap BC & [0.26] & [0.10] & [0.86] & [0.79] & [0.19] & [0.82] \\[2pt]
\hline\hline
\end{tabular}
}

\begin{minipage}{0.98\textwidth}
\footnotesize
\noindent Notes: Entries are two-sided bootstrap $p$-values obtained from the corresponding recentered bootstrap studentized statistic using the two empirical tails.
The signs of the estimated principal components are normalized so that the corresponding uncorrected LS factor coefficients are positive, and this common orientation is retained for all bias-corrected estimators.
The factor-model errors are POET-whitened, block-resampled, and re-colored, while the regression errors are block-resampled using the same block length $\ell$.
The main specification is $\ell=h=12$; $\ell=7$ corresponds to the $\lceil T^{1/3}\rceil$ rule, while $\ell=4$ and $\ell=24$ provide shorter- and longer-block comparisons.
All other elements of the implementation are as in Table~\ref{tab:emp12}.
\end{minipage}
\end{table}

Two features stand out.
First, inference for the second post-2004 macro factor is sensitive to the block length.
For the bootstrap bias-corrected estimator, its reported bootstrap $p$-value increases from $0.04$ at $\ell=4$ to $0.05$ at $\ell=7$, $0.07$ at the main specification $\ell=12$, and $0.10$ at $\ell=24$.
The same pattern is present for the LS and analytically bias-corrected estimators.
Thus, the shortest block length leads to rejection at the 5\% level, whereas the result at $\ell=7$ is borderline at the displayed precision and the longer block lengths $\ell\in\{12,24\}$ do not reject.
Allowing the blocks to span the full MA($11$) dependence therefore produces less aggressive inference than the shortest block choice.

Second, the remaining conclusions are comparatively stable across block lengths.
The CP factor remains strongly significant in the pre-2004 subsample under every estimator and block length, whereas none of the macro factors is significant at the 5\% level under bootstrap inference.
The fifth pre-2004 factor is the closest case, but its bootstrap $p$-values remain above conventional significance levels across the block-length choices.
Overall, the sensitivity exercise supports the main choice $\ell=h$: the principal empirical conclusions are unchanged, while the post-2004 second-factor result shows that using blocks shorter than the dependence horizon can lead to materially more aggressive inference.

\end{document}